\documentclass[journal]{IEEEtran}

\usepackage{newpxtext,newpxmath}

\let\coloneqq\relax

\usepackage[latin1]{inputenc}
\usepackage{amsthm}
\usepackage{amssymb}
\usepackage{amsmath}
\usepackage{bbold}
\usepackage{bbm}
\usepackage{nonfloat}
\usepackage[pdftex]{hyperref}
\usepackage{braket}
\usepackage{dsfont}
\usepackage{mathdots}
\usepackage{mathtools}
\usepackage{enumerate}
\usepackage[shortlabels]{enumitem}
\usepackage{csquotes}
\usepackage{stmaryrd}
\usepackage[cal=boondox]{mathalfa}
\usepackage{graphicx}
\usepackage{stackengine}
\usepackage{scalerel}
\usepackage{xr}
\usepackage{array}
\usepackage{makecell}
\newcolumntype{x}[1]{>{\centering\arraybackslash}p{#1}}
\usepackage{tikz}
\usepackage{pgfplots}
\usetikzlibrary{shapes.geometric, shapes.misc, positioning, arrows, arrows.meta, decorations.pathreplacing, decorations.pathmorphing, patterns, angles, quotes, calc}
\usepackage{booktabs}
\usepackage{xfrac}
\usepackage{siunitx}
\usepackage{centernot}
\usepackage{comment}
\usepackage{chngcntr}

\newtheorem{thm}{Theorem}
\newtheorem*{thm*}{Theorem}

\newtheorem*{prop*}{Proposition}
\newtheorem{lemma}[thm]{Lemma}
\newtheorem*{lemma*}{Lemma}

\newtheorem*{cor*}{Corollary}

\newtheorem*{cj*}{Conjecture}
\newtheorem{Def}[thm]{Definition}
\newtheorem*{Def*}{Definition}

\newtheorem{remark}{Remark}

\makeatletter
\def\thmhead@plain#1#2#3{%
  \thmname{#1}\thmnumber{\@ifnotempty{#1}{ }\@upn{#2}}%
  \thmnote{ {\the\thm@notefont#3}}}
\let\thmhead\thmhead@plain
\makeatother

\theoremstyle{definition}

\newcommand{\bb}{\begin{equation}\begin{aligned}\hspace{0pt}}
\newcommand{\bbb}{\begin{equation*}\begin{aligned}}
\newcommand{\ee}{\end{aligned}\end{equation}}
\newcommand{\eee}{\end{aligned}\end{equation*}}
\newcommand*{\coloneqq}{\mathrel{\vcenter{\baselineskip0.5ex \lineskiplimit0pt \hbox{\scriptsize.}\hbox{\scriptsize.}}} =}

\newcommand\floor[1]{\left\lfloor#1\right\rfloor}
\newcommand\ceil[1]{\left\lceil#1\right\rceil}

\newcommand{\eqt}[1]{\stackrel{\mathclap{\scriptsize \mbox{#1}}}{=}}
\newcommand{\leqt}[1]{\stackrel{\mathclap{\scriptsize \mbox{#1}}}{\leq}}

\newcommand{\geqt}[1]{\stackrel{\mathclap{\scriptsize \mbox{#1}}}{\geq}}
\newcommand{\ketbra}[1]{\ket{#1}\!\!\bra{#1}}

\newcommand{\R}{\mathds{R}}
\newcommand{\N}{\mathds{N}}
\newcommand{\C}{\mathds{C}}

\DeclareMathOperator{\Tr}{Tr}

\DeclareMathAlphabet{\pazocal}{OMS}{zplm}{m}{n}

\newcommand{\HH}{\pazocal{H}}

\newcommand{\lsmatrix}{\left(\begin{smallmatrix}}
\newcommand{\rsmatrix}{\end{smallmatrix}\right)}

\stackMath

\stackMath

\makeatletter
\newcommand*\rel@kern[1]{\kern#1\dimexpr\macc@kerna}
\newcommand*\widebar[1]{%
  \begingroup
  \def\mathaccent##1##2{%
    \rel@kern{0.8}%
    \overline{\rel@kern{-0.8}\macc@nucleus\rel@kern{0.2}}%
    \rel@kern{-0.2}%
  }%
  \macc@depth\@ne
  \let\math@bgroup\@empty \let\math@egroup\macc@set@skewchar
  \mathsurround\z@ \frozen@everymath{\mathgroup\macc@group\relax}%
  \macc@set@skewchar\relax
  \let\mathaccentV\macc@nested@a
  \macc@nested@a\relax111{#1}%
  \endgroup
}

\tikzset{meter/.append style={draw, inner sep=10, rectangle, font=\vphantom{A}, minimum width=30, line width=.8, path picture={\draw[black] ([shift={(.1,.3)}]path picture bounding box.south west) to[bend left=50] ([shift={(-.1,.3)}]path picture bounding box.south east);\draw[black,-latex] ([shift={(0,.1)}]path picture bounding box.south) -- ([shift={(.3,-.1)}]path picture bounding box.north);}}}
\tikzset{roundnode/.append style={circle, draw=black, fill=gray!20, thick, minimum size=10mm}}
\tikzset{squarenode/.style={rectangle, draw=black, fill=none, thick, minimum size=10mm}}

\definecolor{Blues5seq1}{RGB}{239,243,255}
\definecolor{Blues5seq2}{RGB}{189,215,231}
\definecolor{Blues5seq3}{RGB}{107,174,214}
\definecolor{Blues5seq4}{RGB}{49,130,189}
\definecolor{Blues5seq5}{RGB}{8,81,156}

\definecolor{Greens5seq1}{RGB}{237,248,233}
\definecolor{Greens5seq2}{RGB}{186,228,179}
\definecolor{Greens5seq3}{RGB}{116,196,118}
\definecolor{Greens5seq4}{RGB}{49,163,84}
\definecolor{Greens5seq5}{RGB}{0,109,44}

\definecolor{Reds5seq1}{RGB}{254,229,217}
\definecolor{Reds5seq2}{RGB}{252,174,145}
\definecolor{Reds5seq3}{RGB}{251,106,74}
\definecolor{Reds5seq4}{RGB}{222,45,38}
\definecolor{Reds5seq5}{RGB}{165,15,21}

\usepackage{pgfplots}

\newtheorem{definition}{Definition}

\pgfplotsset{width=10cm,compat=1.9}
\usetikzlibrary{decorations.markings}

\allowdisplaybreaks

\newcommand*{\addFileDependency}[1]{
  \typeout{(#1)}
  \@addtofilelist{#1}
  \IfFileExists{#1}{}{\typeout{No file #1.}}
}
\makeatother

\begin{document}

\title{Optical fibres with memory effects and their quantum communication capacities}

 \author{\IEEEauthorblockN{Francesco Anna Mele}
\IEEEauthorblockA{\textit{NEST, Scuola Normale Superiore and Istituto Nanoscienze, Consiglio Nazionale delle Ricerche, Piazza dei Cavalieri 7, IT-56126 Pisa, Italy} \\
francesco.mele@sns.it}\\
\and
\IEEEauthorblockN{  Giacomo De Palma}
\IEEEauthorblockA{\textit{Department of Mathematics, University of Bologna, Piazza di Porta San Donato 5, 40126 Bologna BO, Italy} \\
giacomo.depalma@unibo.it}\\
\and

\IEEEauthorblockN{Marco Fanizza}
\IEEEauthorblockA{\textit{{F\'{\i}sica Te\`{o}rica: Informaci\'{o} i Fen\`{o}mens Qu\`{a}ntics, Departament de F\'{i}sica, Universitat Aut\`{o}noma de Barcelona, ES-08193 Bellaterra (Barcelona), Spain}}\\
marco.fanizza@uab.cat}\\
\and
\IEEEauthorblockN{Vittorio Giovannetti}
\IEEEauthorblockA{\textit{NEST, Scuola Normale Superiore and Istituto Nanoscienze, Consiglio Nazionale delle Ricerche, Piazza dei Cavalieri 7, IT-56126 Pisa, Italy}\\
vittorio.giovannetti@sns.it}\\
\and
\IEEEauthorblockN{ Ludovico Lami}
\IEEEauthorblockA{\textit{QuSoft, Science Park 123, 1098 XG Amsterdam, the Netherlands} \\
\textit{Korteweg-de Vries Institute for Mathematics, University of Amsterdam, Science Park 105-107, 1098 XG Amsterdam, the Netherlands}\\
\textit{Institute for Theoretical Physics, University of Amsterdam, Science Park 904, 1098 XH Amsterdam, the Netherlands}\\
ludovico.lami@gmail.com}
}

\maketitle

\begin{abstract}
If the transmissivity of an optical fibre falls below a critical value, its use as a reliable quantum channel is known to be drastically compromised. However, if the memoryless assumption does not hold --- e.g.~when input signals are separated by a sufficiently short time interval --- the validity of this limitation is put into question. In this work we introduce a model of optical fibre that can describe memory effects for long transmission lines. We then solve its quantum capacity, two-way quantum capacity, and secret-key capacity exactly. By doing so, we show that --- due to the memory cross-talk between the transmitted signals --- reliable quantum communication is attainable even for highly noisy regimes where it was previously considered impossible.
\end{abstract}

\begin{IEEEkeywords}
Quantum Shannon Theory, Continuous Variable systems, Non-markovian noise, Memory effects in quantum communication, Long-distance quantum communication
\end{IEEEkeywords}

\section{Introduction}
The potential applications of a global quantum internet~\cite{quantum_internet_Wehner,Pirandola20} include secure communication~\cite{bennett1984quantum}, efficient entanglement and qubits distribution, enhanced quantum sensing capabilities~\cite{Sidhu}, distributed and blind quantum computing~\cite{Distributed_QC,secure_access_qinternet}, as well as groundbreaking experiments in fundamental physics~\cite{Sidhu}. These applications heavily rely on establishing long-distance quantum communication across optical fibres or free-space links. However, the vulnerability of optical signals to noise poses a significant obstacle to achieving this goal. To overcome this challenge, one possible 
solution is to exploit quantum repeaters~\cite{repeaters, Munro2015} along the communication line. However, current implementations of quantum repeaters remain in the realm of proof-of-principle experiments. In addition, quantum repeaters will likely impose substantial demands on technology resources, making them potentially expensive. 
Consequently, a pressing problem is to develop protocols that can operate without --- or with a modest number of --- quantum repeaters. Recently, in~\cite{Die-Hard-2-PRL,Die-Hard-2-PRA} a theoretical proof-of-principle solution to this problem that takes advantage of \emph{memory effects}~\cite{memory-review, Dynamical-Model} in optical fibres has been proposed. 

Memory effects arise when signals are fed into the optical fibre separated by a sufficiently short time interval~\cite{memory-review, Dynamical-Model,Banaszek-Memory,Ball-Memory}. In this scenario, the noise within the fibre is influenced by prior input signals, exhibiting a form of "memory." Consequently, the commonly assumed memoryless (iid) paradigm, which posits that noise acts uniformly and independently on each signal, becomes invalid.
The exploration of memory effects in optical fibres has been extensively addressed in~\cite{Memory1, Memory2, Memory3, Die-Hard-2-PRA, Die-Hard-2-PRL}. 
While these studies capture  fundamental aspects of the problem, especially in the context of "short" communication lines, a more comprehensive analysis reveals limitations when extending these findings to practical configurations.
Specifically, in scenarios involving "long" communication lines where consecutive signals continually interact throughout the entire length of the fibre, the existing models become inadequate, as they do not accommodate two essential requirements:
\begin{itemize}
    \item \emph{Property~1}: 
    No information can be transmitted across the fibre if its transmissivity is precisely zero;
    \item \emph{Property~2}: The model must remain consistent when optical fibres are composed. In other words, the combination of two models corresponding to fibres of lengths $L_1$ and $L_2$ should yield a model corresponding to a fibre of length $L_1+L_2$.
\end{itemize}
In this work, we overcome the limitations of~\cite{Memory1,Memory2,Memory3,Die-Hard-2-PRA,Die-Hard-2-PRL} by developing a new model of optical fibres with memory effects that accurately encapsulates the essential properties of long optical fibres, such as those outlined in Properties~1 and~2 above. 
Our main result is the calculation of the exact value of the quantum capacity $Q$, the two-way quantum capacity $Q_2$, and the secret-key capacity $K$~\cite{MARK, Sumeet_book} of optical fibres with memory effects in the absence of thermal noise. We recall that:
\begin{itemize}
    \item the \emph{quantum capacity}~\cite{MARK, Sumeet_book}, denoted as $Q$, is the maximum achievable rate of qubits that can be reliably transmitted through the channel in the limit of infinite channel uses. The rate of qubits is the ratio between the number of qubits transmitted and the number of uses of the channel.
    \item the \emph{two-way quantum capacity}~\cite[Chapters 14 and 15]{Sumeet_book}, denoted as $Q_2$, is the maximum achievable rate of qubits that can be reliably transmitted through the channel in the limit of infinite channel uses by assuming that the sender Alice and the receiver Bob have free access to a public, noiseless, two-way classical communication line. Note that $Q_2$ is also the proper capacity for entanglement distribution. Indeed, an ebit (i.e.~a two-qubit maximally entangled states) can teleport an arbitrary qubit~\cite{teleportation}, and the ability of transmitting an arbitrary qubit implies the ability of distributing an ebit.
    \item the \emph{secret-key capacity}~\cite[Chapters 14 and 15]{Sumeet_book}, denoted as $K$, is the maximum achievable rate of secret-key bits that can be reliably transmitted through the channel in the limit of infinite channel uses by assuming that the sender Alice and the receiver Bob have free access to a public, noiseless, two-way classical communication line. 
\end{itemize}
We refer to Section~\ref{sec_def_capacity} for rigorous definitions of capacities (see also~\cite{Sumeet_book,memory-review}). According to our analysis, for any arbitrarily low non-zero value of the transmissivity $\lambda\in(0,1]$ of the fibre and for any arbitrarily large value of its associated thermal noise $\nu\ge0$, there exists a non-zero time interval separating successive signals below which $Q$, $Q_2$, and $K$ all become strictly positive. Specifically, this shows that the presence of memory effects facilitates qubit distribution ($Q>0$) even when $\lambda\in(0,\frac{1}{2}]$, which is not achievable with memoryless optical fibres. Additionally, memory effects enable two-way entanglement distribution ($Q_2>0$) and quantum key distribution ($K>0$) even when $\lambda\in(0,\frac{\nu}{\nu+1}]$, surpassing the capabilities of memoryless optical fibres. Our model therefore supports in a theoretically consistent way the claim that memory effects can be exploited to achieve long-distance quantum communication even in the presence of substantial loss. Since loss is the main factor limiting current quantum communication technologies, the potential of memory effects should be explored thoroughly.

The structure of the paper is as follows. In Section~\ref{Sec_preliminaries}, we 
briefly review some preliminary notions necessary for stating our results. In Section~\ref{sec_local_interaction_model}, we review a previous model of optical fibres with memory effects, termed as \emph{localised interaction model} (LIM). In Section~\ref{Section_DIM}, we introduce a new model of optical fibre with memory effects, termed as \emph{delocalised interaction model} (DIM), which addresses the limitations observed in the localised interaction model. In Section~\ref{Sec_Main_results}, we outline and discuss our main results. In Section~\ref{sec_sm_notation}, we establish notation and introduce basic concepts crucial for proving our results. In Section~\ref{SM_sec_def_model}, we give a detailed examination of the delocalised interaction model. In Section~\ref{sec_solution_property_sm}, we precisely explain the limitations of the local interaction model and elucidates how they are resolved by the delocalised interaction model. In Section~\ref{section_toeplitz}, we review the \emph{Avram--Parter's theorem}~\cite{Avram1988,Parter1986} and, by leveraging it, we establish a matrix-analysis result that plays a crucial role in proving our main results. Finally, in Section~\ref{Sec_capacities_DIM}, we present the proof of our main results.

\section{Preliminaries}\label{Sec_preliminaries} The signals transmitted along an optical fibre can be described in terms of an ordered array of 
localised e.m.~pulses $S_1$, $S_2$, $\cdots$, $S_n$ that travel along the fibre separated by a fixed delay time $\delta t$, each of which associated with independent annihilation operators $a_1$, $a_2$, $\cdots$, $a_n$~\cite{Caves}. 
The memoryless regime is reached when $\delta t$ is sufficiently large to prevent cross-talking among the transmitted signals: accordingly they will experience the same type of noise which, under very general conditions, is typically identified  with a \emph{thermal attenuator} channel $\mathcal{E}_{\lambda,\nu}$. 
This is a continuous-variable~\cite{BUCCO} single-mode quantum channel characterised by two parameters: $\lambda\in [0,1]$, which represents the transmissivity of the fibre 
, and $\nu\in [0,\infty)$, which quantifies the thermal noise 
added by the environment. In the limit of zero temperature $\nu=0$, the transformation is conventionally referred to as the \emph{pure-loss channel}.
Mathematically, the action of $\mathcal{E}_{\lambda,\nu}$ on a generic input state $\rho$ of the $i$th signal can be expressed as a beam splitter mixing the latter with a dedicated local bosonic bath $E_i$, initialised in the thermal state $\tau_\nu$ with mean photon number $\nu$. 
In formula, this can be written as $\mathcal{E}_{\lambda,\nu}(\rho)\coloneqq\Tr_{E_i}\left[U_\lambda \big(\rho \otimes\tau_\nu \big) {U_\lambda}^\dagger\right]$, 
where $U_\lambda \coloneqq e^{\arccos(\sqrt{\lambda})(a_i^\dagger b_i-a_i\, b_i^\dagger)}$ is the beam splitter unitary, $a_i$ and $b_i$ are the annihilation operators of $S_i$ and $E_i$, and $\Tr_{E_i}$ represents the partial trace w.r.t.~$E_i$. 
The quantum capacity $Q$~\cite{holwer, LossyECEAC1, LossyECEAC2}, two-way quantum capacity $Q_2$~\cite{PLOB}, and secret-key capacity $K$~\cite{PLOB} of the pure-loss channel $\mathcal{E}_{\lambda,0}$ have been determined exactly. In contrast, only bounds are known for the capacities $Q$~\cite{PLOB, Rosati2018, Sharma2018, Noh2019,holwer, Noh2020,fanizza2021estimating}, $Q_2$~\cite{PLOB,Pirandola2009,Ottaviani_new_lower,lower-bound}, and $K$~\cite{PLOB,Pirandola2009,Ottaviani_new_lower,lower-bound} of the thermal attenuator $\mathcal{E}_{\lambda,\nu}$. In particular, it is known that the quantum capacity of the thermal attenuator $Q(\mathcal{E}_{\lambda,\nu})$ vanishes if the transmissivity of the fibre falls below the critical value of $\lambda\le\frac{1}{2}$, and, additionally, the converse holds for $\nu=0$. 
It is also known that the two-way quantum capacity $Q_2(\mathcal{E}_{\lambda,\nu})$ and the secret-key capacity $K(\mathcal{E}_{\lambda,\nu})$ of the thermal attenuator vanish if and only if the transmissivity falls below the critical value of $\lambda\le\frac{\nu}{\nu+1}$~\cite{lower-bound}.
Since typically the transmissivity of an optical fibre decreases exponentially with its length, under the memoryless assumption there are strong limitations on the distance at which it is possible to perform qubit distribution $(Q>0)$, two-way entanglement distribution $(Q_2>0)$, and quantum key distribution $(K>0)$ without relying on quantum repeaters.

\section{Localised Interaction Model}\label{sec_local_interaction_model}
Early attempts to incorporate memory effects into optical fibres were made in~\cite{Memory1, Memory2, Memory3} with a model which from now on we shall refer to as "Localised Interaction Model" (LIM). In these works, following the approach outlined in~\cite{Dynamical-Model}, intra-signal couplings are induced 
by an ordered sequence of collisional events in which each  
 transmitted pulse interacts  with a common environment 
$E$ (see upper panel of Fig.~\ref{fig:new_model_main}). The latter is characterised by a resetting mechanism that endeavours to restore $E$ to its initial configuration
over a thermalization time $t_E$. Consequently,
when the time delay $\delta t$ separating two successive input signals exceeds $t_E$, each pulse encounters the same environmental state, rendering the communication effectively memoryless.
Conversely, when $\delta t$ is smaller than or comparable to $t_E$, after colliding with one of the signals, the environment $E$ does not have sufficient time to revert to $\tau_\nu$ and functions as a mediator for pulse interactions.
LIM emulates this intricate dynamics of the $n$ input signals via the $n$-mode quantum channel $\Phi_{\lambda,\mu,\nu}^{(1,n)}$ illustrated in Fig.~\ref{fig:griglia_main}(b), where the parameter $\mu\in[0,1]$ serves as the "memory parameter" of the model. It quantifies what is the fraction of the energy lost by an input signal that can potentially be absorbed by the subsequent ones. 
Ranging from $\mu=0$ (where $\Phi_{\lambda,0,\nu}^{(1,n)}$ reduces to $n$fold memoryless channels ${\mathcal{E}}_{\lambda,\nu}^{\otimes n}$) to $\mu=1$ (full memory), $\mu$ effectively encapsulates the interplay between the time interval $\delta t$ separating two consecutive input signals and the thermalization time $t_E$. For instance, one plausible expression for $\mu$ might be $\mu= \exp(-\delta t/t_E)$.

{Mathematically, the channel $\Phi_{\lambda,\mu,\nu}^{(1,n)}$ acts on any $n$-mode input state $\rho_{S_1 \ldots S_n}$ as
\bb\label{def_Phi_1}
    \Phi_{\lambda,\mu,\nu}^{(1,n)}(\rho_{S_1\ldots S_n})&\coloneqq\Tr_{E}\big[ \mathcal{U}_{\lambda}^{S_n E}\circ \mathcal{E}^{(E)}_{\mu,\nu}\circ\mathcal{U}_{\lambda}^{S_{n-1} E}\circ \mathcal{U}_{\lambda}^{S_{n-1} E}\circ \mathcal{E}^{(E)}_{\mu,\nu}\\
    &\quad\circ\ldots \circ\mathcal{U}_{\lambda}^{S_1 E}\circ \mathcal{E}^{(E)}_{\mu,\nu}\left( \rho_{S_1\ldots S_n}\otimes \tau_\nu^{(E)} \right)\big]\,.
\ee
Here, for each $i=1,\ldots,n$, the channel $\mathcal{U}_{\lambda}^{S_{i} E}(\cdot)\coloneqq U_{\lambda}^{(S_iE)}\cdot (U_{\lambda}^{(S_iE)})^\dagger$ denotes the beam splitter unitary channel of transmissivity $\lambda$ acting on the $i$th signal $S_i$ and on the environment $E$; the channel $\mathcal{E}^{(E)}_{\mu,\nu}$ is the thermal attenuator of transmissivity $\mu$ and thermal noise $\nu$ acting on $E$; the state $\tau_\nu^{(E)}$ is the thermal state of mean photon number $\nu$ on $E$. In \eqref{def_Phi_1}, the thermal attenuator $\mathcal{E}^{(E)}_{\mu,\nu}$ models the thermalisation process affecting $E$, which endeavours to restore the state of $E$ to the thermal state $\tau_\nu^{(E)}$. Note that, in \eqref{def_Phi_1}, we could have omitted the rightmost thermal attenuator $\mathcal{E}^{(E)}_{\mu,\nu}$, as the latter leaves the thermal state unchanged, i.e.~$\mathcal{E}^{(E)}_{\mu,\nu}(\tau_\nu^{(E)})=\tau_\nu^{(E)}$. Hence, by writing each thermal attenuator $\mathcal{E}^{(E)}_{\mu,\nu}$ in terms of beam splitter and thermal state, one obtains the interferometric representation of $\Phi_{\lambda,\mu,\nu}^{(1,n)}$ depicted in Fig.~\ref{fig:griglia_main}(b). The physical interpretation of the mathematical definition of $\Phi_{\lambda,\mu,\nu}^{(1,n)}$ in \eqref{def_Phi_1} is as follows: the signal $S_1$ interacts with the environment $E$, which is initialised in the thermal state $\tau_\nu^{(E)}$, via the beam splitter evolution $\mathcal{U}_{\lambda}^{S_{1} E}$; right after such an interaction, $E$ is affected by the thermalisation process $\mathcal{E}^{(E)}_{\mu,\nu}$; the signal $S_2$ interacts with $E$ via the beam splitter evolution $\mathcal{U}_{\lambda}^{S_{2} E}$; right after such an interaction, the thermalisation process $\mathcal{E}^{(E)}_{\mu,\nu}$ acts once again, and so and so forth for each signal $S_i$ with $i=3,\ldots n$. After all these interactions, the final state of the signals $S_1\ldots S_n$ (tracing out the environment $E$) is precisely the output of the channel $\Phi_{\lambda,\mu,\nu}^{(1,n)}$.
}
\begin{figure}[t]
	\centering
	\includegraphics[width=1\linewidth]{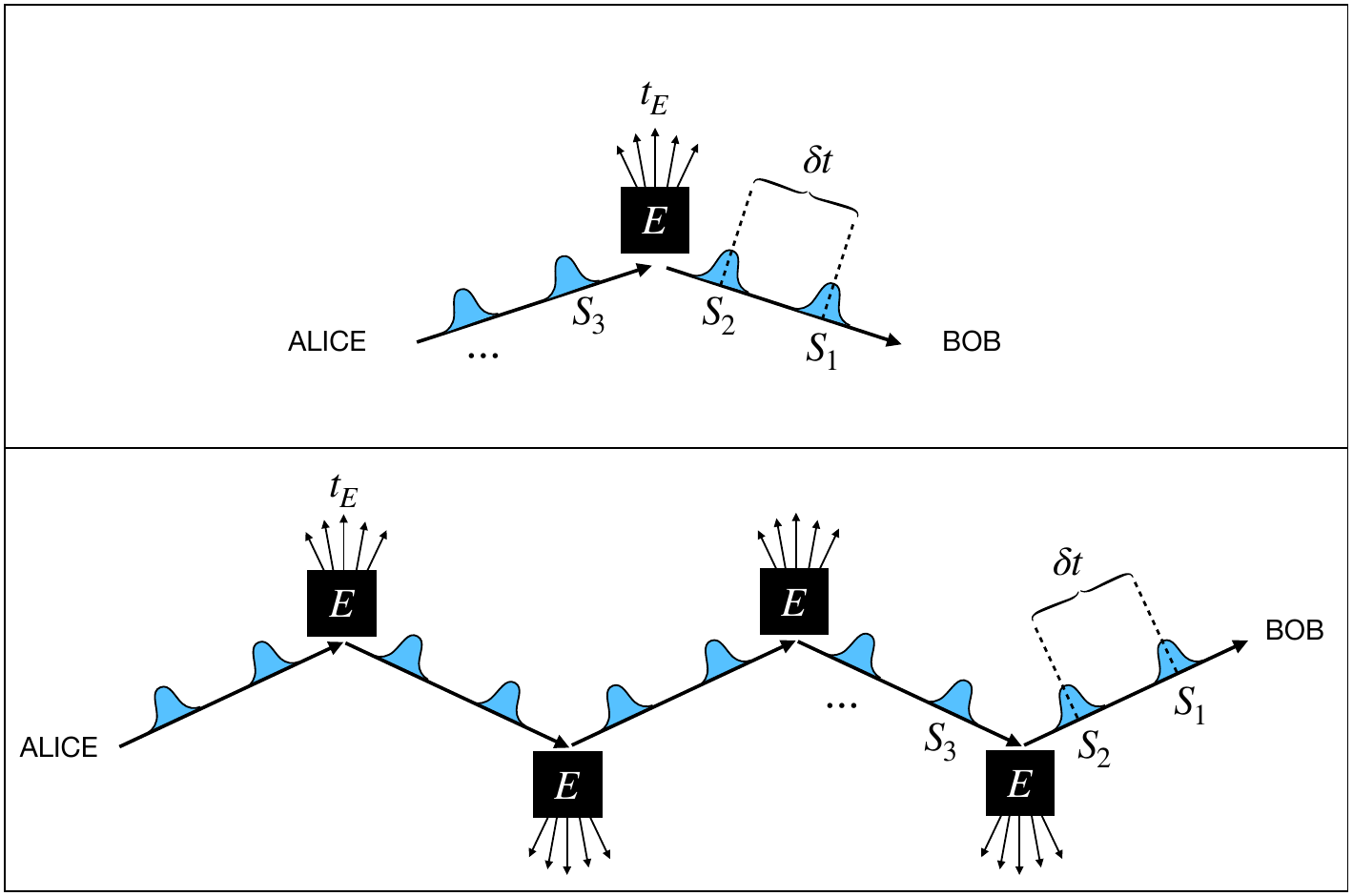}
	\caption{Pictorial representation of the mechanism which is responsible for the intra-signal interactions between a sequence of pulses $S_1$, $S_2$, $\cdots$ that travel along an optical fibre. In the LIM~\cite{Memory1, Memory2, Memory3} such couplings occurs in a single location, through the mediation of a single  common environment (upper panel). A more realistic description of the effect would instead allow for multiple cross-talks events distributed over the entire length of the fibre (bottom panel). 
	}
    \label{fig:new_model_main}
\end{figure}

\section{Delocalised Interaction Model}\label{Section_DIM} Here, we present an improved version of LIM, which we dub "Delocalised Interaction Model" (DIM). This model can be used to describe memory effects in long optical fibres, where transmitted signals have a chance of experiencing multiple interactions along the entire length of the communication line. 
The idea behind the DIM is relatively simple: a spatially homogeneous optical fibre of finite length $L$ is seen as the composition of $M$ identical optical fibres of length $L/M$.  
In the limit of large $M$, such fibres are sufficiently short that they can be effectively described via LIM mappings. The resulting input-output relation of the global fibre is hence computed by first properly concatenating such individual terms and then taking the continuum limit $M\to\infty$. 
Specifically, since the transmissivity $\lambda$ is exponentially decreasing in the fibre length, each of the $M$ components of the fibre has transmissivity $\lambda^{1/M}$. Each component is thus associated with the LIM map $\Phi_{\lambda^{1/M},\mu,\nu}^{(1,n)}$, characterised by the same local temperature parameter $\nu$ and by the same memory parameter~$\mu$ of the global fibre. 
The resulting input-output DIM map is hence provided by the $M$-fold concatenation 
\begin{eqnarray} \label{concatenated} 
\Phi_{\lambda,\mu,\nu}^{(M,n)} \coloneqq \left(  \Phi_{\lambda^{1/M},\mu,\nu}^{(1,n)}\right)^M = \underbrace{ \Phi_{\lambda^{1/M},\mu,\nu}^{(1,n)}  \circ \cdots \circ \Phi_{\lambda^{1/M},\mu,\nu}^{(1,n)}}_{M} \;, 
\end{eqnarray} 
whose interferometric representation is given in Fig.~\ref{fig:griglia_main}(c). 
As for the LIM, the DIM reduces to the memoryless thermal attenuator when the memory parameter $\mu$ vanishes, as shown in Section~\ref{SM_sec_def_model}. 
If $\mu>0$, instead, the local environments are not always in the thermal state $\tau_\nu$, and their state depends on all previous input signals, i.e.~the model exhibits memory effects. As we show in
Theorem~\ref{thm_strong_sm} in Section~\ref{SM_sec_def_model}, 
the mapping $\Phi_{\lambda,\mu,\nu}^{(M,n)}$ \emph{strongly converges} for $M\to\infty$ to an $n$-mode quantum channel, which we denote as $\Phi_{\lambda,\mu,\nu}^{(n)}$.  
The family of quantum channels $\{\Phi_{\lambda,\mu,\nu}^{(n)}\}_{n\in\N}$, which forms a \emph{quantum memory channel}~\cite{memory-review}, characterises our model of optical fibre with memory effects. By analysing such a channel, in Section~\ref{sec_solution_property_sm} we show that the DIM satisfies the above Properties~1--2. We provide a more detailed explanation of the DIM in Sections~\ref{SM_sec_def_model}--\ref{sec_solution_property_sm} below.
 \begin{figure*}[t]
	\includegraphics[width=0.95 \linewidth]{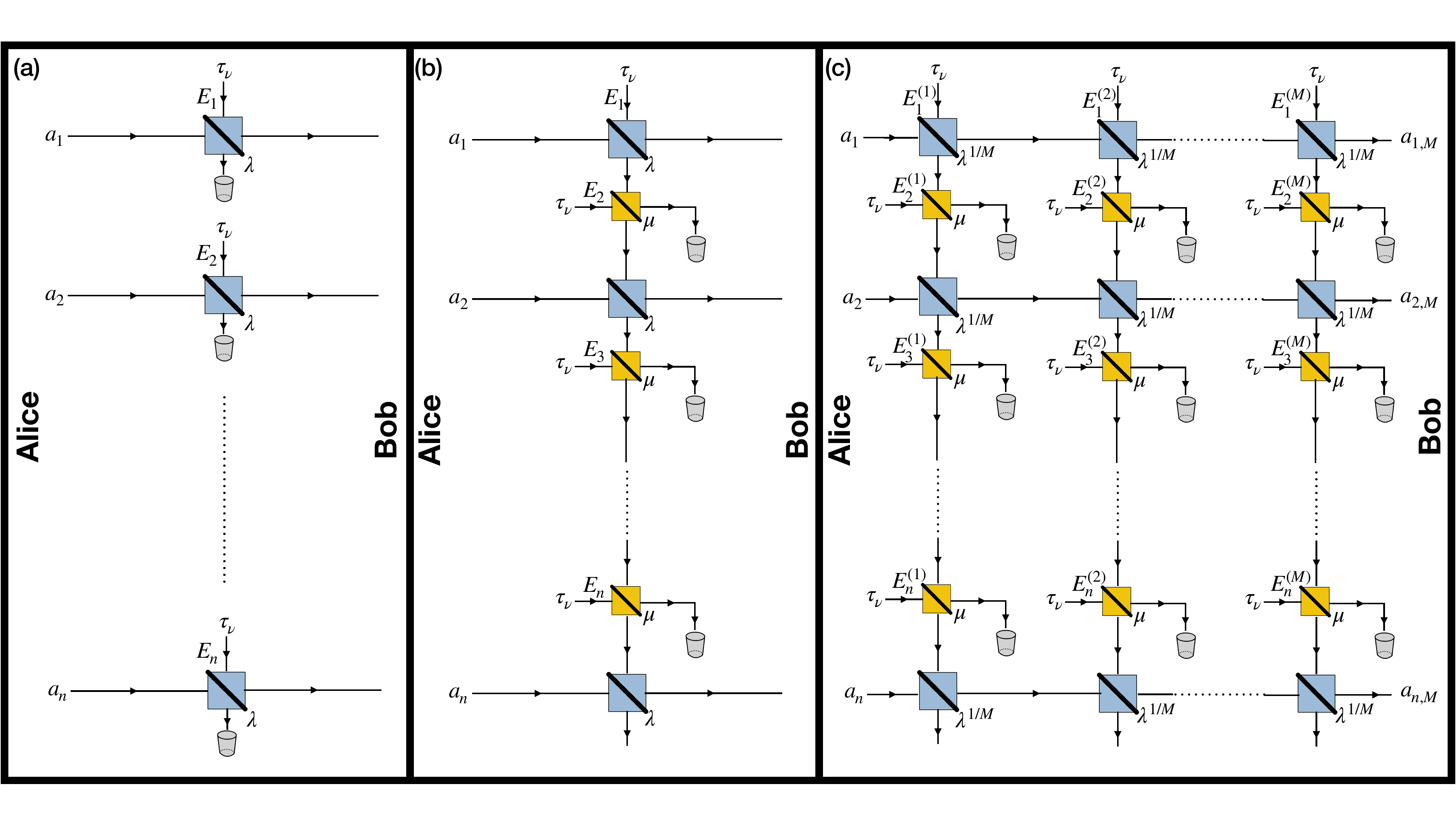}
	\caption{Interferometric representation of the $n$-mode quantum channels which describe the transmission of e.m.~pulses along an optical fibre of transmissivity $\lambda$ and local temperature $\nu$ for different memory configurations. Panel (a) memoryless regime: here each one of the input modes $a_1$, $a_2$, $\cdots$, $a_n$ evolves independently from the others, undergoing to the same thermal attenuation mapping $\mathcal{E}_{\lambda,\nu}$ induced by beam splitter couplings with the local thermal state $\tau_\nu$ of the local environments $E_1$, $E_2$, $\cdots$, $E_n$. Panel (b) LIM channel 
	 $\Phi_{\lambda,\mu,\nu}^{(1,n)}$: in this case signal cross-talks are
mediated by the yellow beam splitters of transmissivity~$\mu$, which model the resetting mechanism that endeavours to restore $E_1$ to its initial configuration $\tau_\nu$. These allow the photons lost by the $i$th input mode to emerge in the output of the subsequent ones by letting them to interfere with the local bath  $E_{i+1}$, $E_{i+2}$, $\cdots$, $E_n$. Setting $\mu=0$ the LIM reduces to the memoryless case of panel (a), i.e.\ $\Phi_{\lambda,0,\nu}^{(1,n)}=\mathcal{E}_{\lambda,\nu}^{\otimes n}$.
	Panel (c) DIM channel
	$\Phi_{\lambda,\mu,\nu}^{(M,n)}$: the fibre is described  by the concatenation~(\ref{concatenated}) of 
	$M$ LIM channels $\Phi_{\lambda^{1/M},\mu,\nu}^{(1,n)}$ of transmissivity $\lambda^{1/M}$.   In the Heisenberg representation, $\Phi_{\lambda,\mu,\nu}^{(M,n)}$ maps the annihilation operator $a_i$ (depicted on the left) of the $i$th input signal into the annihilation operators $a_{i,M}$ (depicted on the right) of the $i$th output signal for all $i=1,2,\ldots,n$. For $j=1,\cdots, M$,  the symbols $E^{(j)}_{1}$, $E^{(j)}_{2}$, $\cdots$, $E^{(j)}_{n}$ 
represent the single-mode environments associated with the $j$th infinitesimal optical fibre element: all of them are initialised in the same thermal state 
$\tau_\nu$. }
\label{fig:griglia_main}
\end{figure*}

\section{Overview of main results} \label{Sec_Main_results} 
In the absence of memory effects ($\mu=0$), it is known that no quantum communication tasks can be achieved when the transmissivity is sufficiently low ($Q=0$ for $\lambda\le \frac{1}{2}$, and $Q_2=K=0$ for $\lambda\le \frac{\nu}{\nu+1}$). However, in the forthcoming Theorem~\ref{main_pos_cap_thm_main}, we establish that for any transmissivity $\lambda>0$ and thermal noise $\nu\ge0$, there exists a critical value of the memory parameter $\mu$ above which it becomes possible to achieve qubit distribution ($Q>0$), entanglement distribution ($Q_2>0$), and secret-key distribution ($K>0$). This result is particularly intriguing as it demonstrates that --- at least within our model --- memory effects provide an advantage, enabling quantum communication tasks to be performed in highly noisy regimes where it was previously considered impossible without the use of quantum repeaters. Let $Q(\lambda,\mu,\nu)$, $Q_2(\lambda,\mu,\nu)$, and $K(\lambda,\mu,\nu)$ be the quantum capacity, two-way quantum capacity, and secret key capacity, respectively, of the DIM quantum memory channel $\{\Phi_{\lambda,\mu,\nu}^{(n)}\}_{n\in\N}$.
\begin{thm}\label{main_pos_cap_thm_main}
    Let $\lambda\in(0,1)$ and $\mu\in[0,1)$. 
    In the absence of thermal noise ($\nu=0$), it holds that
    \bb\label{main_cond_strictly_q_main}
Q(\lambda,\mu,\nu=0)>0\,\Longleftrightarrow\,\sqrt{\mu}>\frac{ \log_2\left(\frac{1}{\lambda}\right)-1 }{  \log_2\left(\frac{1}{\lambda}\right)+1  }\,.
    \ee
    In addition, for all $\nu\ge0$ it holds that
	\bb\label{main_cond_strictly_q2_main} K(\lambda,\mu,\nu),\,Q_2(\lambda,\mu,\nu)>0\Longleftrightarrow 
	 &\sqrt{\mu}> \frac{ \ln\left(\frac{1}{\lambda}\right)-\ln(1+\frac{1}{\nu}) }{  \ln\left(\frac{1}{\lambda}\right)+\ln(1+\frac{1}{\nu})  }\,.
	\ee
\end{thm}
The proof of Theorem~\ref{main_pos_cap_thm_main} is provided in Theorem~\ref{pos_cap_thm_sm} in Section~\ref{Sec_capacities_DIM}.

By relating $\mu$ to the temporal interval $\delta t$ between subsequent input signals (e.g.~$\mu=e^{-\delta t/t_E}$ with $t_E$ being the thermalisation timescale), Theorem~\ref{main_pos_cap_thm_main} can also be expressed in terms of the critical $\delta t$ below which the capacities $Q$, $Q_2$, and $K$ all become strictly positive.

The forthcoming Theorem~\ref{main_Cap_absence_thermal_main} provides the exact solution for the capacities in the absence of thermal noise. It is worth noting that in the presence of thermal noise ($\nu > 0$), the capacities remain unknown even in the absence of memory effects ($\mu = 0$), as the capacities of the thermal attenuator are currently unknown.
\begin{thm}\label{main_Cap_absence_thermal_main}
    Let $\lambda\in(0,1)$ and $\mu\in[0,1)$. 
    In absence of thermal noise ($\nu=0$), it holds that
    \bb\label{main_formula_cap_nu_0}
        Q\left( \lambda,\mu,\nu=0\right)&=\int_{0}^{2\pi}\frac{\mathrm{d}x}{2\pi}\, \max\left\{0,\log_2 \left(\frac{\eta^{(\lambda,\mu)}(x)}{1-\eta^{(\lambda,\mu)}(x)}\right)\right\}\,,\\
        Q_2\left( \lambda,\mu,\nu=0 \right)&=K\left( \lambda,\mu,\nu=0 \right)\\&=\int_{0}^{2\pi}\frac{\mathrm{d}x}{2\pi}\,\log_2\left(\frac{1}{1-\eta^{(\lambda,\mu)}(x)}\right)\,,
    \ee
    where
    \bb\label{main_formula_effect_transm_main}
        \eta^{(\lambda,\mu)}(x) \coloneqq \lambda^{\frac{1-\mu}{ 1+\mu-2\sqrt{\mu}\cos(x/2) }}\quad\forall \,x\in[0,2\pi]\,.
    \ee
\end{thm}
The proof of Theorem~\ref{main_Cap_absence_thermal_main} is provided in Theorem~\ref{Cap_absence_thermal_sm} in Section~\ref{Sec_capacities_DIM}.

In Fig.~\ref{fig_main}, we plot the capacities given in the above Theorem~\ref{main_Cap_absence_thermal_main}. If $\mu=0$ (i.e.~in the memoryless case), all these capacities are equal to the corresponding capacities of the pure-loss channel. Upon plotting the capacities in Fig.~\ref{fig_main}, one observes that for any $\lambda\in (0,1]$ all the capacities $Q\left(\lambda,\mu,\nu=0\right)$, $Q_2\left(\lambda,\mu,\nu=0\right)$, and $K\left(\lambda,\mu,\nu=0\right)$ are monotonically increasing in $\mu \in [0,1]$. Hence, this means that if the memory parameter $\mu$ increases (corresponding to a decrease in the time interval between consecutive signals), then quantum communication performances improve. 

\begin{figure*}[t]
\begin{minipage}{0.5\linewidth}
\centering
\includegraphics[width=1.0\linewidth]{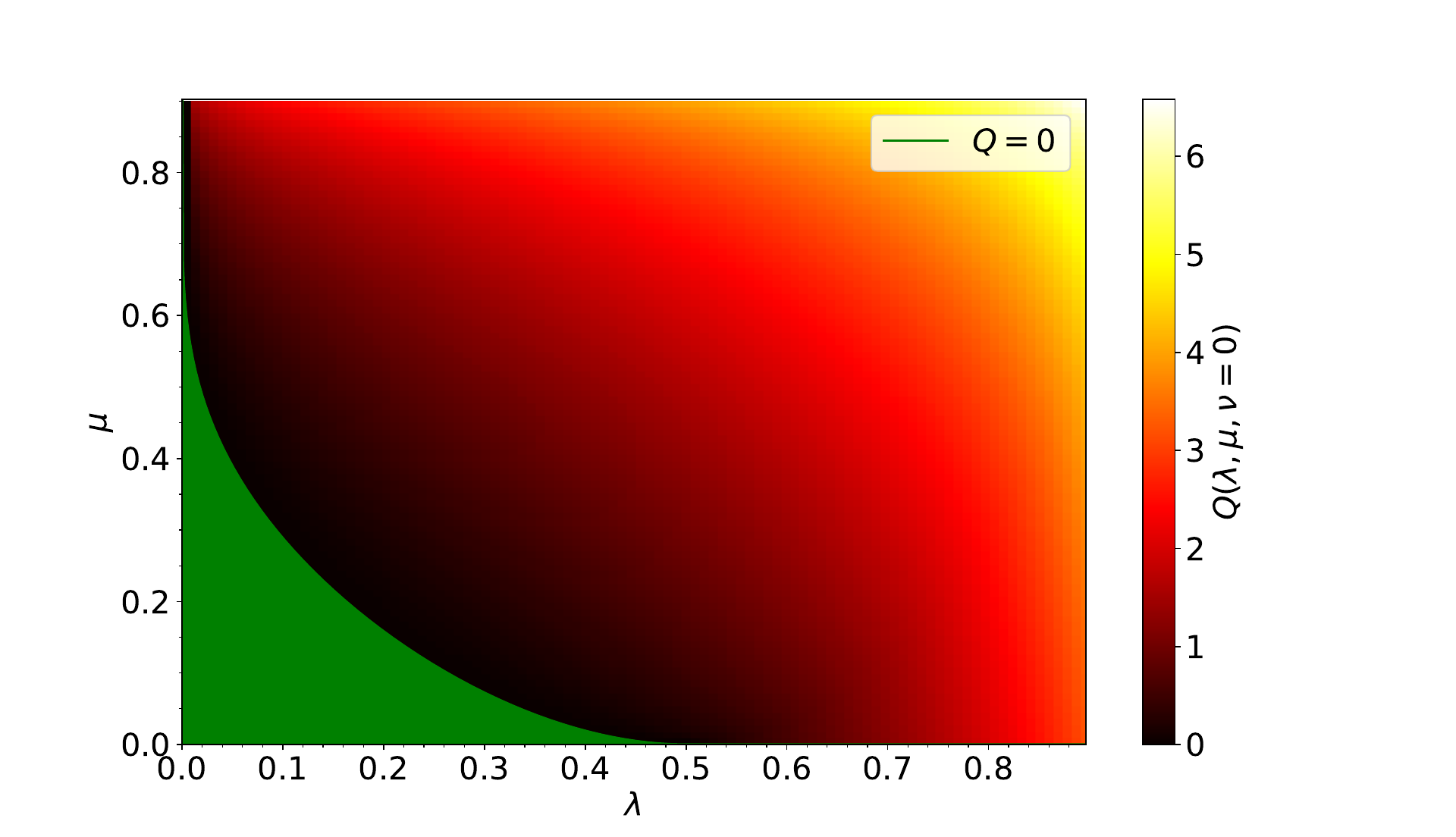} \\
(a)
\end{minipage}%
\begin{minipage}{0.5\linewidth}
\centering
\includegraphics[width=1.0\linewidth]{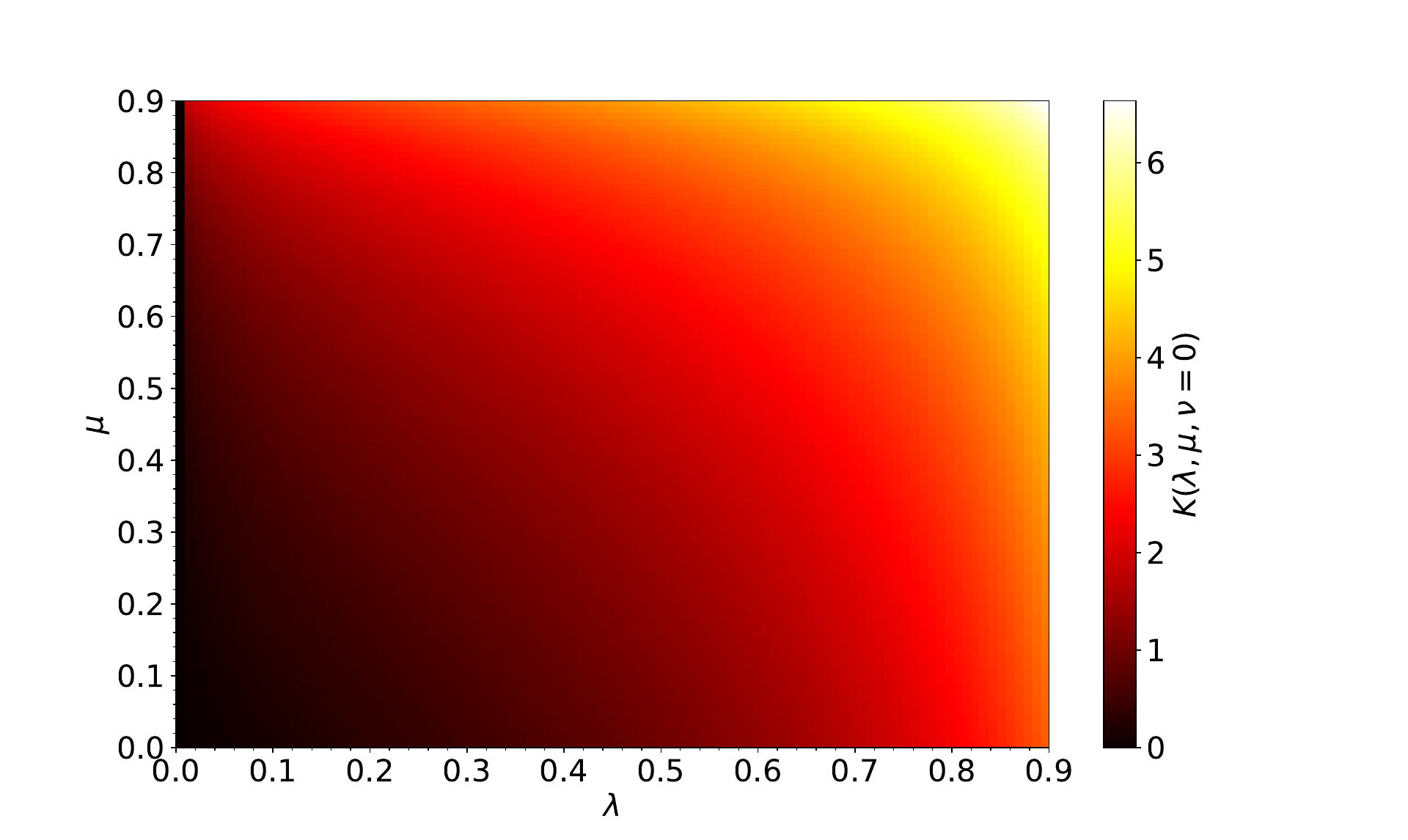} \\
(b)
\end{minipage}
\caption{In \textbf{(a)} we plot the quantum capacity $Q(\lambda,\mu,\nu=0)$ of the quantum memory channel $\{\Phi^{(n)}_{\lambda,\mu,\nu=0}\}_{n\in\N}$ with respect to $\lambda,\mu\in[0,0.9]$ for $\nu=0$. The green region is the one where the quantum capacity vanishes, which is given by all the points $(\lambda,\mu)$ such that $\sqrt{\mu}\le\frac{ \log_2\left(\frac{1}{\lambda}\right)-1 }{  \log_2\left(\frac{1}{\lambda}\right)+1  }$, as established by Theorem~\ref{main_pos_cap_thm_main}.$\quad$ In
\textbf{(b)} we plot the secret-key capacity $K(\lambda,\mu,\nu=0)$ of the quantum memory channel $\{\Phi^{(n)}_{\lambda,\mu,\nu=0}\}_{n\in\N}$ with respect to $\lambda,\mu\in[0,0.9]$ for $\nu=0$. This capacity is exactly equal to the two-way quantum capacity $Q_2(\lambda,\mu,\nu=0)$, as guaranteed by Theorem~\ref{main_Cap_absence_thermal_main}.$\quad$ The capacities plotted in \textbf{(a)} and in \textbf{(b)} are calculated by exploiting the expressions given in \eqref{main_formula_cap_nu_0}. In \textbf{(a)} and in \textbf{(b)} we have limited the range of $\lambda$ and $\mu$ to the interval $[0,0.9]$ because all the capacities $Q$, $Q_2$, and $K$ diverge when either $\lambda$ or $\mu$ approaches $1$.}
\label{fig_main}
\end{figure*}

\section{Notation and basic notions}\label{sec_sm_notation}
Let $\HH$ be an Hilbert space, let $\mathcal{L}(\HH)$ the space of linear operators on $\HH$, and let $\mathfrak{S}(\HH)$ be the set of quantum states on $\HH$. The trace norm of a linear operator $\Theta$ is defined by $\|\Theta\|_1\coloneqq \Tr\sqrt{\Theta^\dagger\Theta}$. A map $\Phi:\mathcal{L}(\HH)\to\mathcal{L}(\HH)$ is a quantum channel if it is linear, completely positive, and trace preserving. Given a quantum channel $\Phi$ and a sequence of quantum channels $\{  
\Phi_k \}_{k\in\N}$, we say that $\{\Phi_k\}_{k\in\N}$ \emph{strongly converges} to $\Phi$ when
\bb\label{strong_convergence}
    \lim\limits_{k\to\infty}\|\Phi_k(\rho)-\Phi(\rho)\|_1=0\quad\forall\,\rho\in\mathfrak{S}(\HH)\,.
\ee
We now define quantum memoryless and memory channels, and introduce relevant notations.
\begin{Def}
    Let $\{\Phi^{(n)}\}_{n\in\N}$ be a family of quantum channels with $\Phi^{(n)}:\mathcal{L}(\HH^{\otimes n})\to\mathcal{L}(\HH^{\otimes n})$ for all $n\in\N$. The family $\{\Phi^{(n)}\}_{n\in\N}$ is \emph{memoryless} if there exists a quantum channel $\Phi:\mathcal{L}(\HH)\to\mathcal{L}(\HH)$ such that $\Phi^{(n)}=\Phi^{\otimes n}$ for all $n\in\N$. In such a scenario, we will also refer to $\Phi$ as a memoryless channel.
\end{Def}  
\begin{Def}~\label{def_memory}
    Let $\{\Phi^{(n)}\}_{n\in\N}$ be a family of quantum channels with $\Phi^{(n)}:\mathcal{L}(\HH^{\otimes n})\to\mathcal{L}(\HH^{\otimes n})$ for all $n\in\N$. The family $\{\Phi^{(n)}\}_{n\in\N}$ is a \emph{memory quantum channel} if it is not memoryless. In such a scenario, the term "$n$ channel uses" corresponds to the quantum channel $\Phi^{(n)}$.
\end{Def}

\subsection{Quantum information with continuous variables systems}
Let us provide a brief overview of the formalism of quantum information with continuous variable systems~\cite{BUCCO}. An $m$-mode system is associated with the Hilbert space $L^2(\mathbb R^m)$ of all square-integrable complex-valued functions over $\mathbb{R}^m$. Each mode represents a single mode of electromagnetic radiation with a definite frequency and polarisation. Let us denote as $S_1$, $S_2$, $\ldots$, $S_m$ the $m$ modes of the system. For each $j=1,2,\ldots,m$, the annihilation operator $a_j$ of the mode $S_j$ is defined as $a_j\coloneqq \frac{\hat{x}_j+i\hat{p}_j}{\sqrt{2}}$, where $\hat{x}_j$ and $\hat{p}_j$ are the well-known position and momentum operators of the mode $S_j$. The operator $a_j^\dagger$ is referred to as the creation operator of the mode $S_j$. The operator $a_j^\dagger a_j$ is the photon number operator of $S_j$, and the state 
\bb
    \ket{n}_{S_j}\coloneqq \frac{(a_j^\dag)^n}{\sqrt{n!}} \ket{0}_{S_j}\,,
\ee
is its $n$th Fock state, where $\ket{0}_{S_j}$ denotes its vacuum state. The characteristic function $\chi_\rho: \mathbb{C}^{m}\to \mathbb{C}$ of an $m$-mode state $\rho$ is defined as 
\bb
\chi_\rho(\alpha) \coloneqq&\ \Tr\left[ \rho D(\alpha) \right] \,,
\ee
where for all $\alpha\in \mathbb{C}^{m}$ the displacement operator $ D(\alpha)$ is defined as 
\bb\label{def_charact_func}
D(\alpha) \coloneqq \exp\left[\sum_{j=1}^{m} (\alpha_j a^\dag_j - \alpha_j^* a_j)\right]\,.
\ee
Any state $\rho$ can be written in terms of its characteristic function as
\bb
\rho =\ \int \frac{\mathrm{d}^{2m}\alpha}{\pi^m}\, \chi_\rho(\alpha) D(-\alpha)\,,
\ee
and, consequently, quantum states and characteristic functions are in one-to-one correspondence.   The following Lemma (see (\cite[Lemma 4]{G-dilatable} or~\cite[Theorem 2]{Cushen1971})) will be useful for the following.
\begin{lemma}~\label{lemmino_charact}
Let $m\in\N$, let $\{\sigma_k\}_{k\in\N}$ be a sequence of $m$-mode states, and let $\sigma$ be an $m$-mode state. Then, $\{\sigma_k\}_{k\in\N}$ converges in trace norm to $\sigma$ if and only if the sequence of characteristic functions $\{\chi_{\sigma_k}\}_{k\in\N}$ converges pointwise to the characteristic function $\chi_\sigma$. In formula, it holds that
\bb
    \lim\limits_{k\to\infty}\|\sigma_k-\sigma\|_1=0\Longleftrightarrow\lim\limits_{k\to\infty}\chi_{\sigma_k}(\alpha)=\chi_{\sigma}(\alpha)\quad\forall\,\alpha\in\mathbb{C}\,.
\ee
\end{lemma}
{
Let $a_1,a_2,\ldots,a_m$ be the annihilation operators of an $m$-mode system. For any $m\times m$ unitary matrix $u$ there exists an $m$-mode Gaussian unitary $U_u$, referred to as a "passive unitary", such that~\cite{BUCCO}
\bb
    (U_u)^\dagger a_i U_u=\sum_{j=1}^m u_{ij} a_{j}\qquad\forall i\in\{1,\ldots,n\}\,.
\ee
or, in vectorial notation,
\bb
(U_u)^\dagger \mathbf{a}U_u= u\mathbf{a}\,,
\ee 
where we have introduced the vector of annihilation operators $\textbf{a}\coloneqq (a_1,a_2,\ldots,a_m)^\intercal$. It simple to prove that any passive unitary $U_u$ preserves the total photon number, i.e.
    \bb
        (U_u)^\dagger \left(\sum_{j=1}^m a_j^\dagger a_j \right)U_u=\sum_{j=1}^m a_j^\dagger a_j\,.
    \ee
Moreover, the unitary channel
\bb\label{definition_passive unitary}
    \mathcal{U}_u(\cdot)\coloneqq U_u \cdot (U_u)^\dagger
\ee
is called a "passive unitary channel". It acts at the level of characteristic functions as
    \bb\label{passive_transf_charact}
        \chi_{\mathcal{U}_u(\rho)}(\alpha)=\chi_\rho(u^\dagger \alpha)\qquad\forall\,\alpha\in\mathbb{C}^m\, 
    \ee
for any $m$-mode quantum state $\rho$.

}
\subsection{Beam splitter and thermal attenuator}
Let $\HH_S$ and $\HH_E$ be single-mode systems and let $a$ and $b$ denote their annihilation operators, respectively.
For all $\lambda\in[0,1]$ the unitary operator describing a beam splitter of transmissivity $\lambda$ is
\bb
	U_{\lambda}^{S E}\coloneqq\exp\left[\arccos\sqrt{\lambda}\left(a^\dagger b-a\, b^\dagger\right)\right]\,.
\ee
Under the beam splitter unitary, the annihilation operators $a$ and $b$ transform as
\bb\label{transf_beam}
    \left(U_\lambda^{SE}\right)^\dagger a\, U_{\lambda}^{SE}&=\sqrt{\lambda}\,a+\sqrt{1-\lambda}\,b\,,\\
    U_\lambda^{SE} a\, \left(U_{\lambda}^{SE}\right)^\dagger&=\sqrt{\lambda}\,a-\sqrt{1-\lambda}\,b\,,\\
    \left(U_\lambda^{SE}\right)^\dagger b\, U_{\lambda}^{SE}&=-\sqrt{1-\lambda}\,a+\sqrt{\lambda}\,b\,,\\
    U_\lambda^{SE} b\, \left(U_{\lambda}^{SE}\right)^\dagger&=\sqrt{1-\lambda}\,a+\sqrt{\lambda}\,b\,.
\ee
At the level of characteristic functions, the beam splitter transforms any two-mode state $\rho\in\mathfrak{S}(\HH_S\otimes\HH_E)$ as
\bb\label{charact_beam}
    \chi_{U^{SE}_\lambda\rho U_\lambda^{{SE}^\dag}}(\alpha,\beta) &\coloneqq \Tr\left[U^{SE}_\lambda\rho U_\lambda^{{SE}^\dag}\, D_S(\alpha)\otimes D_E(\beta)\right]
    \\&= \chi_\rho\left( \sqrt\lambda\,\alpha - \sqrt{1-\lambda}\,\beta,\, \sqrt{1-\lambda}\,\alpha + \sqrt\lambda\, \beta \right)
\ee
for all $\alpha,\beta\in\mathbb{C}$. For all $\lambda\in[0,1]$ and $\nu\ge0$, a \emph{thermal attenuator} $\mathcal{E}_{\lambda,\nu}:\mathfrak{S}(\HH_S)\to\mathfrak{S}(\HH_S)$ is a quantum channel defined as follows: 
\bb\label{def_therm}
    \mathcal{E}_{\lambda,\nu}(\rho)\coloneqq\Tr_E\left[U_\lambda^{SE}\big(\rho^S \otimes\tau_\nu^E \big) {U_\lambda^{SE}}^\dagger\right]\,,
\ee
Here, $\tau_\nu\in\mathfrak{S}(\HH_E)$ denotes the thermal state with mean photon number equal to $\nu$, which is defined by 
\bb
    \tau_{\nu}\coloneqq \frac{1}{\nu+1}\sum_{n=0}^\infty \left(\frac{\nu}{\nu+1}\right)^{n}\ketbra{n}\,,
\ee
and its characteristic function reads
\bb\label{charact_therm}
\chi_{\tau_\nu}(\alpha) \coloneqq  e^{-\frac12 (2\nu+1)|\alpha|^2}\quad\forall\,\alpha\in\mathbb{C}\,.
\ee
By exploiting~\eqref{charact_beam} and~\eqref{charact_therm}, it can be shown that the thermal attenuator transforms any single-mode state $\rho$ at the level of characteristic functions as follows:
\bb\label{caract_att}
\chi_{\mathcal{E}_{\lambda,\nu}(\rho)}(\alpha) \coloneqq \chi_\rho\left(\sqrt\lambda \alpha\right)\, e^{-\frac12 (2\nu+1)(1-\lambda)|\alpha|^2}\quad\forall\,\alpha\in\mathbb{C}\,.
\ee
By exploiting~\eqref{caract_att} and the one-to-one correspondence between quantum states and characteristic functions, one can show that the following composition rule holds for any $\lambda_1,\lambda_2\in[0,1]$ and $\nu\ge0$:
\bb\label{composition_them}
\mathcal{E}_{\lambda_1,\nu}\circ\mathcal{E}_{\lambda_2,\nu}=\mathcal{E}_{\lambda_1\lambda_2,\nu}\,.
\ee

\subsection{Quantum capacity, two-way quantum capacity, and secret-key capacity}\label{sec_def_capacity}
In the framework of quantum Shannon theory~\cite{MARK,Sumeet_book}, the fundamental limitations of point-to-point quantum communication are established by the \emph{capacities} of quantum channels. The capacities quantify the maximum amount of information that can be reliably transmitted per channel use in the asymptotic limit of many uses. Different notions of capacities have been defined, based on the type of information to be transmitted, such as qubits or secret-key bits, and the additional resources permitted in the protocol design, such as classical feedback. In this paper, we investigate three distinct capacities: the quantum capacity $Q$, the two-way quantum capacity $Q_2$, and the secret-key capacity $K$. We recall that $Q$ measures the efficiency in the transmission of qubits with no additional resources, while $Q_2$ and $K$ instead gauge the efficiency in the transmission of qubits and secret key, respectively, with the additional free resource of a public two-way classical communication channel between the sender and the receiver~\cite{MARK, Sumeet_book}. We give rigorous definitions of the quantum capacity $Q$, two-way quantum capacity $Q_2$, and secret-key capacity $K$ below in Sections~\ref{paragraph_quantum_cap}, \ref{paragraph_two_way}, and \ref{paragraph_secret_key}, respectively (see~\cite{Sumeet_book,memory-review} for another presentation).

Let $C$ represent one of the three capacities mentioned above, namely $C\in\{Q,Q_2,K\}$. We will employ the following notation:
\begin{itemize}
    \item The capacity $C$ of a memoryless quantum channel $\Phi$ will be denoted as $C(\Phi)$.
    \item The capacity $C$ of a quantum memory channel $\{\Phi^{(n)}\}_{n\in\mathbb{N}}$ will be denoted as $C\left(\{\Phi^{(n)}\}_{n\in\mathbb{N}}\right)$ (refer to Definition~\ref{def_memory}).
\end{itemize}
Hence, for a quantum memoryless channel $\Phi$ it holds that $C(\{\Phi^{\otimes n}\}_{n\in\N})=C(\Phi)$.

{
\subsubsection{Definition of quantum capacity of quantum (memory) channels}\label{paragraph_quantum_cap}
In this paragraph, we briefly define the quantum capacity of quantum (memory) channels~\cite{Sumeet_book}. Let us begin with the definition of $(|M|,\varepsilon)$-code for quantum communication. Given a quantum channel $\Phi_{A\to B}$, an \emph{$(|M|,\varepsilon)$-code for quantum communication over $\Phi_{A\to B}$} consists of a channel $\mathcal{E}_{M\to A}$, called the "encoding channel", and a channel $\mathcal{D}_{B \to M}$, called the "decoding channel", such that the channel fidelity between $\mathcal{D}_{B \to M}\circ \Phi_{A\to B}  \circ \mathcal{E}_{M\to A}$ and the identity channel $\operatorname{id}_{M}$ is not smaller than $1-\varepsilon$:
\bb
F( \mathcal{D}_{B \to M}\circ \Phi_{A\to B}  \circ \mathcal{E}_{M\to A}, \operatorname{id}_{M} ) \geq 1-\varepsilon
\ee
where the channel fidelity between two channels $\Phi_1$ and $\Phi_2$ is defined in terms of terms of the fidelity $F$ between states as follows~\cite{Sumeet_book}:
\bb\label{channel_fidelity}
F(\Phi_1,\Phi_2) \coloneqq \inf_{\rho} F ((\operatorname{id} \otimes \Phi_1)(\rho),(\operatorname{id} \otimes \Phi_2)(\rho))\,.
\ee
The optimisation in \eqref{channel_fidelity} is over every (arbitrarily large) reference system and every bipartite state $\rho$ on the channel input and on the reference system. The \emph{one-shot quantum capacity $Q_{\varepsilon}(\Phi_{A\to B})$ of the channel $\Phi_{A\to B}$} is defined as
\bb
&Q_{\varepsilon}(\Phi) \coloneqq \sup_{\mathcal{E},\mathcal{D}} \{\log_2 |M|\\
& : \exists \text{$(|M|,\varepsilon)$-code for quantum communication over } \Phi_{A\to B}\},
\ee
where the optimisation is over every encoding channel $\mathcal{E}$ and decoding channel $\mathcal{D}$.

Let $\{\Phi^{(n)}\}_{n\in\N}$ be a quantum memory channel, which maps $n$ (Alice's) input systems $A_1,\ldots, A_n$ to $n$ (Bob's) output systems $B_1\ldots B_n$ for all $n\in\N$. The \emph{quantum capacity} of $\{\Phi^{(n)}\}_{n\in\N}$ is defined as
\bb
Q\left(\{\Phi^{(n)}\}_{n\in\N}\right) \coloneqq \inf_{\varepsilon \in (0,1)} \liminf_{n \to \infty}\frac{1}{n}Q_{\varepsilon}(\Phi^{(n)}).
\ee
Naturally, given a quantum (memoryless) channel $\Phi$, its quantum capacity is defined as 
\bb
    Q(\Phi)\coloneqq Q\left(\{\Phi^{\otimes n}\}_{n\in\N}\right)=\inf_{\varepsilon \in (0,1)} \liminf_{n \to \infty}\frac{1}{n}Q_{\varepsilon}(\Phi^{\otimes n})\,.
\ee

\subsubsection{Definition of two-way quantum capacity of quantum (memory) channels}\label{paragraph_two_way}
In this paragraph, we briefly define the two-way quantum capacity of quantum (memory) channels~\cite{Sumeet_book}. This is a figure of merit of quantum communication in scenarios where the sender and the receiver are able to implement arbitrary local operations assisted
by classical communication (LOCC)~\cite{Sumeet_book,LOCC}.

By definition, an \emph{$(n,|M|,\varepsilon)$ LOCC-assisted quantum communication protocol over a (memoryless) quantum channel $\Phi_{A \to B}$} consists in a product state $\ket{0}_{A_0}\otimes\ket{0}_{B_0}$ and $n+1$ LOCC channels~\cite{Sumeet_book} $\mathcal{L}^{(0)}_{A_0B_0\to A'_1 A_1 B'_1}$, $\mathcal{L}^{(1)}_{A'_1 B_1 B'_1\to A'_2 A_2 B'_2}$, $\mathcal{L}^{(2)}_{A'_2 B_2 B'_2\to A'_3 A_3 B'_3}$, $\ldots$, $\mathcal{L}^{(n-1)}_{A'_{n-1} B_{n-1} B'_{n-1}\to A'_n A_n B'_n}$, $\mathcal{L}^{(n)}_{A'_n B_n B'_n\to M_A M_B}$ such that the final state of the protocol
\bb\label{final_state_eta}
\eta_{M_AM_B} &\coloneqq (\mathcal{L}^{(n)}_{A'_{n}B_{n}B'_{n}\to M_A M_B} \circ \Phi_{A_n \to B_n} \\
&\qquad\circ \mathcal{L}^{(n-1)}_{A'_{n-1}B_{n-1}B'_{n-1} \to A'_{n}A_{n}B'_{n}}\circ \cdots\\
&\qquad\circ  \mathcal{L}^{(1)}_{A'_{1}B_{1}B'_{1} \to A'_{2}A_{2}B'_{2}}  \\
& \qquad\circ\Phi_{A_1 \to B_1}\circ \mathcal{L}^{(0)}_{A_0B_0\to A'_1 A_1 B'_1})(\ketbra{0}_{A_0}\otimes\ketbra{0}_{B_0})
\ee
satisfies
\bb
F(\eta_{M_A M_B}, \Gamma_{M_A M_B}) \geq 1- \varepsilon,
\ee
where $\Gamma_{M_A M_B}$ is a maximally entangled state of Schmidt rank $|M|$. A pictorial representation of such a protocol is provided in Fig.~\ref{fig_LOCC_protocol}.
\begin{figure}[!h]
	\centering
	\includegraphics[width=1\linewidth]{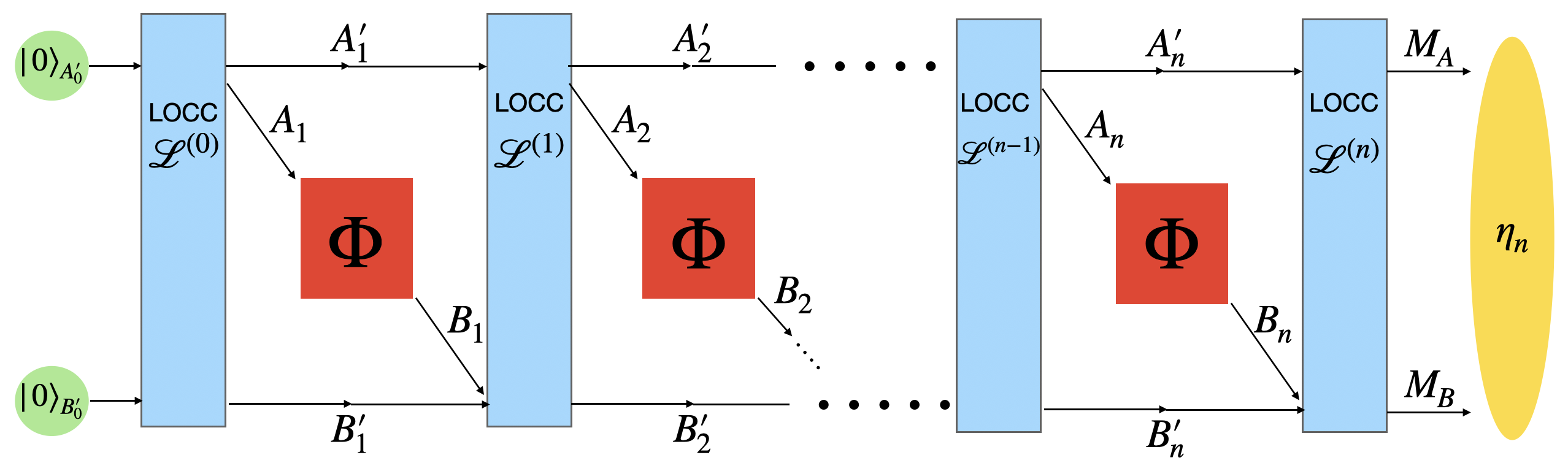}
	\caption{Pictorial representation of an LOCC-assisted quantum communication protocol over a (memoryless) quantum channel $\Phi$. 
	}
    \label{fig_LOCC_protocol}
\end{figure}

The \emph{$n$-shot two-way quantum capacity} of $\Phi_{A\to B}$ is defined as
\bb
Q_{2;n,\varepsilon}(\Phi) &\coloneqq
\sup \{\frac{1}{n} \log_2 |M| \\
&: \exists (n,|M|,\varepsilon) \text{ LOCC-assisted quantum}    \\
&\text{ communication protocol over}  \Phi_{A\to B}  \},
\ee
where the optimisation is over every LOCC-assisted quantum communication protocol. Finally, the \emph{two-way quantum capacity} of $\Phi_{A\to B}$ is defined as
\bb
Q_{2}(\Phi) \coloneqq \inf_{\varepsilon \in (0,1)} \liminf_{n \to \infty}Q_{2;n,\varepsilon}(\Phi).
\ee

Let us now define the two-way quantum capacity of a quantum memory channel $\{\Phi^{(n)}\}_{n\in\N}$, which maps $n$ (Alice's) input systems $A_1,\ldots, A_n$ to $n$ (Bob's) output systems $B_1\ldots B_n$ for all $n\in\N$. An \emph{$(n,|M|,\varepsilon)$ LOCC-assisted quantum communication protocol over $\Phi^{(n)}_{A_1\ldots A_n \to B_1\ldots B_n}$} is defined essentially in the same way as in the memoryless case above, except that the expression of the final state in \eqref{final_state_eta} now involves the $n$ uses of the quantum memory channel $\Phi^{(n)}_{A_1\ldots A_n \to B_1\ldots B_n}$ instead of the $n$ uses of the memoryless channel $\Phi$. A pictorial representation of such a protocol over a quantum memory channel is provided in Fig.~\ref{fig_LOCC_protocol_memory}.
\begin{figure}[!h]
	\centering
	\includegraphics[width=1\linewidth]{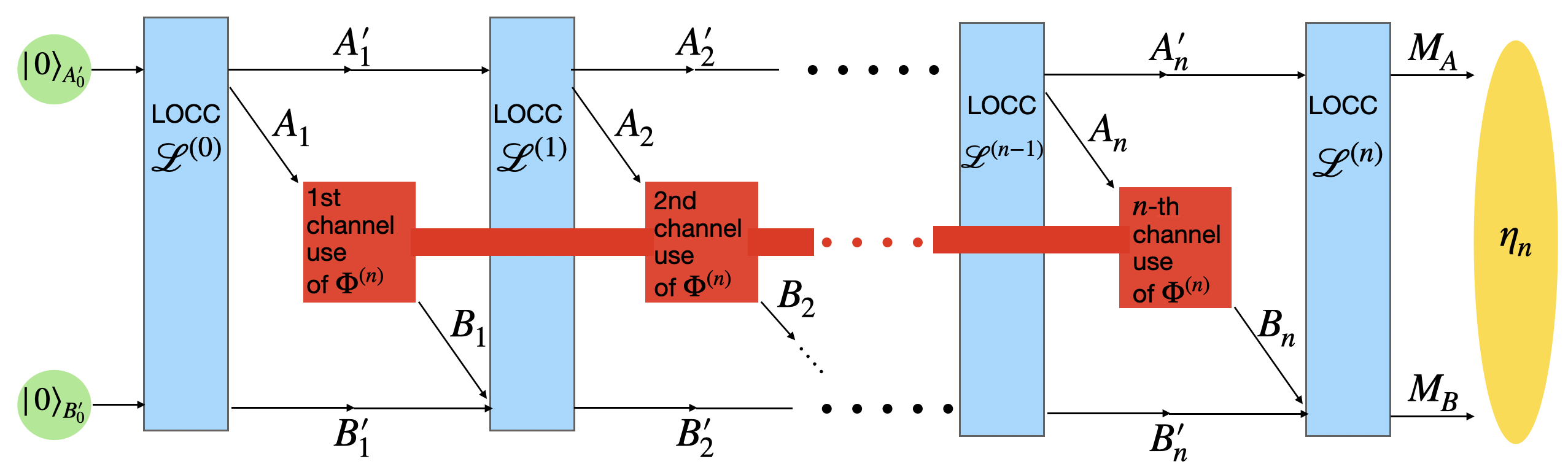}
	\caption{Pictorial representation of an LOCC-assisted quantum communication protocol over a quantum memory channel $\Phi^{(n)}_{A_1\ldots A_n \to B_1\ldots B_n}$.}
    \label{fig_LOCC_protocol_memory}
\end{figure}
The \emph{$n$-shot two-way quantum capacity} of a quantum memory channel $\Phi^{(n)}_{A_1\ldots A_n \to B_1\ldots B_n}$ is defined as
\bb
&Q_{2;n,\varepsilon}(\Phi^{(n)}) \coloneqq
\sup \{\frac{1}{n} \log_2 |M| : \exists (n,|M|,\varepsilon) \text{ LOCC-assisted} \\
&\text{ quantum communication protocol over } \Phi^{(n)}_{A_1\ldots A_n\to B_1\ldots B_n}  \},
\ee
where the optimisation is over every LOCC-assisted quantum communication protocol. Finally, the \emph{two-way quantum capacity} of a quantum memory channel $\{\Phi^{(n)}\}_{n\in\N}$ is defined as
\bb
Q_{2}\left(\{\Phi^{(n)}\}_{n\in\N}\right) \coloneqq \inf_{\varepsilon \in (0,1)} \liminf_{n \to \infty}Q_{2;n,\varepsilon}(\Phi^{(n)}).
\ee

As defined above, the two-way quantum capacity $Q_{2}$ of a quantum (memory) channel is the maximum achievable rate of \emph{ebits} (i.e.~two-qubit maximally entangled states) in the asymptotic limit of many channel uses when the sender and the receiver are assisted by LOCCs. Note that, since the sender and receiver can exploit an ebit to teleport an arbitrary qubit~\cite{teleportation} and since the ability of transmitting an arbitrary qubit implies the ability of distributing an ebit, the two-way quantum capacity $Q_{2}$ is not only the maximum achievable rate of ebits but also the maximum achievable rate of qubits. In other words, $Q_{2}$ is the proper figure of merit for both entanglement distribution and quantum communication.

\subsubsection{Definition of secret-key capacity of quantum (memory) channels}\label{paragraph_secret_key}
In this paragraph, we briefly define the secret-key capacity of quantum (memory) channels~\cite{Sumeet_book}. Let us start with the definition of an $(n,|M|,\varepsilon)$ protocol for secret key agreement over a (memory) quantum channel. The latter is defined essentially in the same way as we defined an LOCC-assisted quantum communication protocol in paragraph~\ref{paragraph_two_way}, except for the last LOCC and for the final state of the protocol. The last LOCC is now substituted by a more general LOCC $\mathcal{L}^{(n)}_{A'_{n}B_{n}B'_{n}\to M_A M_B S_A S_B}$, which outputs not only Alice's system $ M_A$ and Bob's system $ M_B $ but also an additional  Alice's system $S_A$ and an additional Bob's system $S_B$. Moreover, the final state of the protocol, denoted as $\eta_{M_A M_B S_A S_B}$, must satisfy
\bb
F(\eta_{M_A M_B S_A S_B},\gamma_{M_A M_B S_A S_B}) \geq 1-\varepsilon,
\ee
where $\gamma_{M_A M_B S_A S_B} $ is a \emph{private state} of dimension $|M|$~\cite{private,Horodecki_2009_secret2}. By definition, a private state of dimension $|M|$ is a state the form~\cite{private,Horodecki_2009_secret2}:
\bb
\gamma_{M_A M_B S_A S_B} \coloneqq U_{M_A M_B S_A S_B} (\Phi_{M_A M_B} \otimes \theta_{S_A S_B} )U_{M_A M_B S_A S_B}^\dag\,,
\ee
where $U_{M_A M_B S_A S_B}$ is a unitary of the form
\bb
U_{M_A M_B S_A S_B} \coloneqq \sum_{i,j} \ketbra{i}_{M_A} \otimes \ketbra{j}_{M_B} \otimes U^{i,j}_{S_A S_B},
\ee
with each $U^{i,j}_{S_A S_B}$ being a unitary, $\Phi_{M_A M_B}$ representing a maximally entangled state of Schmidt rank $|M|$, and $\theta_{S_A S_B}$ being an arbitrary state. As proved in~\cite{private,Horodecki_2009_secret2}, generating a private state of dimension $|M|$ is equivalent to generating a secret key of length $|M|$, which has to be secret to any third party, typically referred to as Eve. Mathematically, generating a secret-key of length $|M|$ entails generating a tripartite "secret-key state", i.e.~a state of the form 
\bb
    \frac{1}{|M|} \sum_{m=0}^{|M|-1}\ketbra{m}_{M_A} \otimes \ketbra{m}_{M_B} \otimes \sigma_E\,
\ee
where $\sigma_E$ is an arbitrary state on Eve's system.

The $n$-shot secret-key capacity of a quantum (memory) channel $\Phi^{(n)}_{A_1\ldots A_n\to B_1\ldots B_n}$ is defined as
\bb
&K_{n,\varepsilon}(\Phi^{(n)}_{A_1\ldots A_n\to B_1\ldots B_n}) \coloneqq
\sup \{\frac{1}{n} \log_2 |M| : \exists (n,|M|,\varepsilon) \\
&\text{ secret-key-agreement protocol for } \Phi^{(n)}_{A_1\ldots A_n\to B_1\ldots B_n}  \},
\ee
where the optimisation is over every secret key agreement protocol.
Finally, the secret-key capacity of a quantum (memory) channel $\{\Phi^{(n)}\}_{n\in\N}$ is defined as
\bb
K\left(\{\Phi^{(n)}\}_{n\in\N}\right) \coloneqq \inf_{\varepsilon \in (0,1)} \liminf_{n \to \infty}K_{n,\varepsilon}(\Phi^{(n)}).
\ee
Naturally, the definition of secret-key capacity of a quantum \emph{memoryless} channel $\Phi$ can be retrieved by setting $\Phi^{(n)}\coloneqq \Phi^{\otimes n}$ in the definition above:
\bb
    K(\Phi)\coloneqq K\left(\{\Phi^{\otimes n}\}_{n\in\N}\right)\,.
\ee
Due to the fact that a maximally entangled state of Schmidt rank $|M|$ is a private state of dimension $|M|$, it follows that the two-way quantum capacity is always upper bounded by the secret-key capacity:
\bb
    Q_2\left(\{\Phi^{(n)}\}_{n\in\N}\right)\le K\left(\{\Phi^{(n)}\}_{n\in\N}\right)\,.
\ee
}

\subsubsection{Quantum capacity, two-way quantum capacity, and secret-key capacity of the pure-loss channel and thermal attenuator}
For all $\lambda\in[0,1]$ the capacities $Q$~\cite{holwer, LossyECEAC1, LossyECEAC2}, $Q_2$~\cite{PLOB}, and $K$~\cite{PLOB} of the pure-loss channel $\mathcal{E}_{\lambda,0}$ are given by:
\bb\label{capacities_pure_loss}
    Q(\mathcal{E}_{\lambda,0})&=   
        \begin{cases}
        \log_2\left(\frac{\lambda}{1-\lambda}\right) &\text{if $\lambda\in[\frac{1}{2},1]$ ,} \\
        0 &\text{if $\lambda\in[0,\frac{1}{2}]$ .}
    \end{cases}\\
    Q_2(\mathcal{E}_{\lambda,0})&=K(\mathcal{E}_{\lambda,0})= \log_2\left(\frac{1}{1-\lambda}\right)\,.
\ee
In particular, it holds that
\bb\label{q_equal_zero_sm}
    \lambda\le\frac{1}{2}\Longleftrightarrow Q(\mathcal{E}_{\lambda,0})=0\,.
\ee
For $\lambda\in[0,1]$ and $\nu>0$, the capacities $Q$, $Q_2$, and $K$ of the thermal attenuator $\mathcal{E}_{\lambda,\nu}$ are not known, but bounds have been established: see~\cite{PLOB, Rosati2018, Sharma2018, Noh2019,holwer, Noh2020,fanizza2021estimating} for bounds on $Q(\mathcal{E}_{\lambda,\nu})$, while see~\cite{PLOB,Davis2018,Goodenough16,TGW,holwer,MMMM,squashed_channel,holwer,Pirandola2009,Noh2020,Ottaviani_new_lower,Pirandola18,Wang_Q2_amplifier,lower-bound} for bounds on $Q_2(\mathcal{E}_{\lambda,\nu})$.
Nevertheless, the region of zero two-way quantum capacity and zero secret-key capacity of the thermal attenuator have been determined and it reads~\cite{lower-bound}:
\bb\label{q_equal_zero_sm2}
    \lambda\le \frac{\nu}{\nu+1}\Longleftrightarrow Q_2(\mathcal{E}_{\lambda,\nu})=K(\mathcal{E}_{\lambda,\nu})=0\,.
\ee
Although the exact region of zero quantum capacity of the thermal attenuator has not been determined, there exist known bounds on this region. In particular, it has been established that~\cite{Rosati2018, extendibility,Caruso2006,holwer}
\bb\label{q_equal_zero_sm3}
    \lambda\le \frac{\nu+\frac{1}{2}}{\nu+1}&\Longrightarrow Q(\mathcal{E}_{\lambda,\nu})=0\,,\\
    \lambda>\frac{1}{1+2^{-g(\nu)}} &\Longrightarrow Q(\mathcal{E}_{\lambda,\nu})>0\,,
\ee
where 
\bb
g(x) \coloneqq (x+1)\log_2 (x+1) - x\log_2 x\qquad\forall\,x\ge0 .
\ee

\section{Analysis of the Delocalised Interaction Model}\label{SM_sec_def_model}
In this section, we will conduct a thorough and self-contained examination of the "Delocalised Interaction Model" (DIM), introduced above in Section~\ref{Section_DIM} in order to describe optical fibres with memory effects.

\begin{figure*}[t]
	\centering
	\includegraphics[width=0.9\linewidth]{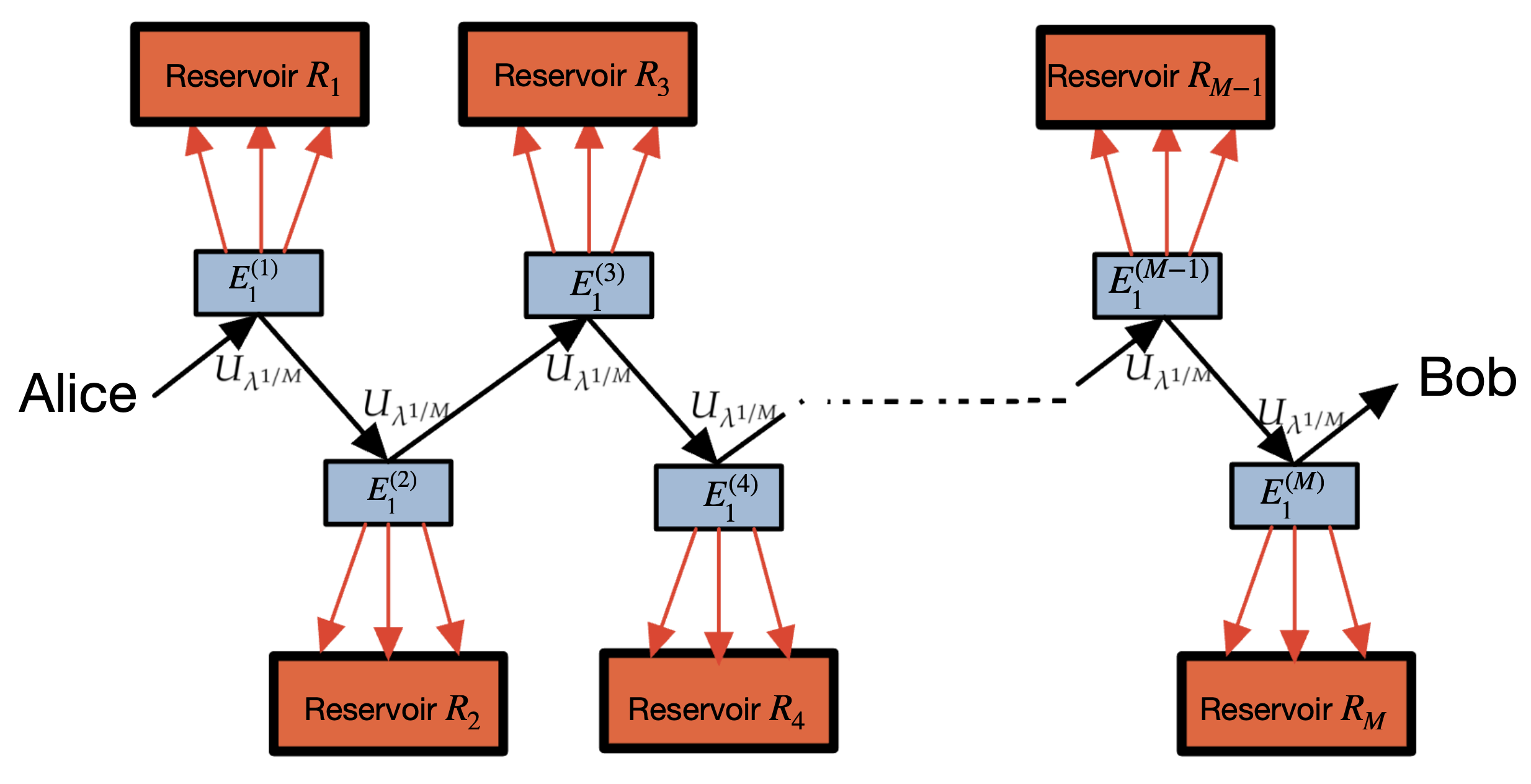}
	\caption{Depiction of our model of optical fibre with memory effects, termed as "Delocalised Interaction Model" (DIM). As Alice's input signal travels across the optical fibre, it interacts with each of the $M$ single-mode environments via the beam splitter unitary $U_{\lambda^{1/M}}$. For each $j = 1, 2, \ldots, M$, the environment $E_1^{(j)}$ undergoes a thermalisation process induced by the reservoir $R_j$, which attempts to restore the state of $E_1^{(j)}$ to its initial thermal state $\tau_\nu$. If the time interval between consecutive input signals is much longer than the thermalisation timescale, the memoryless assumption holds, and our model simplifies to that of a thermal attenuator $\mathcal{E}_{\lambda,\nu}$. However, in the opposite regime, where the time interval is shorter than the thermalisation timescale, each environment does not have sufficient time to return to its initial thermal state between consecutive signals. Consequently, in this case, when an input signal interacts with the environments $\{E^{(j)}_1\}_{j=1,2,\ldots,M}$, the state of each environment is no longer a thermal state $\tau_\nu$, but depends on the previously transmitted input signals, thereby introducing memory effects into the system.}
	\label{fig:new_model}
\end{figure*}

Consider an optical fibre with length $L$ and transmissivity $\lambda$. We can imagine the optical fibre as a composition of $M\in\N$ infinitesimal optical fibres, each with length $L/M$ and transmissivity $\lambda^{1/M}$ with $M\rightarrow\infty$. The signals transmitted by Alice, which propagate through the optical fibre, are single-mode of electromagnetic radiation with a definite frequency and polarisation. To model the noise affecting each signal, we employ the following approach. For all $j=1,2,\ldots,M$, the $j$th infinitesimal optical fibre is represented as a beam splitter with transmissivity $\lambda^{1/M}$ that interacts with a single-mode environment denoted as $E_1^{(j)}$. Additionally, each environment $E_1^{(j)}$ is influenced by both the output signal of the composition of the first $(j-1)$ infinitesimal optical fibres via the beam splitter interaction $U_{\lambda^{1/M}}$, and a remote reservoir $R_j$. The remote reservoir solely interacts with $E_1^{(j)}$ and attempts to reset its state to $\tau_\nu$ through a thermalisation process, as illustrated in Fig.~\ref{fig:new_model}.

\begin{figure*}[t]
	\centering
	\includegraphics[width=0.8\linewidth]{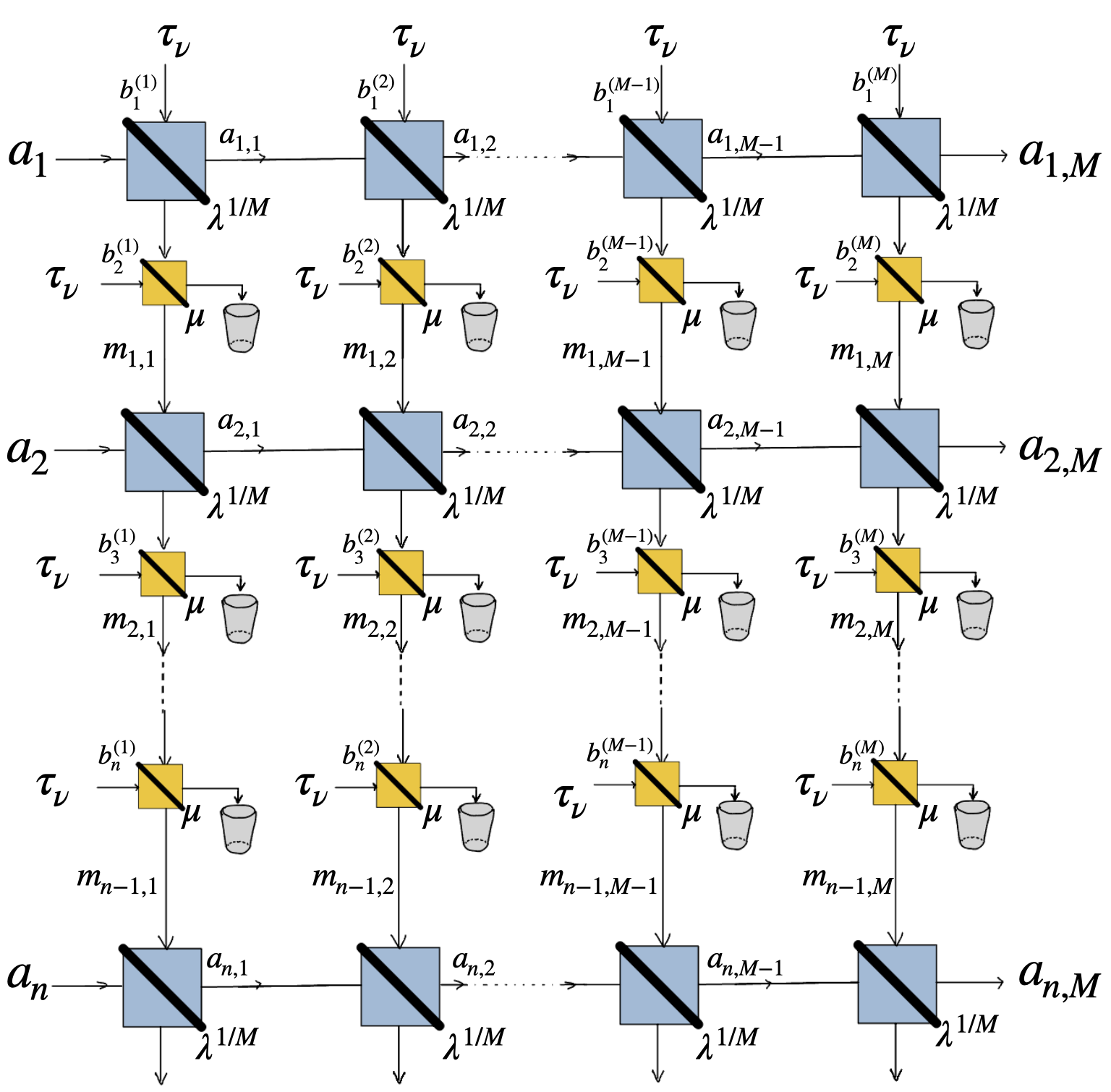}
	\caption{Depiction of the $n$-mode quantum channel $\Phi_{\lambda,\mu,\nu}^{(M,n)}$, which defines the "Delocalised Interaction Model" (DIM). Here, $M\in\N$ is the number of infinitesimal optical fibres, $n\in\N$ is the number of uses of the optical fibre, $\lambda$ is the transmissivity of the optical fibre, $\nu$ is the thermal noise, and $\mu$ is the memory parameter. For all $i=1,2,\ldots, n$, the $i$th row of the figure, which consists of $M$ blue beam splitters of transmissivity $\lambda^{1/M}$, represents the $i$th use of the optical fibre. For all $j=1,2,\ldots,M$, the $j$th column of the figure, which consists of $n$ blue splitters of transmissivity $\lambda^{1/M}$ and $n$ yellow beam splitters of transmissivity $\mu$, represents the $j$th infinitesimal optical fibre. The yellow beam splitters, coupling the environments with the thermal state $\tau_\nu$, represent the thermalisation process $\mathcal{E}_{\mu,\nu}$. If $\mu=0$ then $\Phi_{\lambda,\mu=0,\nu}^{(M,n)}=\mathcal{E}_{\lambda,\nu}^{\otimes n}$, meaning that the memoryless scenario is recovered. In the Heisenberg representation, $\Phi_{\lambda,\mu,\nu}^{(M,n)}$ maps the annihilation operator $a_i$ (depicted on the left) of the $i$th input signal into the annihilation operators $a_{i,M}$ (depicted on the right) of the $i$th output signal for all $i=1,2,\ldots,n$. Our model of optical fibre with memory effects coincides with the \emph{strong} limit $\Phi_{\lambda,\mu,\nu}^{(n)}=\lim\limits_{M\rightarrow\infty}\Phi_{\lambda,\mu,\nu}^{(M,n)}$. }
    \label{fig:griglia_SM}
\end{figure*}

We describe the thermalisation process caused by the reservoir $R_j$ on the environment $E_1^{(j)}$ as a thermal attenuator $\mathcal{E}_{\mu,\nu}$ with transmissivity $\mu\in[0,1]$ and thermal noise $\nu\ge0$. The transmissivity $\mu$ serves as a \emph{memory parameter} that relates to the time interval $\delta t$ between consecutive signals from Alice and the characteristic thermalisation time $t_E$. Specifically, we can define $\mu\coloneqq \exp(-\delta t/t_E)$ as a possible relationship between these quantities.

Let us consider the special case of $\mu=0$ to further enhance our understanding. In this case, the property $\mathcal{E}_{\mu=0,\nu}(\rho)=\tau_\nu$ holds for all single-mode states $\rho$. This means that when an input signal interacts with any environment $E_1^{(j)}$, the state of $E_1^{(j)}$ just before the interaction is a thermal state $\tau_\nu$, regardless of the previously sent input signals.
As a result, in this case, we can represent each infinitesimal optical fibre as a thermal attenuator $\mathcal{E}_{\lambda^{1/M},\nu}$ with transmissivity $\lambda^{1/M}$. Finally, according to the composition rule~\eqref{composition_them}, the composition of $M$ thermal attenuators with transmissivity $\lambda^{1/M}$ is equal to a single thermal attenuator with transmissivity $\lambda$.

On the other hand, when $\mu>0$, the environments $\{E_1^{(j)}\}_{j=1,2,\ldots, M}$ deviate from the initial thermal state $\tau_\nu^{\otimes M}$ at the moment of the interaction with Alice's signal, and their state becomes dependent on the previously transmitted input signals. This indicates that when $\mu>0$, memory effects come into play. The interaction of each input signal with the environments introduces a memory effect, causing the state of the environments to be influenced by the history of the transmitted signals. Thus, the presence of non-zero $\mu$ leads to the emergence of memory effects in the system.

The discussion above leads us to the definition of our model of optical fibre with memory effects. Namely, fixed the transmissivity $\lambda\in[0,1]$, the thermal noise $\nu\ge 0$, and the memory parameter $\mu\in[0,1]$, in our model the output state of $n$ uses of the optical fibre is given by the output of the $n$-mode quantum channel $\Phi_{\lambda,\mu,\nu}^{(M,n)}$, which is defined above in \eqref{concatenated}, in the limit $M\rightarrow\infty$. We describe the channel $\Phi_{\lambda,\mu,\nu}^{(M,n)}$ in the caption of Fig.~\ref{fig:griglia_SM}, illustrating its interferometric representation (see also Fig.~\ref{fig:griglia_main}(c)).
In the forthcoming Theorem~\ref{thm_strong_sm} we show that the sequence of quantum channels $\{\Phi_{\lambda,\mu,\nu}^{(M,n)}\}_{M\in\N}$ strongly converges for $M\to\infty$ to a quantum channel $\Phi_{\lambda,\mu,\nu}^{(n)}$ (see \eqref{strong_convergence} for the definition of strong convergence):
\bb\label{def_strong_limit}
\Phi_{\lambda,\mu,\nu}^{(n)}\coloneqq\lim\limits_{M\rightarrow\infty}\Phi_{\lambda,\mu,\nu}^{(M,n)}\qquad\text{(strongly)}\,,
\ee
and we give an explicit expression for $\Phi_{\lambda,\mu,\nu}^{(n)}$. 
The family of quantum channels $\{\Phi_{\lambda,\mu,\nu}^{(n)}\}_{n\in\N}$ defined in~\eqref{def_strong_limit} is a quantum memory channel~\cite{memory-review}, which completely defines our new model of optical fibre with memory effects, termed as "Delocalised Interaction Model" (DIM). In the DIM, when Alice transmits an $n$-mode state $\rho$ through $n$ uses of the optical fibre, Bob receives $n$ signals in the $n$-mode state $\Phi_{\lambda,\mu,\nu}^{(n)}(\rho)$.

\begin{thm}\label{thm_strong_sm}
{Let $\lambda\in[0,1]$ be the transmissivity, let $\nu\ge 0$ be the thermal noise, let $\mu\in[0,1]$ be the memory parameter, and let $n\in\N$ be the number of channel uses. For any $M\in\N$ let $\Phi_{\lambda,\mu,\nu}^{(M,n)}$ be the $n$-mode quantum channel defined in \eqref{concatenated}. The sequence of quantum channels $\{\Phi_{\lambda,\mu,\nu}^{(M,n)}\}_{M\in\N}$ strongly converges for $M\rightarrow\infty$ to a quantum channel $\Phi_{\lambda,\mu,\nu}^{(n)}$ of the form 
        \bb\label{eq_unitarily_equivalence_sm}
        \Phi_{\lambda,\mu,\nu}^{(n)}=\mathcal{U}_{O_2^\intercal}\circ \left(\bigotimes_{i=1}^{n} \mathcal{E}_{\eta_i^{(n,\lambda,\mu)},\nu}\right)\circ \mathcal{U}_{O_1}\,,
    \ee 
    where:
    \begin{itemize}
        \item $\mathcal{U}_{O_2^\intercal}$ and $\mathcal{U}_{O_1}$ are $n$-mode "passive unitary channels" (see~\eqref{definition_passive unitary} for the definition of passive unitary channels) associated with the $n\times n$ orthogonal matrices $O_2^\intercal$ and $O_1$, which are defined below in \eqref{singular_A_prop}. 
        \item The transmissivities $\{\eta_i^{(n,\lambda,\mu)}\}_{i=1,2,\ldots,n}$ of the thermal attenuators in~\eqref{eq_unitarily_equivalence_sm} and the $n\times n$ orthogonal matrices $O_2^\intercal$ and $O_1$ are defined via the singular value decomposition of the $n\times n$ real matrix $\bar{A}^{(n,\lambda,\mu)}$, which is defined below in \eqref{matrix_sm}, as
        \bb\label{singular_A_prop}
            \bar{A}^{(n,\lambda,\mu)}&=O_2^\intercal D O_1\,,
        \ee
        where $D\coloneqq\text{diag}\left(\sqrt{\eta_1^{(n,\lambda,\mu)}},\sqrt{\eta_2^{(n,\lambda,\mu)}},\ldots, \sqrt{\eta_n^{(n,\lambda,\mu)}}\right)$ is an $n\times n$ positive diagonal matrix with 
        \bb
            1\ge\eta_n^{(n,\lambda,\mu)}\ge \eta_{n-1}^{(n,\lambda,\mu)}\ge \ldots\ge \eta_1^{(n,\lambda,\mu)}\ge0\,.
        \ee
        \item For each $i,k\in\{1,2,\ldots,n\}$, the $(i,k)$ element of the $n\times n$ real matrix $\bar{A}^{(n,\lambda,\mu)}$ is defined as
        \bb\label{matrix_sm}
            \bar{A}^{(n,\lambda,\mu)}_{i,k}= \Theta(i-k)\,\sqrt{\lambda}\mu^{\frac{i-k}{2}} L_{i-k}^{(-1)}(-\ln\lambda)\,.
        \ee
    where $\Theta$ denotes the Heaviside Theta function defined as 
    \bb
\Theta(x)\coloneqq\begin{cases}
1, & \text{if $x\ge0$,} \\
0, & \text{otherwise.}
\end{cases}
    \ee
    and where $\{L_m^{(-1)}\}_{m\in\N}$ are the generalised Laguerre polynomials defined as 
    \bb
        L_m^{(-1)}(x)&\coloneqq \sum_{l=1}^m\binom{m-1}{l-1}\frac{(-x)^l}{l!}\qquad\forall\,m\in\mathbb{N}^+ \,,\\
        L_0^{(-1)}(x)&\coloneqq 1\,,
    \ee for all $x\in\mathbb{R}$.
    \end{itemize}
    That is, for any $n$-mode state $\rho$ the following limit holds:
    \bb\label{strong_limit_prop}
    \lim\limits_{M\rightarrow\infty} & \Big\| \Phi_{\lambda,\mu,\nu}^{(M,n)}(\rho)- \Phi_{\lambda,\mu,\nu}^{(n)}(\rho)  \Big\|_1=0\,,
    \ee
    where $\|\cdot\|_1$ denotes the trace norm. Eq.~\eqref{eq_unitarily_equivalence_sm} establishes that $\Phi_{\lambda,\mu,\nu}^{(n)}$ is unitary equivalent to a tensor product of $n$ distinct thermal attenuators $\bigotimes_{i=1}^{n} \mathcal{E}_{\eta_i^{(n,\lambda,\mu)},\nu}$. In particular, the capacities of the quantum memory channels $\{\Phi_{\lambda,\mu,\nu}^{(n)}\}_{n\in\N}$ and $\{\bigotimes_{i=1}^{n} \mathcal{E}_{\eta_i^{(n,\lambda,\mu)},\nu}\}_{n\in\N}$ are equal.}
\end{thm}

Before presenting the proof of Theorem~\ref{thm_strong_sm}, let us understand its meaning.  Suppose Alice prepares an $n$-mode state $\rho$, applies the passive transformation $\mathcal{U}_{O_1^\intercal}$, and then sends each of the $n$ signals through the optical fibre. Suppose further that, after receiving the signals, Bob applies the passive transformation $\mathcal{U}_{O_2}$. Consequently, Bob receives the state $\mathcal{U}_{O_2}\circ \Phi_{\lambda,\mu,\nu}^{(n)}\circ\mathcal{U}_{O_1^\intercal}(\rho)$, which can be expressed as
\bb
\mathcal{U}_{O_2}\circ \Phi_{\lambda,\mu,\nu}^{(n)}\circ\mathcal{U}_{O_1^\intercal}(\rho)=\bigotimes_{i=1}^n\mathcal{E}_{\eta_i^{(n,\lambda,\mu)},\nu}(\rho)\,,
\ee
thanks to Theorem~\ref{thm_strong_sm}.
Therefore, by employing appropriate encoding and decoding passive unitary transformations, Alice and Bob can effectively communicate via 
\bb\label{tensor_product_memory}
    \left\{\bigotimes_{i=1}^n \mathcal{E}_{\eta_i^{(n,\lambda,\mu)},\nu} \right\}_{n\in\mathbb{N}}\,,
\ee 
instead of via $\{ \Phi_{\lambda,\mu,\nu}^{(n)}\}_{n\in\mathbb{N}}$. To summarise, the quantum memory channel $\Phi_{\lambda,\mu,\nu}^{(n)}$, which models $n$ uses of the optical fibre, is unitary equivalent to a tensor product of $n$ distinct thermal attenuators. These attenuators have the same thermal noise $\nu$ but differ in their transmissivities. This equivalence allows us to understand the behaviour of the quantum memory channel in terms of individual thermal attenuators operating on each signal independently. The transmissivities $\{\eta_i^{(n,\lambda,\mu)}\}_{i=1,2,\ldots,n}$ can be calculated via the singular value decomposition of the matrix reported in~\eqref{matrix_sm}. Notably, each element of this matrix depends solely on the difference between the row and column indices, making it a \emph{Toeplitz matrix}. In Section~\ref{section_toeplitz} we will review important properties of Toeplitz matrices. Fortunately, these matrices allow for the calculation of the asymptotic distribution of their singular values as the dimension approaches infinity. Leveraging this observation, we will be able to analytically determine the transmissivities $\{\eta_i^{(n,\lambda,\mu)}\}_{i=1,2,\ldots,n}$ as the number of uses of the optical fibre, $n$, approaches infinity.

Now, we are ready to proceed with the proof of Theorem~\ref{thm_strong_sm}, which generalises the methods used in~\cite{Memory1,Memory2,Memory3}.

\begin{proof}[Proof of Theorem~\ref{thm_strong_sm}]
    { The structure of the proof is as follows:
     \begin{itemize}
         \item \textbf{Part 1)} First, we observe that, since the interferometric representation of the channel $\Phi_{\lambda,\mu,\nu}^{(M,n)}$ involves only beam splitters (see Fig.~\ref{fig:griglia_SM}), the output annihilation operators $\{a_{i,M}\}_{i=1,2,\ldots,n}$ can be expressed as a linear combination of the input annihilation operators $\{a_{i}\}_{i=1,2,\ldots,n}$, without any creation operators.
         \item \textbf{Part 2)} Second, we show that such a property implies that the sequence $\{\Phi_{\lambda,\mu,\nu}^{(M,n)}\}_{M\in\N}$ strongly converges to a quantum channel $\Phi_{\lambda,\mu,\nu}^{(n)}$ that is unitary equivalent to a tensor product of $n$ distinct thermal attenuators. 
         \item \textbf{Part 3)} Finally, in order to explicitly find the transmissivities of such thermal attenuators and the unitaries that achieve such unitary equivalence, we explicitly calculate the coefficients of the linear combination of the input annihilation operators in the expression of the output annihilation operators.
     \end{itemize} 
     We adopt the notation introduced in Fig.~\ref{fig:griglia_SM} and in Fig.~\ref{fig:griglia_main}(c). Specifically, $n\in\N$ represents the number of uses of the optical fibre, while $M\in\N$ represents the number of infinitesimal optical fibres. Moreover, for each $i=1,2,\ldots,n$ the annihilation operator associated with the $i$th input signal is denoted as $a_i$, and for each $j=1,2,\ldots,M$ the annihilation operator associated with the $j$th single-mode environment $E_1^{(j)}$ is denoted as $b_1^{(j)}$. In addition, as shown in Fig.~\ref{fig:griglia_main}(c), for each $i=2,3,\ldots n$ and $j=1,2,\ldots,M$ we denote as $E_{i}^{(j)}$ the single-mode environment associated with the thermal attenuator $\mathcal{E}_{\mu,\nu}$ that represents the thermalisation process acting on the $j$th environment $E_1^{(j)}$ right before the $i$th use of the fibre. As shown in Fig.~\ref{fig:griglia_SM}, the annihilation operator of the environment $E_{i}^{(j)}$ is denoted as $b_i^{(j)}$.  
}

\textbf{Part 1)} Let us write the quantum channel $\Phi_{\lambda,\mu,\nu}^{(M,n)}$ as
\bb\label{phys_rep}
    \Phi_{\lambda,\mu,\nu}^{(M,n)}(\cdot)=\Tr_{E}\left[U_{\lambda,\mu}^{(n,M)}\left((\cdot)\otimes \tau_\nu^{\otimes nM}\right){U_{\lambda,\mu}^{(n,M)}}^\dagger\right]\,,
\ee
where $\Tr_{E}$ denotes the partial trace over all the environment systems ($E_{i}^{(j)}$ for all $i=1,2,\ldots,n$ and all $j=1,2,\ldots,M$) and where $U_{\lambda,\mu}^{(n,M)}$ is the unitary whose interferometric representation is depicted in Fig.~\ref{fig:griglia_SM}. In addition, as depicted in Fig.~\ref{fig:griglia_SM}, for each $i=1,2,\ldots,n$ let 
\bb\label{def_a_i_M}
    a_{i,M}\coloneqq {U_{\lambda,\mu}^{(n,M)}}^\dagger a_{i} U_{\lambda,\mu}^{(n,M)}
\ee
be the output of the $i$th annihilation operator in Heisenberg representation. Crucially, since the unitary $U_{\lambda,\mu}^{(n,M)}$ is composed of only beam splitters (see Fig.~\ref{fig:griglia_SM}) and since the beam splitter transformations in~\eqref{transf_beam} map annihilation operators into annihilation operators (without introducing any creation operator), it follows that each output annihilation operator $a_{i,M}$ can be expressed as a linear combination of input annihilation operators and environmental annihilation operators. In formula,
\bb\label{heisen_annihil_output}
    a_{i,M}=\sum_{h=1}^n  A^{(M,n,\lambda,\mu)}_{i,h}\,a_h+\sum_{l=1}^{M}\sum_{h=1}^{n}E^{(M,n,\lambda,\mu)}_{l,i,h}\,b_{h}^{(l)}\,,
\ee
where $ A^{(M,n,\lambda,\mu)}_{i,h}$ and $E^{(M,n,\lambda,\mu)}_{l,i,h}$ are suitable real numbers satisfying
\bb\label{comm_rel}
[a_{i,M},(a_{j,M})^{\dagger}]=\delta_{i,j}\mathbb{1}\,,
\ee
where the latter follows from the definition of $a_{i,M}$ in \eqref{def_a_i_M}.

\textbf{Part 2)} By defining the following vectors of annihilation operators
\bb
\mathbf{a}&\coloneqq(a_1,\,a_2,\ldots,\,a_n)^\intercal\,,\\ \mathbf{a_M}&\coloneqq(a_{1,M},\,a_{2,M},\ldots,\,a_{n,M})^\intercal\,,\\
\mathbf{b}&\coloneqq (b_{1}^{(1)},\,b_{1}^{(2)},\ldots,\, b_{1}^{(M)},\, b_{2}^{(1)},\,b_{2}^{(2)},\ldots,\,b_{2}^{(M)},\\
&\qquad\ldots,\,b_{n}^{(1)},\,b_{n}^{(2)},\ldots,\,b_{n}^{(M)})^\intercal\,,
\ee
one can rewrite~\eqref{heisen_annihil_output} as
\bb\label{vector_annihilation}
    \mathbf{a_M}=A^{(M,n,\lambda,\mu)}\mathbf{a}+E^{(M,n,\lambda,\mu)}\mathbf{b}\,,
\ee
where $A^{(M,n,\lambda,\mu)}$ is an $n\times n$ real matrix and $E^{(M,n,\lambda,\mu)}$ is an $n\times nM$ real matrix which satisfy
\bb\label{implication_comm}
    A^{(M,n,\lambda,\mu)}{A^{(M,n,\lambda,\mu)}}^\intercal+E^{(M,n,\lambda,\mu)}{E^{(M,n,\lambda,\mu)}}^\intercal=\mathbb{1}_{n\times n}\,,
\ee
thanks to~\eqref{comm_rel}. Here, $\mathbb{1}_{n\times n}$ denotes the $n\times n$ identity matrix. In the Part 3) of the proof we will show that the limit of each matrix element $\lim\limits_{M\rightarrow\infty}A_{i,h}^{(M,n,\lambda,\mu)}$ does exist and it is precisely given by~\eqref{matrix_sm}. Let $\bar{A}^{(n,\lambda,\mu)}$ be the $n\times n$ matrix whose $(i,h)$ element is given by
\bb\label{lim_elements_def}
    \bar{A}_{i,h}^{(n,\lambda,\mu)}\coloneqq \lim\limits_{M\rightarrow\infty}A_{i,h}^{(M,n,\lambda,\mu)}\,.
\ee
Now, we are going to prove the validity of~\eqref{strong_limit_prop}. For this purpose, let $\rho$ be an $n$-mode input state. The characteristic function of $\Phi_{\lambda,\mu,\nu}^{(M,n)}(\rho)$ can be calculated as
\bb
&\chi_{\Phi_{\lambda,\mu,\nu}^{(M,n)}(\rho)}(\mathbf{z})\coloneqq \Tr\left[\Phi_{\lambda,\mu,\nu}^{(M,n)}(\rho)\,e^{\mathbf{a}^\dagger\mathbf{z}-\mathbf{z}^\dagger\mathbf{a} }\right]\\&\eqt{(i)} \chi_{\rho}\left({A^{(M,n,\lambda,\mu)}}^\intercal\mathbf{z}\right)\chi_{\tau_\nu^{\otimes nM}}\left({E^{(M,n,\lambda,\mu)}}^\intercal\mathbf{z}\right)\\&\eqt{(ii)}\chi_{\rho}\left({A^{(M,n,\lambda,\mu)}}^\intercal\mathbf{z}\right)e^{-(\nu+\frac{1}{2})\mathbf{z}^\dagger E^{(M,n,\lambda,\mu)}{E^{(M,n,\lambda,\mu)}}^\intercal \mathbf{z}}\\&\eqt{(iii)}\chi_{\rho}\left({A^{(M,n,\lambda,\mu)}}^\intercal\mathbf{z}\right)e^{-(\nu+\frac{1}{2})\mathbf{z}^\dagger ( \mathbb{1}_{n\times n}- A^{(M,n,\lambda,\mu)}{A^{(M,n,\lambda,\mu)}}^\intercal ) \mathbf{z}}
\ee
for all $\mathbf{z}\in\mathbb{C}^{n}$, where: in (i) we exploited~\eqref{phys_rep} and~\eqref{vector_annihilation}; in (ii) we used the expression in~\eqref{charact_therm} of the characteristic function of the thermal state $\tau_\nu$; in (iii) we just exploited~\eqref{implication_comm}. Since $\lim\limits_{M\rightarrow\infty} A^{(M,n,\lambda,\mu)}\mathbf{z} = \bar{A}^{(n,\lambda,\mu)}\mathbf{z}$ (thanks to~\eqref{lim_elements_def}) and since every characteristic function is continuous~\cite{BUCCO,HOLEVO}, it holds that
\bb\label{lim_char}
&\lim\limits_{M\rightarrow\infty}\chi_{\Phi_{\lambda,\mu,\nu}^{(M,n)}(\rho)}(\mathbf{z}) \\
&=\chi_{\rho}\left(\bar{A}^{(n,\lambda,\mu)^\intercal}\mathbf{z}\right)e^{-(\nu+\frac{1}{2})\mathbf{z}^\dagger ( \mathbb{1}_{n\times n}- \bar{A}^{(n,\lambda,\mu)}\bar{A}^{(n,\lambda,\mu)^\intercal} ) \mathbf{z}}\,.
\ee
By applying the singular value decomposition, $\bar{A}^{(n,\lambda,\mu)}$ can be written as
\bb\label{singular_A}
    \bar{A}^{(n,\lambda,\mu)}&=O_2^\intercal D O_1\,,
\ee
where $O_1$ and $O_2$ are $n\times n$ real orthogonal matrices, $D=\text{diag}(s_1^{(n,\lambda,\mu)},s_2^{(n,\lambda,\mu)},\ldots, s_n^{(n,\lambda,\mu)})$ is an $n\times n$ positive diagonal matrix with 
\bb
    1\ge s_n^{(n,\lambda,\mu)}\ge s_{n-1}^{(n,\lambda,\mu)}\ge \ldots\ge s_1^{(n,\lambda,\mu)}\ge0\,,
\ee
where the fact that $s_n^{(n,\lambda,\mu)}\le1$ follows from~\eqref{implication_comm}. Let us define for all $i=1,2,\ldots,n$ the $i$th \emph{transmissivity} $\eta_i^{(n,\lambda,\mu)}$ as the square of the $i$th diagonal element of $D$, i.e.
\bb\label{effect_transm_discrete}
\eta_i^{(n,\lambda,\mu)}\coloneqq (s_i^{(n,\lambda,\mu)})^2\,.
\ee
Note that $\eta_i^{(n,\lambda,\mu)}$ is a proper transmissivity because $0\le(s_i^{(n,\lambda,\mu)})^2\le 1$. In addition, for all $\mathbf{z}\in\mathbb{C}^n$ it holds that
\bb\label{pointwise_convergence}
     \lim\limits_{M\rightarrow\infty}\chi_{\Phi_{\lambda,\mu,\nu}^{(M,n)}(\rho)}(\mathbf{z}) 
    &\eqt{(i)}\chi_{\rho}\left(   O_1^\intercal D O_2   \mathbf{z} \right)e^{-(\nu+\frac{1}{2})\mathbf{z}^\dagger ( \mathbb{1}_{n\times n}- O_2^\intercal D^2 O_2 ) \mathbf{z}}\\
    &\eqt{(ii)}  \chi_{\mathcal{U}_{O_1}(\rho)}\left( D O_2  \mathbf{z} \right)e^{-(\nu+\frac{1}{2})\mathbf{z}^\dagger O_2^\intercal ( \mathbb{1}_{n\times n}-  D^2  ) O_2\mathbf{z}}\\
    &\eqt{(iii)}  \chi_{\mathcal{U}_{O_1}(\rho)}\left( D O_2  \mathbf{z} \right)\chi_{\tau_\nu^{\otimes n}}(\sqrt{\mathbb{1}_{n\times n}-  D^2 } O_2\mathbf{z} )\\
    &\eqt{(iv)}  \chi_{(\bigotimes_{i=1}^n\mathcal{E}_{\eta_i^{(n,\lambda,\mu)},\nu})\circ\mathcal{U}_{O_1}(\rho)}\left(  O_2  \mathbf{z} \right)\\
    &\eqt{(v)}\chi_{\mathcal{U}_{O_2^\intercal}\circ\left(\bigotimes_{i=1}^n\mathcal{E}_{\eta_i^{(n,\lambda,\mu)},\nu}\right)\circ\mathcal{U}_{O_1}(\rho)}\left(   \mathbf{z} \right)\,\,,
\ee
where: in (i) we used~\eqref{lim_char} and~\eqref{singular_A}; in (ii) we exploited~\eqref{passive_transf_charact}; (iii) is a consequence of the expression in~\eqref{charact_therm} of the characteristic function of the thermal state $\tau_\nu$; (iv) comes from the expression in~\eqref{caract_att} of the characteristic function of the output state of the thermal attenuator; in (v) we exploited again~\eqref{passive_transf_charact}. Hence,~\eqref{pointwise_convergence} establishes that the characteristic function of $\{\Phi_{\lambda,\mu,\nu}^{(M,n)}(\rho)\}_{M\in\N}$ converges pointwise for $M\rightarrow\infty$ to the characteristic function $\mathcal{U}_{O_2^\intercal}\circ\left(\bigotimes_{i=1}^n\mathcal{E}_{\eta_i^{(n,\lambda,\mu)},\nu}\right)\circ\mathcal{U}_{O_1}(\rho)$. Consequently, Lemma~\eqref{lemmino_charact} implies that  
\bb
    &\lim\limits_{M\rightarrow\infty}\Big\| \Phi_{\lambda,\mu,\nu}^{(M,n)}(\rho)- \mathcal{U}_{O_2^\intercal}\circ\left(\bigotimes_{i=1}^n\mathcal{E}_{\eta_i^{(n,\lambda,\mu)},\nu}\right)\circ\mathcal{U}_{O_1}(\rho)\Big\|_1=0\,,
\ee
i.e.~we have proved~\eqref{strong_limit_prop}.

\textbf{Part 3)} Now, we only need to show that the limit $\lim\limits_{M\rightarrow\infty}A_{i,h}^{(M,n,\lambda,\mu)}$ does exist and it is given by~\eqref{matrix_sm} for all $i,h=1,2,\ldots n$. To proceed, we begin by calculating the elements of the matrix $A^{(M,n,\lambda,\mu)}$ defined in~\eqref{vector_annihilation}. Our goal is thus to express $a_{i,M}$ in terms of $a_{1},a_2,\ldots,a_n$.

In the following derivation, $m_{i,j}$ and $a_{i,j}$ denote the annihilation operators in Heisenberg representation introduced in Fig.~\ref{fig:griglia_SM}. Specifically, $m_{i,j}$ denotes the Heisenberg evolution of the annihilation operator $b_1^{(j)}$ right before the $(i+1)$th use of the fibre. Moreover, $a_{i,j}$ denotes the Heisenberg evolution of the annihilation operator $a_i$ of the $i$th input signal right after its transmission through the first $j$ beam splitters of transmissivity $\lambda^{1/M}$. Additionally, we use the notations $a_{i,0}\coloneqq a_{i}$ and $m_{0,i}\coloneqq b^{(i)}_1$. 
By exploiting~\eqref{transf_beam}, one deduces that for all $i=1,2,\ldots,n$ and all $j=1,2,\ldots,M$ it holds that
\bb\label{eq_a_m}
a_{i,j}&=\sqrt{\lambda^{1/M}}\,a_{i,j-1}+\sqrt{1-\lambda^{1/M}}\,m_{i-1,j}\,,\\
m_{i,j}&=-\sqrt{1-\mu}\,b_{i+1}^{(j)}+\sqrt{\mu\lambda^{1/M}}\,m_{i-1,j} -\sqrt{\mu(1-\lambda^{1/M})}\,a_{i,j-1}\,.
\ee
Consequently, $a_{i,M}$ can be expressed as a linear combination of the annihilation operators $\{a_i\}_i$ and $\{b_i^{(j)}\}_{i,j}$.  In order to calculate the matrix $A^{(M,n,\lambda,\mu)}$, we only need to calculate the coefficients of $\{a_i\}_i$ in the expression of $a_{i,M}$, while we do not need to calculate the coefficients of $\{b_{i}^{(j)}\}_{i,j}$. Hence, we introduce the relation $A\simeq B$ between operators as follows:
\bb
 A\simeq B\,\Longleftrightarrow \, A=B+\text{linear combination of $\{b_{i}^{(j)}\}_{i,j}$. }
\ee
Note that~\eqref{eq_a_m} implies that for all $i=1,2,\ldots,n$ and all $j=1,2,\ldots,M$ it holds that
\bb\label{a_transf}
    a_{i,j}&=\sqrt{\lambda^{j/M}}\,a_i+\sqrt{1-\lambda^{1/M}}\sum_{l=1}^j \lambda^{\frac{j-l}{2M}}\,m_{i-1,l}
\ee
and
\bb\label{m_recursive}
    m_{i,j}&\simeq  \sqrt{\mu\lambda^{1/M}}m_{i-1,j}-\sqrt{\mu(1-\lambda^{1/M})}a_{i,j-1}\,\\
           &=\sqrt{\mu\lambda^{1/M}}m_{i-1,j}-\sqrt{\mu(1-\lambda^{1/M})\lambda^{(j-1)/M}}a_i \\
           &\quad-\sqrt{\mu}(1-\lambda^{1/M})\sum_{l=1}^{j-1} \lambda^{\frac{j-1-l}{2M}}m_{i-1,l}\,.
\ee
    From~\eqref{a_transf} we deduce that in order to express $a_{i,M}$ in terms of $\{a_i\}_i$, we now only need to express $\{m_{i,l}\}_{i,l}$ in terms of $\{a_i\}_i$ by solving the recurrence relation in~\eqref{m_recursive}. One can show that the latter implies that for all $i=1,2,\ldots,n$ and all $j=2,3\ldots,M$ it holds that
\bb
    m_{i,1}&\simeq -\sqrt{\mu(1-\lambda^{1/M})}\sum_{l=1}^i\left({\mu\lambda^{1/M}}\right)^{\frac{i-l}{2}}\,a_l\,,\\
    m_{i,j}&\simeq \sqrt{\mu\lambda^{1/M}}\,m_{i-1,j}-\sqrt{\mu}\,m_{i-1,j-1}+\sqrt{\lambda^{1/M}}\,m_{i,j-1}\,.
\ee
As a consequence of this, substituting terms $m_{l,j}$ with $l\leq i$, through repeated application of the second equality, and finally dropping the term $m_{0,j}$, one can show that 
\bb\label{eq_m0}
    m_{i,j}\simeq \lambda^{\frac{1}{2M}}\,m_{i,j-1}+(\lambda^{\frac{1}{2M}}-\lambda^{-\frac{1}{2M}})\sum_{l=1}^{i-1}(\mu\lambda^{1/M})^{\frac{i-l}{2}}\,m_{l,j-1}\,
\ee
We stress that, here and in the following, summations are zero whenever the lower extreme is strictly larger than the upper extreme.
By defining $\tilde{m}_{i,j}\coloneqq \lambda^{-\frac{j}{2M}}(\mu\lambda^{1/M})^{-\frac{i}{2}}\,m_{i,j}$, the recurrence relation in~\eqref{eq_m0} can be rewritten as
\bb\label{eq_m1}
    \tilde{m}_{i,j}&\simeq \tilde{m}_{i,j-1}+(1-\lambda^{-\frac{1}{M}})\sum_{l=1}^{i-1}\tilde{m}_{l,j-1}\,.
\ee
In order to solve this, let us introduce the following notation:
\bb
    \mathcal{F}_{i,n}\left(\{\tilde{m}_{ \cdot ,1}\}\right)\coloneqq\sum_{i_n=1}^{i-1}\sum_{i_{n-1}=1}^{i_n-1}\dots\sum_{i_1=1}^{i_2-1}\tilde{m}_{i_1,1}\,.
\ee
By noting that the cardinality of the set $\{(i_2,\ldots,i_n)\in\mathbb{N}^{n-1}\,:\,i>i_n>i_{n-1}>\ldots >i_2>d\}$ is $\binom{i-d-1}{n-1}$, we deduce that
\bb
    \mathcal{F}_{i,n}\left(\{\tilde{m}_{\cdot,1}\}\right)= \sum_{d=1}^{i-n}\binom{i-d-1}{n-1}\tilde{m}_{d,1}\,.
\ee
In addition, note that
\bb
\mathcal{F}_{i,n}\left(\{\tilde{m}_{\cdot,1}\}\right)&=0\qquad\text{if }i\le n\,,
\ee
and 
\bb
    \sum_{l=1}^{i-1}\mathcal{F}_{l,n}\left(\{\tilde{m}_{\cdot,1}\}\right)&=\mathcal{F}_{i,n+1}\left(\{\tilde{m}_{\cdot,1}\}\right)\,.
\ee
As a consequence, by exploiting~\eqref{eq_m1}, one can show by induction that for all $i=1,2,\ldots,n$ and all $j=1,2,\ldots, M$ it holds that
\bb
    \tilde{m}_{i,j}\simeq  \tilde{m}_{i,1}+\sum_{l=1}^{j-1}\binom{j-1}{l}(1-\lambda^{-1/M})^l\mathcal{F}_{i,l}\left(\{\tilde{m}_{\cdot,1}\}\right)\,,
\ee
and from this one can obtain
\bb
    m_{i,j}&\simeq\sum_{k=1}^i \mu^{\frac{i-k+1}{2}}\lambda^{\frac{j+i-k+1}{2M}}\sqrt{1-\lambda^{1/M}} \\
    &\qquad\left[\sum_{l=0}^{\min(i-k,j-1)}\binom{j-1}{l}\binom{i-k}{l}(1-\lambda^{-1/M})^{l}\right]a_k
\ee
Consequently,~\eqref{a_transf} implies that
\bb
    &a_{i,M}= \sqrt{\lambda}a_i+\sqrt{1-\lambda^{1/M}}\sum_{j=1}^M \lambda^{\frac{M-j}{2M}}m_{i-1,j}\\&\simeq \sum_{k=1}^n\Big[\sqrt{\lambda}\delta_{i,k}-\Theta(i-k-1)\,\mu^{\frac{i-k}{2}}\lambda^{\frac{M+i-k}{2M}}(1-\lambda^{1/M})\\
    &\qquad\times\sum_{j=1}^M\sum_{l=0}^{\min(i-k-1,j-1)}\binom{j-1}{l}\binom{i-k-1}{l}(1-\lambda^{-1/M})^{l}\Big]a_k,
\ee
where we have introduced the Heaviside function $\Theta(x)$ defined as $\Theta(x)=1$ if $x\ge 0$, and $\Theta(x)=0$ if $x<0$. Hence, we deduce that for all $i,k\in\{1,2,\ldots,n\}$ the $(i,k)$ matrix element of $A^{(M,n,\lambda,\mu)}$ is
\bb
   &A^{(M,n,\lambda,\mu)}_{i,k}= \sqrt{\lambda}\delta_{i,k}-\Theta(i-k-1)\,\mu^{\frac{i-k}{2}}\lambda^{\frac{M+i-k}{2M}}(1-\lambda^{1/M})\\
    &\qquad\times\sum_{j=1}^M\sum_{l=0}^{\min(i-k-1,j-1)}\binom{j-1}{l}\binom{i-k-1}{l}(1-\lambda^{-1/M})^{l}.
\ee
The sum over $j$ can be split in two parts, obtaining
\bb
   &(1-\lambda^{1/M})\sum_{j=1}^M\sum_{l=0}^{\min(i-k-1,j-1)}\binom{j-1}{l}\binom{i-k-1}{l}(1-\lambda^{-1/M})^{l}
   \\&=(1-\lambda^{1/M})\sum_{j=1}^{i-k}\sum_{l=0}^{j-1}\binom{j-1}{l}\binom{i-k-1}{l}(1-\lambda^{-1/M})^{l}\\
   &\,+(1-\lambda^{1/M})\sum_{j=i-k+1}^M\sum_{l=0}^{i-k-1}\binom{j-1}{l}\binom{i-k-1}{l}(1-\lambda^{-1/M})^{l}.
\ee

The limit for $M\rightarrow \infty$ of the first piece is zero, while for the second we use that

\bb
  \lim_{M\rightarrow \infty}(1-\lambda^{1/M})\sum_{j=i-k+1}^{M} \binom{j-1}{l}(1-\lambda^{-1/M})^{l}=-\frac{(\ln\lambda)^{l+1}}{(l+1)!}
\ee
to obtain
\bb\label{eq_m0_3}
    &\bar{A}^{(n,\lambda,\mu)}_{i,k}\coloneqq \lim\limits_{M\rightarrow\infty}A^{(M,n,\lambda,\mu)}_{i,k}\\
    &=\sqrt{\lambda}\left[\delta_{i,k}+\Theta(i-k-1)\,\mu^{\frac{i-k}{2}}\sum_{l=1}^{i-k}\binom{i-k-1}{l-1}\frac{(\ln\lambda)^l}{l!}\right].
\ee
By defining the generalised Laguerre polynomial $L_m^{(-1)}(x)\coloneqq \sum_{l=1}^m\binom{m-1}{l-1}\frac{(-x)^l}{l!}$\ for all $x\in\mathbb{R}$ and $m\in\mathbb{N}^+ $ and $L_0^{(-1)}(x)\coloneqq 1$, for all $i,k\in\{1,2,\ldots,n\}$ the $(i,k)$ matrix element of $\bar{A}^{(n,\lambda,\mu)}$ in~\eqref{eq_m0_3} can be rewritten as
\bb
    \bar{A}^{(n,\lambda,\mu)}_{i,k}=  \Theta(i-k)\,\sqrt{\lambda}\mu^{\frac{i-k}{2}} L_{i-k}^{(-1)}(-\ln\lambda)  \,.
\ee
\end{proof}


\section{Problems of the "Localised" interaction model solved by the "Delocalised" one}~\label{sec_solution_property_sm}
Memory effects in optical fibres have been studied in~\cite{Memory1,Memory2,Memory3,Die-Hard-2-PRA,Die-Hard-2-PRL}. As explained in Section~\ref{sec_local_interaction_model}, the model examined in~\cite{Memory1,Memory2,Memory3,Die-Hard-2-PRA,Die-Hard-2-PRL} corresponds to a specific case of the DIM we have presented here, specifically the case in which we set $M=1$ (see Section~\ref{SM_sec_def_model} for the meaning of $M$). This corresponds to the configuration depicted in Fig.~\ref{fig:griglia_main}(b), where the interaction between the signal and the fibre is localised. For this reason, we termed the model analysed in~\cite{Memory1,Memory2,Memory3,Die-Hard-2-PRA,Die-Hard-2-PRL} as "Localised Interaction Model" (LIM). The memory quantum channel characterising the LIM is $\{\Phi_{\lambda,\mu,\nu}^{(1,n)}\}_{n\in\N}$ (see Section~\ref{SM_sec_def_model} for its definition).

However, it is important to note that the LIM is not realistic because it fails to satisfy the following two essential properties that any "reasonable" model of an optical fibre should possess:
\begin{itemize}
    \item \emph{Property 1}: When the transmissivity $\lambda$ is exactly equal to zero, no information can be transmitted;
    \item \emph{Property 2}: The model is consistent under concatenation of optical fibres. In other words, the concatenation of $k$ optical fibres, each with transmissivity $\lambda_1,\lambda_2,\ldots,\lambda_k$ and identical thermal noise $\nu$, results in another optical fibre with a transmissivity equal to the product of the individual transmissivities $\lambda_1\lambda_2\ldots\lambda_k$, and with the same thermal noise $\nu$.
\end{itemize}
\begin{remark}
The LIM, which is characterised by the quantum memory channel $\{\Phi_{\lambda,\mu,\nu}^{(1,n)}\}_{n\in\N}$, fails to satisfy Property 1.
\end{remark}
\begin{proof}
    The LIM is associated with the channel  $\{\Phi_{\lambda,\mu,\nu}^{(1,n)}\}_{n\in\N}$ that transmits information even if $\lambda$ is exactly zero. Indeed, for $\lambda=0$ the channel $\Phi_{0,\mu,\nu}^{(1,n)}$ can be expressed in terms of the \emph{quantum shift channel of order $1$}, the thermal attenuator $\mathcal{E}_{\mu,\nu}$, and the unitary phase space inversion operation $\mathcal{V}(\cdot)\coloneqq (-1)^{a^\dagger a}\cdot(-1)^{a^\dagger a}$. Specifically, if the state of the $i$th input signal is $\rho$, then the state of the $(i+1)$th output signal is $\mathcal{V}\circ\mathcal{E}_{\mu,\nu}(\rho)$, which is not constant with respect to $\rho$ for $\mu>0$. 
\end{proof}
In particular, from the above proof we conclude that the LIM would lead to the paradoxical result that an optical fibre with zero transmissivity and a memory parameter $\mu>0$ is equivalent to a memoryless optical fibre with transmissivity $\mu>0$ with a certain time-delay $\delta t$ (e.g.~$\delta t$ can be determined in terms of $\mu$ by $\mu = \exp(-\delta t/t_E)$), provided the output of the first signal is discarded.

Furthermore, the DIM with a finite value of $M$ is also unrealistic because it fails to satisfy Property 1, as we now demonstrate.
\begin{remark}
The DIM with a finite value of $M$, which is characterised by the quantum memory channel $\{\Phi_{\lambda,\mu,\nu}^{(M,n)}\}_{n\in\N}$, does not satisfy Property 1.
\end{remark}
\begin{proof}
    The model with a finite value of $M$ is associated with the channel $\{\Phi_{\lambda,\mu,\nu}^{(M,n)}\}_{n\in\N}$ that transmits information even if $\lambda$ is exactly zero. Indeed, for $\lambda=0$ the channel $\Phi_{0,\mu,\nu}^{(M,n)}$ can be expressed in terms of the \emph{quantum shift channel of order $M$}, the thermal attenuator $\mathcal{E}_{\mu,\nu}$, and the unitary phase space inversion operation $\mathcal{V}^{M\!\!\mod 2}$. Specifically, if the state of the $i$th input signal is $\rho$, then the state of the $(i+M)$th output signal is $\mathcal{V}^{M\!\!\mod 2}\circ\mathcal{E}_{\mu,\nu}(\rho)$, which is not constant with respect to $\rho$ for $\mu>0$. 
\end{proof}
From the above proof we conclude that the model with a finite value of $M$ would lead to the paradoxical result that an optical fibre with zero transmissivity and a memory parameter $\mu>0$ is equivalent to a memoryless optical fibre with transmissivity $\mu>0$ with a time-delay $M\,\delta t$, provided the output of the first $M$ signals is discarded.

Furthermore, the DIM with a finite value of $M$ also fails to satisfy Property 2, meaning that it is not consistent under the composition of optical fibres. This can be proved as follows.
\begin{remark}
    The DIM with a finite value of $M$, which is characterised by the quantum memory channel $\{\Phi_{\lambda,\mu,\nu}^{(M,n)}\}_{n\in\N}$, does not satisfy Property 2.
\end{remark}
\begin{proof}
    One can show that $\Phi_{\lambda_1,\mu,\nu}^{(M,n)}\circ\Phi_{\lambda_2,\mu,\nu}^{(M,n)}$ is not equal to $\Phi_{\lambda_1\lambda_2,\mu,\nu}^{(M,n)}$ in general. Indeed, if $\lambda_1=\lambda_2=0$, then the channel $\Phi_{0,\mu,\nu}^{(M,n)}\circ\Phi_{0,\mu,\nu}^{(M,n)}$ is expressed in terms of a \emph{quantum shift channel of order $2M$}, while the channel $\Phi_{0,\mu,\nu}^{(M,n)}$ is expressed in terms of a \emph{quantum shift channel of order $M$}.
\end{proof}

Let us now show that our model of optical fibre with memory effects, which corresponds to the DIM with $M\rightarrow\infty$, does satisfy both Property 1 and Property 2. We recall that such a model is characterised by the quantum memory channel $\{\Phi_{\lambda,\mu,\nu}^{(n)}\}_{n\in\N}$, which is defined via the following strong limit
\bb\label{def_strong_limit2}
\Phi_{\lambda,\mu,\nu}^{(n)}\coloneqq\lim\limits_{M\rightarrow\infty}\Phi_{\lambda,\mu,\nu}^{(M,n)}\qquad\text{(strongly)}\,,
\ee
meaning that we are taking into account the model in Fig.~\ref{fig:new_model} and in Fig.~\ref{fig:griglia_SM} with an infinite number of environments.
\begin{remark}\label{Remark_prop1}
    The DIM with $M\rightarrow\infty$, which is characterised by the quantum memory channel $\{\Phi_{\lambda,\mu,\nu}^{(n)}\}_{n\in\N}$, does satisfy Property 1. 
\end{remark}
\begin{proof}
    Fix an $n$-mode state $\rho$. For $\lambda=0$, if $M\ge n$ it holds that $\Phi_{0,\mu,\nu}^{(M,n)}(\rho)=\tau_\nu^{\otimes n}$. Consequently, we have $\Phi_{0,\mu,\nu}^{(n)}(\rho)=\tau_\nu^{\otimes n}$, indicating that the output state of the fibre is independent of the input state. This implies that the fibre does not transmit any information.
\end{proof}
\begin{remark}
    The DIM with $M\rightarrow\infty$, which is characterised by the quantum memory channel $\{\Phi_{\lambda,\mu,\nu}^{(n)}\}_{n\in\N}$, does satisfy Property 2. Mathematically, this means that for all $n\in\N$ and all $\lambda_1,\lambda_2\in[0,1]$ it holds that
    \bb\label{eq_thesis_comp}
        \Phi_{\lambda_1,\mu,\nu}^{(n)}\circ\Phi_{\lambda_2,\mu,\nu}^{(n)}=\Phi_{\lambda_1\lambda_2,\mu,\nu}^{(n)}
    \ee
\end{remark}
\begin{proof}
The validity of Property 2 becomes intuitive when observing Fig.~\ref{fig:griglia_SM} in the limit as $M\rightarrow\infty$. Nevertheless, let us provide a rigorous proof to establish its validity. Let us fix an $n$-mode state $\rho$. By exploiting~\eqref{lim_char}, the characteristic function of $\Phi_{\lambda_1\lambda_2,\mu,\nu}^{(n)}(\rho)$ is given by 
\bb\label{lim_char2}
    \chi_{\Phi_{\lambda_1\lambda_2,\mu,\nu}^{(n)}(\rho)}(\mathbf{z})&=\chi_{\rho}\left(\bar{A}^{(n,\lambda_1\lambda_2,\mu)^\intercal}\mathbf{z}\right)
    \\
    &\quad \times e^{-(\nu+\frac{1}{2})\mathbf{z}^\dagger ( \mathbb{1}_{n\times n}- \bar{A}^{(n,\lambda_1\lambda_2,\mu)}\bar{A}^{(n,\lambda_1\lambda_2,\mu)^\intercal} ) \mathbf{z}}
\ee
for all $\mathbf{z}\in\mathbb{C}^n$. On the other hand, the characteristic function of $\Phi_{\lambda_1,\mu,\nu}^{(n)}\circ\Phi_{\lambda_2,\mu,\nu}^{(n)}(\rho)$ is given by
\bb\label{lim_char3}
        &\chi_{ \Phi_{\lambda_1,\mu,\nu}^{(n)}\circ\Phi_{\lambda_2,\mu,\nu}^{(n)}(\rho) }(\mathbf{z})\\
        &=\chi_{ \Phi_{\lambda_2,\mu,\nu}^{(n)}(\rho) }\left(\bar{A}^{(n,\lambda_1,\mu)^\intercal}\mathbf{z}\right)  e^{-(\nu+\frac{1}{2})\mathbf{z}^\dagger ( \mathbb{1}_{n\times n}- \bar{A}^{(n,\lambda_1,\mu)}\bar{A}^{(n,\lambda_1,\mu)^\intercal} ) \mathbf{z}}\\
        &=\chi_{\rho }\left(\bar{A}^{(n,\lambda_2,\mu)^\intercal}\bar{A}^{(n,\lambda_1,\mu)^\intercal}\mathbf{z}\right) \\
        &\qquad\times e^{-(\nu+\frac{1}{2})\mathbf{z}^\dagger \bar{A}^{(n,\lambda_1,\mu)}( \mathbb{1}_{n\times n}- \bar{A}^{(n,\lambda_2,\mu)}\bar{A}^{(n,\lambda_2,\mu)^\intercal} ) \bar{A}^{(n,\lambda_1,\mu)^\intercal}\mathbf{z}} 
        \\
        &\qquad\times e^{-(\nu+\frac{1}{2})\mathbf{z}^\dagger ( \mathbb{1}_{n\times n}- \bar{A}^{(n,\lambda_1,\mu)}\bar{A}^{(n,\lambda_1,\mu)^\intercal} ) \mathbf{z}}\\
        &=\chi_{\rho }\left(\bar{A}^{(n,\lambda_2,\mu)^\intercal}\bar{A}^{(n,\lambda_1,\mu)^\intercal}\mathbf{z}\right) \\
        &\qquad\times
        e^{-(\nu+\frac{1}{2})\mathbf{z}^\dagger ( \mathbb{1}_{n\times n}- \bar{A}^{(n,\lambda_1,\mu)}\bar{A}^{(n,\lambda_2,\mu)}\bar{A}^{(n,\lambda_2,\mu)^\intercal}\bar{A}^{(n,\lambda_1,\mu)^\intercal} ) \mathbf{z}}
    \ee
for all $\mathbf{z}\in\mathbb{C}^n$. Since quantum states and characteristic functions are in one-to-one correspondence,~\eqref{lim_char2} and~\eqref{lim_char3} imply that, in order to show~\eqref{eq_thesis_comp}, we only need to show that
\bb\label{thesis_remark_prop2}
    \bar{A}^{(n,\lambda_1,\mu)}\bar{A}^{(n,\lambda_2,\mu)}=\bar{A}^{(n,\lambda_1\lambda_2,\mu)}\,.
\ee
By using~\eqref{matrix_sm}, for all $i,k=1,2,\ldots,n$ the element $(i,k)$ of the matrix $\bar{A}^{(n,\lambda_1,\mu)}\bar{A}^{(n,\lambda_2,\mu)}$ can be expressed as
\bb
    &(\bar{A}^{(n,\lambda_1,\mu)}\bar{A}^{(n,\lambda_2,\mu)})_{i,k}\\
    &=\sum_{j=0}^n\bar{A}^{(n,\lambda_1,\mu)}_{i,j}\bar{A}^{(n,\lambda_2,\mu)}_{j,k}\\&=\sum_{j=0}^n\Theta(i-j)\,\sqrt{\lambda_1}\mu^{\frac{i-j}{2}} L_{i-j}^{(-1)}(-\ln\lambda_1)\,\Theta(j-k) \sqrt{\lambda_2}\mu^{\frac{j-k}{2}} L_{j-k}^{(-1)}(-\ln\lambda_2)\\&=\Theta(k-i)\sqrt{\lambda_1\lambda_2}\mu^{\frac{i-k}{2}}\sum_{j=k}^iL_{i-j}^{(-1)}(-\ln\lambda_1)\,L_{j-k}^{(-1)}(-\ln\lambda_2)\\&=\Theta(k-i)\sqrt{\lambda_1\lambda_2}\mu^{\frac{i-k}{2}}\sum_{j=0}^{i-k}L_{i-k-j}^{(-1)}(-\ln\lambda_1)\,L_{j}^{(-1)}(-\ln\lambda_2)\\&\eqt{(i)}\Theta(k-i)\sqrt{\lambda_1\lambda_2}\mu^{\frac{i-k}{2}}\,L_{i-k}^{(-1)}\left(-\ln(\lambda_1\lambda_2)\right)\\&=\bar{A}^{(n,\lambda_1\lambda_2,\mu)}_{i,k}\,,
\ee
where in (i) we exploited the addition formula for Laguerre polynomials. Therefore,~\eqref{thesis_remark_prop2} holds, which, as mentioned earlier, is sufficient to conclude the proof.

\end{proof}

To summarise, while the DIM with a finite value of $M$ does not satisfy Property 1 and Property 2, the DIM with $M\rightarrow\infty$ does. In particular, the DIM with $M\rightarrow\infty$ is more realistic than the LIM analysed in~\cite{Memory1,Memory2,Memory3,Die-Hard-2-PRA,Die-Hard-2-PRL} (which corresponds to the DIM with $M=1$).

Since the DIM with a finite value of $M$ fails to satisfy Property 1 and Property 2, it is not a reasonable model of optical fibre with memory effects. We will now provide an intuitive explanation of this fact and also explain why the DIM with $M \rightarrow \infty$ serves as a reasonable model. This explanation hinges on the fact that the interaction between an optical signal and the optical fibre can occur at \emph{any} position along the entire length of the fibre. Therefore, to accurately model this interaction, it is essential to represent the optical fibre as composed of infinitesimal constituents by taking the limit $M \rightarrow \infty$, corresponding to a continuous optical fibre. The reasoning behind proposing our DIM with \(M \rightarrow \infty\) is inspired by the (classical) model of a vibrating string proposed by Jean Le Rond d'Alembert in 1746~\cite{alembert1749suite2}. To model a vibrating string, d'Alembert first considered $k$ discrete masses connected by springs. He then took the limit as $k$ approached infinity, while the size and mass of each constituent approached zero, ensuring the string's density remained finite. This approach led d'Alembert to derive the celebrated wave equation for a one-dimensional string~\cite{alembert1749suite2}.

\section{Avram--Parter theorem and its consequences}\label{section_toeplitz}
In this section, we review the \emph{Avram--Parter's theorem}~\cite{Avram1988,Parter1986} and, by leveraging it, we establish a matrix-analysis result in the forthcoming Theorem~\ref{Thm_consequence_avram_parter} that will play a crucial role in the analysis of the quantum communication performances which are achievable within our model of optical fibre with memory effects. The Avram--Parter's theorem is a matrix-analysis result concerning the asymptotic behaviour of the singular values of an $n\times n$ \emph{Toeplitz} matrix as $n\to\infty$ (see Theorem~\eqref{Avram--Parter} below for its statement). Such a theorem can be seen as a generalisation of the so-called Szeg\H{o} theorem~\cite{Szego1920, GRENADER}, which pertains to eigenvalues instead of singular values. Importantly, our Theorem~\ref{Thm_consequence_avram_parter}, derived by exploiting the Avram--Parter's theorem, appears to be novel in the field of matrix analysis and we believe that it may hold independent interest. We begin with the definition of Toeplitz matrices~\cite{toeplitz_review}.
\begin{definition}
    an $n\times n$ real matrix $A$ is called a Toeplitz matrix if there exist real numbers $\{a_j\}_{j}$ such that the elements of $A$ satisfy
    \bb 
        A_{i,k}=a_{i-k}\qquad\forall\, i,k\in\{1,2,\ldots,n\}\,.
    \ee
\end{definition}
In addition, a \emph{circulant matrix} is a square matrix that has all of its row vectors composed of the same elements, with each row vector rotated one element to the right relative to the preceding row vector. Circulant matrices have the advantageous property that their spectral and singular value decompositions can be analytically calculated. In a sense, infinite-dimensional Toeplitz matrices can be viewed as infinite-dimensional circulant matrices. This intuitive idea forms the basis for the proof of the so-called Szeg\H{o} theorem~\cite{Szego1920, GRENADER}, which pertains to eigenvalues, as well as the Avram--Parter theorem~\cite{Avram1988,Parter1986}, which pertains to singular values. For an understanding of topics related to Toeplitz matrices, Szeg\H{o} theorem, and Avram--Parter theorem, we recommend referring to the references~\cite{GRUDSKY, BOETTCHER, GARONI, Grudsky-lectures, toeplitz_review}. It is worth noting that the Szeg\H{o} theorem has been previously applied in the literature of quantum information theory, specifically in works such as~\cite{Memory1,Memory2,Memory3,LL-bosonic-dephasing}.

\begin{thm}[\emph{[Avram--Parter's Theorem]}]\label{Avram--Parter}
    Let $\{a_k\}_{k\in\mathbb{Z}}$ be a sequence of real numbers. For all $n\in\N$ let $T^{(n)}$ be the $n\times n$ Toeplitz matrix with elements $T^{(n)}_{k,j}\coloneqq a_{k-j}$ for all $k,j\in\{1,2,\ldots,n\}$. Let $\{s^{(n)}_j\}_{j=1,2,\ldots,n}$ be the singular values of the matrix $T^{(n)}$ ordered in increasing order in $j$. Assume that the function
    $s:[0,2\pi]\to \C$ defined by 
    \bb\label{asymptotic_distribution0}
    s(x)\coloneqq \left|\sum_{k=-\infty}^{+\infty}a_k\,e^{\frac{ikx}{2}}\right|\,\qquad\forall \,x\in[0,2\pi]\,,
    \ee
    is bounded. Then for all continuous function $F:\mathbb{R}\to\mathbb{R}$ with bounded support it holds that
    \begin{equation}
        \lim\limits_{n\rightarrow\infty}\frac{1}{n}\sum_{j=1}^nF\left(s^{(n)}_j\right)=\int_{0}^{2\pi}\frac{\mathrm{d}x}{2\pi}F\left(s(x)\right)\,.
    \end{equation}
\end{thm}
\begin{lemma}\label{lemmino_LL}
    By using the notations of Theorem~\ref{Avram--Parter}, for all $n\in\N$ it holds that all the singular values of $T^{(n)}$ are upper bounded by $\sup_{x\in[0,2\pi]}s(x)$, i.e.
    \bb
        s^{(n)}_n\le \sup_{x\in[0,2\pi]}s(x)\,\qquad\forall\,n\in\N\,.
    \ee
\end{lemma}
\begin{proof}
    The maximum singular value of the matrix $T^{(n)}$ satisfy
\bb
        &s^{(n)}_n = \max_{\substack{\mathbf{v},\mathbf{w}\in\mathbb{C}^n, \\ \mathbf{v}^T\mathbf{v}=\mathbf{w}^T\mathbf{w}=1}} |\mathbf{v}^T T^{(n)} \mathbf{w}|\\
        &= \max_{\substack{\mathbf{v},\mathbf{w}\in\mathbb{C}^n, \\ \mathbf{v}^T\mathbf{v}=\mathbf{w}^T\mathbf{w}=1}} \left|\sum_{k,j=1}^n v_k^\ast T^{(n)}_{k,j}w_j \right|\\
        &= \max_{\substack{\mathbf{v},\mathbf{w}\in\mathbb{C}^n, \\ \mathbf{v}^T\mathbf{v}=\mathbf{w}^T\mathbf{w}=1}} \left|\sum_{k,j=1}^n v_k^\ast w_j a_{k-j} \right|\\
        &\eqt{(i)} \max_{\substack{\mathbf{v},\mathbf{w}\in\mathbb{C}^n, \\ \mathbf{v}^T\mathbf{v}=\mathbf{w}^T\mathbf{w}=1}} \left|\sum_{k,j=1}^n v_k^\ast w_j \int_{0}^{4\pi}\frac{\text{d}x}{4\pi}e^{-i(k-j)\frac{x}{2}}\sum_{l=-\infty}^{+\infty}a_l e^{il \frac{x}{2}} \right|\\
        &\le \max_{\substack{\mathbf{v},\mathbf{w}\in\mathbb{C}^n, \\ \mathbf{v}^T\mathbf{v}=\mathbf{w}^T\mathbf{w}=1}} \int_{0}^{4\pi}\frac{\text{d}x}{4\pi}\left|\left(\sum_{k=1}^n v_k e^{ik\frac{x}{2}}\right)^\ast  \left(\sum_{j=1}^n w_j e^{ij\frac{x}{2}}\right)\right|\\& \quad\times\left|\sum_{l=-\infty}^{+\infty}a_l e^{il \frac{x}{2}} \right| \\&=\max_{\substack{\mathbf{v},\mathbf{w}\in\mathbb{C}^n, \\ \mathbf{v}^T\mathbf{v}=\mathbf{w}^T\mathbf{w}=1}} \int_{0}^{4\pi}\frac{\text{d}x}{4\pi}\left|\left(\sum_{k=1}^n v_k e^{ik\frac{x}{2}}\right)^\ast  \left(\sum_{j=1}^n w_j e^{ij\frac{x}{2}}\right)\right|s(x) \\&\le \left(\sup_{x\in[0,4\pi]}s(x)\right)\max_{\substack{\mathbf{v},\mathbf{w}\in\mathbb{C}^n, \\ \mathbf{v}^T\mathbf{v}=\mathbf{w}^T\mathbf{w}=1}} \\
        &\quad\times\int_{0}^{2\pi}\frac{\text{d}x}{2\pi}\left|\left(\sum_{k=1}^n v_k e^{ikx}\right)^\ast  \left(\sum_{j=1}^n w_j e^{ijx}\right)\right|\\&\le  \left(\sup_{x\in[0,4\pi]}s(x)\right)\max_{\substack{\mathbf{v},\mathbf{w}\in\mathbb{C}^n, \\ \mathbf{v}^T\mathbf{v}=\mathbf{w}^T\mathbf{w}=1}} \sqrt{\int_{0}^{2\pi}\frac{\text{d}x}{2\pi}\left|\sum_{k=1}^n v_k e^{ikx}\right|^2} \\
        & \quad\times \sqrt{\int_{0}^{2\pi}\frac{\mathrm{d}x}{2\pi}\left|\sum_{j=1}^n w_j e^{ijx}\right|^2}\\
        &=\sup_{x\in[0,4\pi]}s(x)\,.
    \ee
    where in (i) we exploited the fact that
    \bb
        \int_{0}^{4\pi}\frac{\text{d}x}{4\pi}e^{-i(k-j)\frac{x}{2}}e^{il \frac{x}{2}} =\begin{cases}
        1 &\text{if $l=k-j$ ,} \\
        0 &\text{otherwise.}
    \end{cases} 
    \ee
    Since $\{a_k\}_k$ are real numbers, the function $s(x)\coloneqq \left|\sum_{k=-\infty}^{+\infty}a_k\,e^{\frac{ikx}{2}}\right|$ satisfies $s(4\pi-x)=s(x)$.
    Consequently, it holds that 
    \bb
         \sup_{x\in[0,4\pi]}s(x)=\sup_{x\in[0,2\pi]}s(x)\,.
    \ee
    Hence, this concludes the proof.
\end{proof}
\begin{lemma}\label{lemma_consequence}
    For all $n\in \N$ let $\{s^{(n)}_j\}_{j\in\{1,2,\ldots,n\}}$ be $n$ real numbers that are non-decreasing in $j$ and uniformly bounded, i.e.~there exists $M>0$ such that $|s^{(n)}_n|\le M$ for all $n\in\N$. Assume that there exists a monotonically non-decreasing continuous function $s:[0,2\pi]\to\mathbb{R}$ such that the following relation holds for any continuous function $F:\mathbb{R}\to\mathbb{R}$ with bounded support:
    \bb\label{relation_F}
        \lim\limits_{n\rightarrow\infty}\frac{1}{n}\sum_{j=1}^nF\left(s^{(n)}_j\right)=\int_{0}^{2\pi}\frac{\mathrm{d}x}{2\pi}F\left(s(x)\right)\,.
    \ee
    Then, it holds that
    \bb\label{eq:F}        &\lim\limits_{n\rightarrow\infty}\max\Bigg\{\left|s^{(n)}_j-s\left(\frac{2\pi j}{n}\right)\right|\,\\
    &\qquad:\, j\in\{1,2,\ldots,n\},\quad s^{(n)}_j\in[s(0),s(2\pi)]\Bigg\}=0
    \ee
    and
    \bb\label{eq:boundary}
        \lim\limits_{n\rightarrow\infty}\frac{1}{n}\left|\left\{j:\, j\in\{1,2,\ldots,n\},\,\quad s^{(n)}_j\notin[s(0),s(2\pi)]\,  \right\}\right|=0\,.
    \ee
\end{lemma} 
\begin{proof}
For each $n\in\mathbb{N}$, let
\begin{equation}
  \mu_n \coloneqq \frac{1}{n}\sum_{j=1}^n\delta_{s_j^{(n)}}\,,
\end{equation}
where for any $x\in\mathbb{R}$, $\delta_x$ is the Dirac delta probability distribution centered in $x$.
From~\eqref{relation_F}, the sequence of probability distributions $\mu_n$ converges weakly to the probability distribution $\mu$ with cumulative distribution function $C$~\cite{chow2012probability}, which is defined by the following equation:
\begin{equation}
C^{-1}(y) = s(2\pi y)\qquad\forall\,y\in[0,1]\,,
\end{equation}
where $C^{-1}$ is the inverse function of $C$.
For any $n\in\N$, let $C_n$ be the cumulative distribution function of $\mu_n$, which satisfies
\begin{equation}
C_n\left(s^{(n)}_j\right) = \frac{j}{n}\qquad\forall\,j=1,\,\ldots,\,n\,.
\end{equation}
We have
\begin{align}\label{eq:limsup}
&\limsup_{n\to\infty}\max\Bigg\{\left|s^{(n)}_j - s\left(\frac{2\pi\,j}{n}\right)\right|:  \\
&\qquad j=1,\,\ldots,\,n\,,\;s(0)\le s^{(n)}_j \le s(2\pi)\Bigg\}\nonumber\\
&=\limsup_{n\to\infty}\max\Bigg\{\left|s^{(n)}_j - s\left(2\,\pi\,C_n\left(s^{(n)}_j\right)\right)\right|: \\
&\qquad j=1,\,\ldots,\,n\,,\;s(0)\le s^{(n)}_j \le s(2\pi)\Bigg\}\nonumber\\
&\le\limsup_{n\to\infty}\max\Big\{\left|\lambda - C^{-1}\left(C_n(\lambda)\right)\right|:s(0)\le \lambda \le s(2\pi)\Big\}\,.
\end{align}
Since $C$ is continuous and $\mu_n$ converges weakly to $\mu$, $C_n$ converges uniformly to $C$~\cite{chow2012probability}.
Therefore, $C^{-1}\circ C_n$ converges uniformly to the identity function on the interval $[s(0),\,s(2\pi)]$, and the claim~\eqref{eq:F} follows. Moreover, for any $n\in\mathbb{N}$, let
\bb
j_n &= \max\left\{j=1,\,\ldots,\,n : s^{(n)}_j < s(0)\right\} \\
&= \left|\left\{j=1,\,\ldots,\,n : s^{(n)}_j < s(0)\right\}\right|\,.
\ee
We have
\bb
\limsup_{n\to\infty}\frac{j_n}{n} &= \limsup_{n\to\infty}C_n\left(s^{(n)}_{j_n}\right) \\
&\le \limsup_{n\to\infty}C_n\left(s(0)\right)\\
&= C(s(0)) \\
&= 0\,.
\ee
Analogously, one can show that
\bb
    \limsup_{n\rightarrow\infty}\frac{1}{n}\left|\left\{j=1,\,\ldots,\,n : s^{(n)}_j > s(2\pi)\right\}\right|=0\,,
\ee
and the claim~\eqref{eq:boundary} follows.

\end{proof}

\begin{thm}\label{Thm_consequence_avram_parter}
    Let $\{a_k\}_{k\in\mathbb{Z}}$ be a sequence of real numbers. For all $n\in\N$ let $T^{(n)}$ be the $n\times n$ Toeplitz matrix with elements $T^{(n)}_{k,j}\coloneqq a_{k-j}$ for all $k,j\in\{1,2,\ldots,n\}$. Let $\{s^{(n)}_j\}_{j=1,2,\ldots,n}$ be the singular values of the matrix $T^{(n)}$ ordered in increasing order in $j$. Assume that the function
    $s:[0,2\pi]\to \C$ defined by 
    \bb\label{asymptotic_distribution}
    s(x)\coloneqq \left|\sum_{k=-\infty}^{+\infty}a_k\,e^{\frac{ikx}{2}}\right|\,\qquad\forall\,x\in[0,2\pi]\,,
    \ee
    is monotonically non-decreasing and continuous. Then, the following notion of convergence of the singular values $\{s^{(n)}_j\}_{j=1,2,\ldots,n}$ to the function $s(\cdot)$ 
 holds:
    \bb
        \lim\limits_{n\rightarrow\infty}\max\left\{\left|s^{(n)}_j-s\left(\frac{2\pi j}{n}\right)\right|:\, j\in\{k_n,k_n+1,\ldots,n\}\right\}=0\,,
    \ee
    where $\{k_n\}_{n\in\N}\subseteq \N$ is a suitable sequence such that $k_n\le n$ for all $n\in\N$ and $\lim\limits_{n\rightarrow\infty}\frac{k_n}{n}=0$. In particular, for all $x\in(0,1]$ it holds that 
    \bb
        \lim\limits_{n\rightarrow\infty}s^{(n)}_{ \floor{x n} }=s(2\pi x)\,,
    \ee
    where $\floor{\cdot}$ is the floor function.
    Furthermore, as $n$ approaches infinity, the fraction of singular values $\{s^{(n)}_j\}_{j=1,2,\ldots,n}$ that fall outside the image of the function $s(\cdot)$ becomes negligible. Mathematically, this can be expressed as
    \bb\label{sing_negligible_fraction}
        \lim\limits_{n\rightarrow\infty}\frac{1}{n}\left|\left\{j:\, j\in\{1,2,\ldots,n\},\,\quad s^{(n)}_j\notin[s(0),s(2\pi)]\,  \right\}\right|=0\,.
    \ee
More specifically, the inequality $s^{(n)}_n\le s(2\pi)$ holds for any $n\in\N$.
\end{thm}
\begin{proof}
    This is a direct consequence of Theorem~\ref{Avram--Parter}, Lemma~\ref{lemmino_LL}, and Lemma~\ref{lemma_consequence}.
\end{proof}
Analogously to Theorem~\ref{Thm_consequence_avram_parter}, a similar statement can be made regarding the convergence of the eigenvalues of $T^{(n)}$ as $n$ approaches infinity, leveraging the Szeg\H{o} Theorem~\cite{Szego1920, GRENADER}.

\section{Capacities of the Delocalised Interaction Model}\label{Sec_capacities_DIM}
In this section, we derive the exact solutions for the quantum capacity $Q$, the two-way quantum capacity $Q_2$, and the secret-key capacity $K$ of the DIM in the absence of thermal noise ($\nu=0$). Furthermore, we investigate the parameter regions of noise where the capacities of our model are strictly positive. This characterisation allows us to determine the regions where qubit distribution, two-way entanglement distribution, and quantum key distribution can be achieved.

Let us recall that, in our model, the channel $\Phi_{\lambda,\mu,\nu}^{(n)}$ corresponds to $n$ uses of the optical fibre. In addition, in Theorem~\ref{thm_strong_sm} we proved that $\Phi_{\lambda,\mu,\nu}^{(n)}$ is unitary equivalent to the following channel:
\begin{equation}
\left\{\bigotimes_{i=1}^n \mathcal{E}_{\eta_i^{(n,\lambda,\mu)},\nu} \right\}_{n\in\mathbb{N}},
\end{equation}
which is the tensor product of $n$ distinct thermal attenuators. The transmissivities $\{\eta_i^{(n,\lambda,\mu)}\}_{i=1,2,\ldots,n}$ of these attenuators are the square of the singular values of the Toeplitz matrix reported in~\eqref{matrix_sm}. Hence, we can exploit here the properties of Toeplitz matrices, that we discussed in Section~\ref{section_toeplitz}, to analyse the performance of quantum communication tasks over optical fibres with memory effects. We begin by establishing the forthcoming Lemma~\ref{thm_conv_main_sm}, which, roughly speaking, establishes that by plotting the points
$$\left\{\left(2\pi\frac{j}{n}\,,\, \eta_j^{(n,\lambda,\mu)}\right)\quad \text{for } j=1,2,\ldots,n \right\}$$ on a two-dimensional plane, they converge for $n\to\infty$ to the graph of a certain \emph{known} function $\eta^{(\lambda,\mu)}:[0,2\pi]\to\mathbb{R}$ reported in~\eqref{formula_effect_transm_sm}, which we dub \emph{effective transmissivity function} (see Fig.~\ref{fig:plot_effective_trans} for an example). 
\begin{figure}[t]
	\includegraphics[width=1\linewidth]{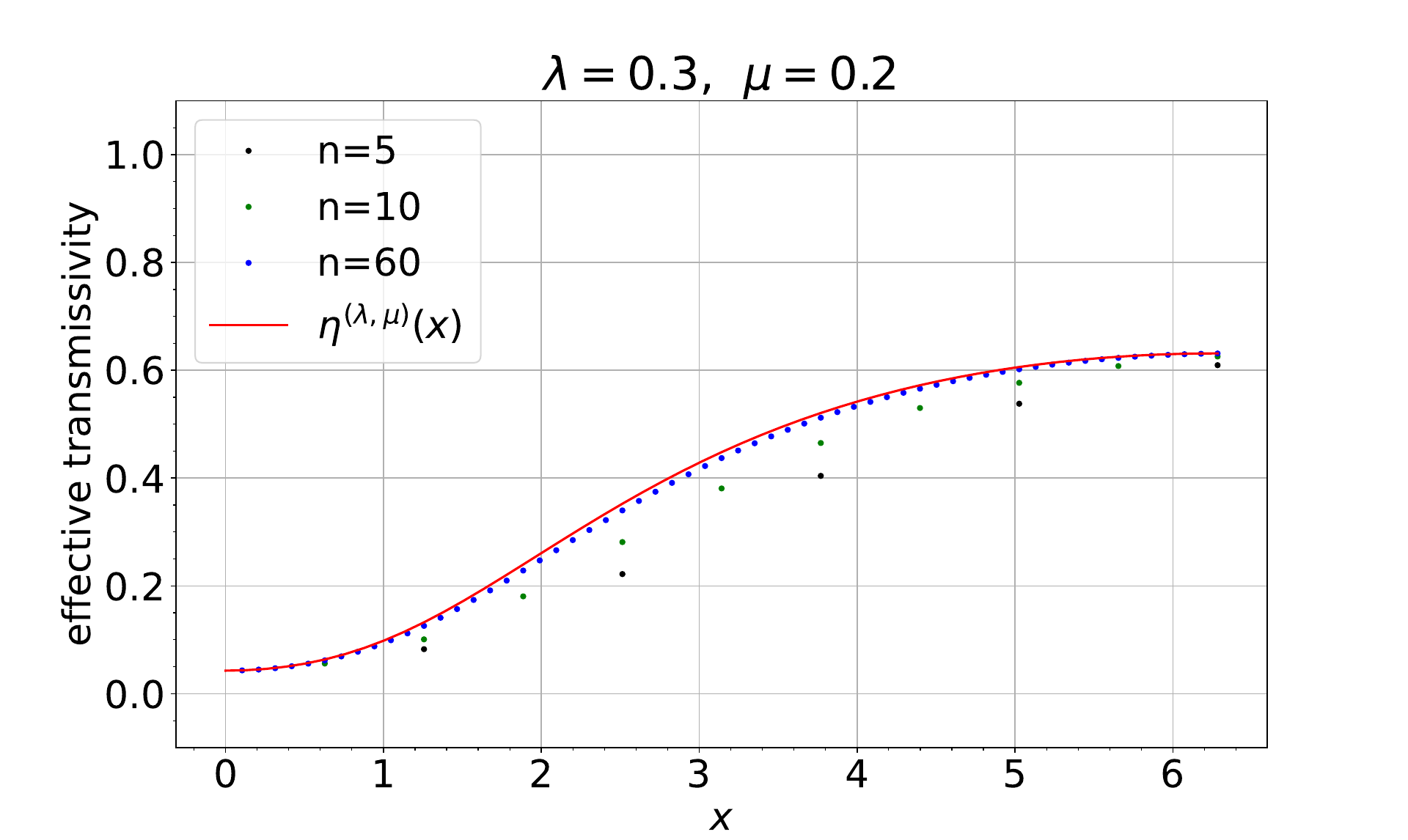}
	\caption{The red curve is the effective transmissivity function $\eta^{(\lambda,\mu)}(x)$, as defined in~\eqref{formula_effect_transm_sm}, plotted with respect $x\in[0,2\pi]$ for $\lambda=0.3$ and $\mu=0.2$. For $n=4$ (black points), $n=10$ (green points), and $n=60$ (blue points), we plot the $n$ points $\left\{\left(2\pi\frac{j}{n}\,,\, \eta_j^{(n,\lambda,\mu)}\right)\,:\, j=1,2,\ldots,n \right\}$. We observe that these $n$ points converge for $n\rightarrow\infty$ to the graph of the effective transmissivity function $\eta^{(\lambda,\mu)}$, as established by Lemma~\ref{thm_conv_main_sm}. Here, $\{\eta_j^{(n,\lambda,\mu)}\}_{j=1,2,\ldots,n}$ are the transmissivities defined in Theorem~\ref{thm_strong_sm}. }
\label{fig:plot_effective_trans}
\end{figure}
\begin{lemma}\label{thm_conv_main_sm}
    Let $\lambda\in[0,1]$ and $\mu\in[0,1)$. There exists a sequence $\{j_n\}_{n\in\N}\subseteq \N$ such that $j_n\le n$ for all $n\in\N$, $\lim\limits_{n\rightarrow\infty}\frac{j_n}{n}=0$, and
    \bb\label{eq_conv_statemente_sm}
    \lim\limits_{n\rightarrow\infty}\max\left\{ \left|\eta^{(n,\lambda,\mu)}_j-\eta^{(\lambda,\mu)}\left(\frac{2\pi j}{n}\right)\right|:\,  j\in\{j_n,\ldots,n\}  \right\}=0\,,
    \ee
    where $\eta^{(\lambda,\mu)}$ is the effective transmissivity function defined by
    \bb\label{formula_effect_transm_sm}
        \eta^{(\lambda,\mu)}(x) \coloneqq \lambda^{\frac{1-\mu}{ 1+\mu-2\sqrt{\mu}\cos(x/2) }}\quad\forall \,x\in[0,2\pi]\,.
    \ee
    In particular, for all $x\in(0,1]$ it holds that 
    \bb
        \lim\limits_{n\rightarrow\infty}\eta^{(n,\lambda,\mu)}_{ \floor{x n} }=\eta^{(\lambda,\mu)}(2\pi x)\,,
    \ee
    where $\floor{\cdot}$ is the floor function.
    In addition, all the transmissivities are uniformly bounded by $\eta^{(\lambda,\mu)}(2\pi)$, i.e.
    \bb\label{upp_eta2pi}
        \eta^{(n,\lambda,\mu)}_n\le \eta^{(\lambda,\mu)}(2\pi)=\lambda^{ \frac{1-\sqrt{\mu}}{1+\sqrt{\mu}} }\qquad\forall\, n\in\N\,.
    \ee
\end{lemma} 
\begin{proof}
    The transmissivities $\{\eta_i^{(n,\lambda,\mu)}\}_{i=1,2,\ldots,n}$ are defined by 
    \bb
        \eta_i^{(n,\lambda,\mu)}\coloneqq (s_i^{(n,\lambda,\mu)})^2\,,
    \ee
    where $\{s_i^{(n,\lambda,\mu)}\}_{i=1,2,\ldots,n}$ are the singular values, ordered in increasing order in $j$, of the $n\times n$ Toeplitz matrix $\bar{A}^{(n,\lambda,\mu)}$ reported in~\eqref{matrix_sm}. For all $n\in\N$ and all $i,k\in\{1,2,\ldots,n\}$, the $(i,k)$ element of $\bar{A}^{(n,\lambda,\mu)}$ can be expressed as 
        \bb\label{matrix_sm22}
            \bar{A}^{(n,\lambda,\mu)}_{i,k}=a^{(\lambda,\mu)}_{i-k}
        \ee
    where 
    \bb
        a_j^{(\lambda,\mu)}\coloneqq  \Theta(j)\,\sqrt{\lambda}\mu^{\frac{j}{2}} L_{j}^{(-1)}(-\ln\lambda)\qquad \forall\, j\in\mathbb{Z}\,.
    \ee
    By using that the generalised Laguerre polynomials $\{L^{(-1)}_k\}_{k\in\N}$ satisfies 
\bb
    \sum_{k=0}^\infty L^{(-1)}_k(x)\,w^k=e^{\frac{wx}{w-1}}
\ee
for all $x\in\R$ and all $w\in\C$ such that $|w|<1$, one can show that
\bb
    \left|\sum_{k=-\infty}^{+\infty}a_k^{(\lambda,\mu)}\,e^{\frac{ikx}{2}}\right|^2 = \lambda^{\frac{1-\mu}{ 1+\mu-2\sqrt{\mu}\cos(x/2) }}\qquad\forall\, x\in [0,2\pi]\,.
\ee
Consequently, by applying Theorem~\ref{Thm_consequence_avram_parter}, the proof of~\eqref{upp_eta2pi} is complete. Moreover, Theorem~\ref{Thm_consequence_avram_parter} implies that there exists a sequence $\{j_n\}_{n\in\N}\subseteq \N$ such that $j_n\le n$ for all $n\in\N$, $\lim\limits_{n\rightarrow\infty}\frac{j_n}{n}=0$, and
\bb
    &\lim\limits_{n\rightarrow\infty}\max\Bigg\{ \left|\sqrt{\eta^{(n,\lambda,\mu)}_j}-\sqrt{\eta^{(\lambda,\mu)}\left(\frac{2\pi j}{n}\right)}\right|:  j\in\{j_n,\ldots,n\}  \Bigg\}=0\,.
\ee
Since $\{\eta_i^{(n,\lambda,\mu)}\}_{i=1,2,\ldots,n}$ and $\eta^{(\lambda,\mu)}$ are both bounded, with lower bound $0$ and upper bound $1$, the limit in~\eqref{eq_conv_statemente_sm} follows.
\end{proof}
In the forthcoming Lemma~\ref{pos_cap_thm_sm} we study the parameter region of zero capacities of our model of optical fibre with memory effects. 
\begin{thm}\label{pos_cap_thm_sm}
    Let $\lambda\in(0,1]$, $\mu\in[0,1)$, and $\nu\ge 0$. Let $C(\lambda,\mu,\nu)$ be one of the following capacities of the quantum memory channel $\{\Phi_{\lambda,\mu,\nu}^{(n)}\}_{n\in\N}$: quantum capacity $Q$, two-way quantum capacity $Q_2$, or secret key capacity $K$. It holds that
    \bb\label{iff_cap_0}
        C(\lambda,\mu,\nu)>0\qquad\Longleftrightarrow \qquad C(\mathcal{E}_{\eta^{(\lambda,\mu)}(2\pi),\nu})>0\,,
    \ee
    where $\eta^{(\lambda,\mu)}$ is reported in~\eqref{formula_effect_transm_sm}. In particular, in the absence of thermal noise, i.e.~$\nu=0$, it holds that
\bb\label{cond_strictly_q_sm}
    Q(\lambda,\mu,\nu=0)>0\qquad\Longleftrightarrow\qquad\sqrt{\mu}>\frac{ \log_2\left(\frac{1}{\lambda}\right)-1 }{  \log_2\left(\frac{1}{\lambda}\right)+1  }\,.
\ee
In addition, for all $\nu\ge0$ it holds that
\bb\label{cond_strictly_q2_sm}
    &K(\lambda,\mu,\nu),\,Q_2(\lambda,\mu,\nu)>0\Longleftrightarrow
	\sqrt{\mu}> \frac{ \ln\left(\frac{1}{\lambda}\right)-\ln(1+\frac{1}{\nu}) }{  \ln\left(\frac{1}{\lambda}\right)+\ln(1+\frac{1}{\nu})  }\,,\\
 &Q(\lambda,\mu,\nu)>0\Longrightarrow\sqrt{\mu}> \frac{ \ln\left(\frac{1}{\lambda}\right)-\ln\left(1+\frac{1}{2\nu+1}\right) }{  \ln\left(\frac{1}{\lambda}\right)+\ln\left(1+\frac{1}{2\nu+1}\right)  }\,,\\
&Q(\lambda,\mu,\nu)>0\Longleftarrow 
	\sqrt{\mu}> \frac{ \ln\left(\frac{1}{\lambda}\right)-\ln\left( 1+2^{-g(\nu)}\right) }{  \ln\left(\frac{1}{\lambda}\right)+\ln\left( 1+2^{-g(\nu)}\right)  }\,.
 \ee
\end{thm}
\begin{proof}
    Theorem~\ref{thm_strong_sm} establishes that the quantum memory channel $\{\Phi_{\lambda,\mu,\nu}^{(n)}\}_{n\in\N}$ is unitary equivalent to the quantum memory channel $\left\{\bigotimes_{i=1}^n \mathcal{E}_{\eta_i^{(n,\lambda,\mu)},\nu} \right\}_{n\in\mathbb{N}}$. Hence, their capacities $C$ are identical. Let us begin with the proof of the following implication: 
    \bb\label{implication1}
        C(\lambda,\mu,\nu)>0\Longrightarrow C(\mathcal{E}_{\eta^{(\lambda,\mu)}(2\pi),\nu})>0\,.
    \ee
    The composition rule in~\eqref{composition_them} establishes that any thermal attenuator with a given transmissivity and thermal noise can be simulated using any other thermal attenuator with higher transmissivity but the same thermal noise. Specifically, for two thermal attenuators $\mathcal{E}_{\lambda_1,\nu}$ and $\mathcal{E}_{\lambda_2,\nu}$ where $\lambda_1 \le \lambda_2$, the composition rule in~\eqref{composition_them} implies that $\mathcal{E}_{\lambda_1,\nu}=\mathcal{E}_{\frac{\lambda_1}{\lambda_2},\nu} \circ \mathcal{E}_{\lambda_2,\nu}$. Therefore, if Alice and Bob are connected through $\mathcal{E}_{\lambda_2,\nu}$, they can simulate the channel $\mathcal{E}_{\lambda_1,\nu}$ provided that Bob applies the channel $\mathcal{E}_{\frac{\lambda_1}{\lambda_2},\nu}$ to the output of $\mathcal{E}_{\lambda_2,\nu}$. As a consequence, since \eqref{upp_eta2pi} establishes that all the transmissivities $\{\eta^{(n,\lambda,\mu)}_i\}_{i=1,2,\ldots,n}$ are upper bounded by $\eta^{(\lambda,\mu)}(2\pi)$, Alice and Bob can simulate $\bigotimes_{i=1}^n \mathcal{E}_{\eta_i^{(n,\lambda,\mu)},\nu}$ by using $n$ uses of the memoryless channel $\mathcal{E}_{\eta^{(\lambda,\mu)}(2\pi),\nu}$. This implies that $C(\mathcal{E}_{\eta^{(\lambda,\mu)}(2\pi),\nu})\ge C(\lambda,\mu,\nu)$, which constitutes a proof of~\eqref{implication1}. Now, let us show that 
    \bb\label{implication2}
        C(\lambda,\mu,\nu)>0\Longleftarrow C(\mathcal{E}_{\eta^{(\lambda,\mu)}(2\pi),\nu})>0\,.
    \ee
    Assume $C(\mathcal{E}_{\eta^{(\lambda,\mu)}(2\pi),\nu})>0$. Since the function $x\longmapsto C(\mathcal{E}_{\eta^{(\lambda,\mu)}(2\pi x),\nu})$ is continuous in $x\in[0,1]$, there exists $\bar{x}\in[0,1]$ such that $C(\mathcal{E}_{\eta^{(\lambda,\mu)}(2\pi x),\nu})>0$ for all $x\in [\bar{x},1]$. Observe that Lemma~\ref{thm_conv_main_sm} implies that 
    \bb
        \lim\limits_{n\rightarrow\infty} \eta^{(n,\lambda,\mu)}_{\floor{\left(\frac{1+\bar{x}}{2}\right)n}}&=\eta^{(\lambda,\mu)}\left(2\pi \left(\frac{1+\bar{x}}{2}\right)\right)>\eta^{(\lambda,\mu)}(2\pi \bar{x})\,,
    \ee
    where $\floor{\cdot}$ is the floor function. Consequently, if $n$ is sufficiently large it holds that 
    \bb
        \eta^{(n,\lambda,\mu)}_{j}>\eta^{(\lambda,\mu)}(2\pi \bar{x})\qquad\forall\, j\in\left\{\floor{\left(\frac{1+\bar{x}}{2}\right)n},\ldots,n\right\}\,.
    \ee
     Hence, since any thermal attenuator with a given transmissivity and thermal noise can be simulated using any other thermal attenuator with higher transmissivity but the same thermal noise (as demonstrated above), Alice and Bob can simulate $n-\floor{\left(\frac{1+\bar{x}}{2}\right)n}+1$ uses of the memoryless channel $\mathcal{E}_{ \eta^{(\lambda,\mu)}(2\pi)  }$ by using the channel $\bigotimes_{i=1}^n \mathcal{E}_{\eta_i^{(n,\lambda,\mu)},\nu}$, for $n$ sufficiently large. Hence, since
     \bb
        \lim\limits_{n\rightarrow\infty} \frac{n-\floor{\left(\frac{1+\bar{x}}{2}\right)n}+1}{n}=\frac{1-\bar{x}}{2}\,,
     \ee
     it holds that 
     \bb
        C(\lambda,\mu,\nu)\ge\left(\frac{1-\bar{x}}{2}\right)C(\mathcal{E}_{ \eta^{(\lambda,\mu)}(2\pi)  })>0\,,
     \ee
    i.e.~the implication in~\eqref{implication2} is proved. Hence, we have proved that $C(\lambda,\mu,\nu)>0$ if and only if $ C(\mathcal{E}_{\eta^{(\lambda,\mu)}(2\pi),\nu})>0 $. As a consequence of this fact, by using that 
    \bb
    \eta^{(\lambda,\mu)}(2\pi)=\lambda^{ \frac{1-\sqrt{\mu}}{1+\sqrt{\mu}} }\,,
    \ee
    and by exploiting the bounds on the parameter region of zero capacities reported in~\eqref{q_equal_zero_sm},~\eqref{q_equal_zero_sm2}, and~\eqref{q_equal_zero_sm3}, one can show the validity of~\eqref{cond_strictly_q_sm} and~\eqref{cond_strictly_q2_sm}.
\end{proof}
Now, let us discuss why Theorem~\ref{pos_cap_thm_sm} is interesting. In the absence of memory effects ($\mu=0$), it is known that no quantum communication tasks can be achieved when the transmissivity is sufficiently low ($Q=0$ for $\lambda\le \frac{1}{2}$, and $Q_2=K=0$ for $\lambda\le \frac{\nu}{\nu+1}$). However, in Theorem~\ref{pos_cap_thm_sm} we established that for any transmissivity $\lambda>0$ and thermal noise $\nu\ge0$, there exists a critical value of the memory parameter $\mu$ above which, it becomes possible to achieve qubit distribution ($Q>0$), entanglement distribution ($Q_2>0$), and secret-key distribution ($K>0$). This result is particularly intriguing as it demonstrates that memory effects provide an advantage, enabling quantum communication tasks in highly noisy regimes that were previously deemed impossible. Specifically, this phenomenon bears resemblance to the ``die-hard quantum communication'' phenomenon observed in~\cite{Die-Hard-2-PRL,Die-Hard-2-PRA} within their toy model of optical fibre with memory effects, which was limited to the analysis of $Q$ only and does not consider $K$ and $Q_2$. In our paper, we not only demonstrate the persistence of such a phenomenon in a more realistic model, accounting for $Q$, $Q_2$, and $K$, but we also derive an analytical expression for the critical value of the memory parameter $\mu$, which goes beyond the findings of~\cite{Die-Hard-2-PRL,Die-Hard-2-PRA}. Note that, by relating $\mu$ to the temporal interval $\delta t$ between subsequent input signals (e.g.~$\mu=e^{-\delta t/t_E}$ with $t_E$ being the thermalisation timescale), Theorem~\ref{pos_cap_thm_sm} can also be expressed in terms of the critical $\delta t$ below which the above mentioned quantum communication tasks can be achieved.

In the forthcoming Theorem~\ref{Cap_absence_thermal_sm} we calculate the quantum capacity $Q$, the two-way quantum capacity $Q_2$, and the secret-key capacity $K$ of our model in the absence of thermal noise ($\nu=0$). 
\begin{thm}\label{Cap_absence_thermal_sm}
    Let $\lambda\in[0,1)$, $\mu\in[0,1)$, $\nu\ge0$. Let $C$ be one of the following capacities: quantum capacity $Q$, two-way quantum capacity $Q_2$, or secret key capacity $K$. In addition, let $C(\lambda,\mu,\nu)$ be the capacity $C$ of the quantum memory channel $\{\Phi_{\lambda,\mu,\nu}^{(n)}\}_{n\in\N}$. 
    In absence of thermal noise, i.e.~$\nu=0$, it holds that
    \bb
        C(\lambda,\mu,\nu=0)=\int_{0}^{2\pi}\frac{\mathrm{d}x}{2\pi} C\left(\mathcal{E}_{ \eta^{(\lambda,\mu)}(x)   ,0}\right)\,.
    \ee
    where $\eta^{(\lambda,\mu)}(x)$ is the effective transmissivity function expressed in~\eqref{formula_effect_transm_sm}. In particular,
    \bb\label{thesis_ris_cap}
        Q\left( \lambda,\mu,\nu=0\right)&=\int_{0}^{2\pi}\frac{\mathrm{d}x}{2\pi}\, \max\left\{0,\log_2 \left(\frac{\eta^{(\lambda,\mu)}(x)}{1-\eta^{(\lambda,\mu)}(x)}\right)\right\}\,,\\
        Q_2\left( \lambda,\mu,\nu=0 \right)&=K\left( \lambda,\mu,\nu=0 \right)\\&=\int_{0}^{2\pi}\frac{\mathrm{d}x}{2\pi}\,\log_2\left(\frac{1}{1-\eta^{(\lambda,\mu)}(x)}\right)\,.
    \ee
    Moreover, in the presence of thermal noise, i.e.~$\nu>0$, it holds that 
    \bb\label{eq_low_therm}
        C(\lambda,\mu,\nu)\ge\int_{0}^{2\pi}\frac{\mathrm{d}x}{2\pi} C\left(\mathcal{E}_{ \eta^{(\lambda,\mu)}(x)   ,\nu}\right)\,.
    \ee
\end{thm}
\begin{proof}
{The structure of the proof is as follows:
\begin{itemize}
    \item \textbf{Part 1)} We find a lower and an upper bound on the capacity $C(\lambda,\mu,\nu)\coloneqq C\left(   \left\{\bigotimes_{i=1}^n \mathcal{E}_{\eta_i^{(n,\lambda,\mu)},\nu} \right\}_{n\in\mathbb{N}} \right) $ in terms of the capacity of a multimode thermal attenuator whose transmissivities are expressed in terms the effective transmissivity function given in~\eqref{formula_effect_transm_sm}. 
    \item \textbf{Part 2)} We exploit the lower bound proved in the Part 1 in order to prove \eqref{eq_low_therm}.
    \item \textbf{Part 3)} We leverage the fact that the capacities of the pure loss channel are additive (as proved in~\cite{Wolf2007} for the quantum capacity, and in~\cite{PLOB} for the two-way quantum and secret-key capacity) in order to show that the lower and upper bound proved in the Part 1 converges to the same quantity, thus proving \eqref{thesis_ris_cap}.
\end{itemize}
}
\textbf{Part 1)} Let us fix $P\in\N$ with $P\ge 2$. In addition, for all $n\in\N$ and all $i\in\N$ with $i\ge n$, let us define 
\bb
    \eta^{(n,\lambda,\mu)}_{i}\coloneqq \eta^{(n,\lambda,\mu)}_{n}\,.
\ee
By Lemma~\ref{thm_conv_main_sm}, there exists a sequence $\{j_k\}_{k\in\N}\subseteq \N$ such that $j_k\le k$ for all $k\in\N$, $\lim\limits_{k\rightarrow\infty}\frac{j_k}{k}=0$, and
    \bb
    \max\left\{ \left|\eta^{(n,\lambda,\mu)}_j-\eta^{(\lambda,\mu)}\left(\frac{2\pi j}{n}\right)\right|:\,  j\in\{j_{n},\ldots,n\}  \right\}<\frac{1}{P}
    \ee
for $n$ sufficiently large. In particular, 
    \bb\label{ineq_proof_sm}
    \eta^{(n,\lambda,\mu)}_{\ceil{\frac{n}{P}}p}&<\eta^{(\lambda,\mu)}\left(\frac{2\pi p}{P}\right)+\frac{1}{P}\,\qquad  \forall\,p\in\{1,2,\ldots,P\}  \,,\\
    \eta^{(n,\lambda,\mu)}_{\floor{\frac{n}{P}}p+1}&>
    \eta^{(\lambda,\mu)}\left(\frac{2\pi p}{P}\right)-\frac{1}{P}\,\qquad  \forall\,p\in\{1,2,\ldots,P-1\}  \,,
\ee
for $n$ sufficiently large, where $\ceil{\cdot}$ and $\floor{\cdot}$ denote the ceil and floor function, respectively. It holds that
\bb
    C(\lambda,\mu,\nu)&\eqt{(i)}C\left(   \left\{\bigotimes_{i=1}^n \mathcal{E}_{\eta_i^{(n,\lambda,\mu)},\nu} \right\}_{n\in\mathbb{N}} \right)\\&\leqt{(ii)} C\left(   \left\{\bigotimes_{i=1}^{\ceil{\frac{n}{P}}P} \mathcal{E}_{\eta_i^{(n,\lambda,\mu)},\nu} \right\}_{n\in\mathbb{N}} \right)\\&\leqt{(iii)} C\left(   \left\{\left(\bigotimes_{p=1}^{P} \mathcal{E}_{\eta_{\ceil{\frac{n}{P}}p}^{(n,\lambda,\mu)},\nu}\right)^{\otimes \ceil{\frac{n}{P}}} \right\}_{n\in\mathbb{N}} \right)\\&\leqt{(iv)} C\left(   \left\{\left(\bigotimes_{p=1}^{P} \mathcal{E}_{\min\{1,\eta^{(\lambda,\mu)}(2\pi\frac{p}{P})+\frac{1}{P}\},\nu}\right)^{\otimes \ceil{\frac{n}{P}}} \right\}_{n\in\mathbb{N}} \right)\\&\eqt{(v)} \frac{1}{P} C\left( \bigotimes_{p=1}^{P} \mathcal{E}_{\min\{1,\eta^{(\lambda,\mu)}(2\pi\frac{p}{P})+\frac{1}{P}\},\nu}\right)\,,
\ee
where: in (i) we applied Theorem~\ref{thm_strong_sm}; (ii) is a consequence of the fact that $\bigotimes_{i=1}^n \mathcal{E}_{\eta_i^{(n,\lambda,\mu)},\nu}$
can be simulated via $\bigotimes_{i=1}^{\ceil{\frac{n}{P}}P} \mathcal{E}_{\eta_i^{(n,\lambda,\mu)},\nu}$ by using only the first $n$ channel uses and discarding the remaining $\ceil{\frac{n}{P}}P-n$ ones;
in (iii) we exploited the composition rule in~\eqref{composition_them} and the fact that the transmissivity $\eta_i^{(Pl,\lambda,\mu)}$ is monotonically non-decreasing in $i$; in (iv) we made the assumption that $n$ is sufficiently large, which can be done without loss of generality because the capacities are defined in the limit as the number of channel uses $n$ approaches infinity, and in addition we employed the inequality stated in~\eqref{ineq_proof_sm}; in (v) we exploited that for any $\alpha\in(0,1]$ and any quantum channel $\Phi$ it holds that $C(\{\Phi^{\otimes \ceil{\alpha n}}\}_{n\in\N})=\alpha C(\Phi)$, where $C(\Phi)$ denotes the capacity of the memoryless channel $\Phi$. Analogously, it holds that 
\bb
    C(\lambda,\mu,\nu) 
    &= C\left(   \left\{\bigotimes_{i=1}^n \mathcal{E}_{\eta_i^{(n,\lambda,\mu)},\nu} \right\}_{n\in\mathbb{N}} \right)\\&\geqt{(vi)} C\left(   \left\{\bigotimes_{i=1}^{\floor{\frac{n}{P}}P} \mathcal{E}_{\eta_i^{(n,\lambda,\mu)},\nu} \right\}_{n\in\mathbb{N}} \right)\\&\ge C\left(   \left\{\left(\bigotimes_{p=1}^{P-1} \mathcal{E}_{\eta_{\floor{\frac{n}{P}}p+1}^{(n,\lambda,\mu)},\nu}\right)^{\otimes \floor{\frac{n}{P}}} \right\}_{n\in\mathbb{N}} \right)\\&\ge C\left(   \left\{\left(\bigotimes_{p=1}^{P} \mathcal{E}_{\max\{0,\eta^{(\lambda,\mu)}(2\pi\frac{p}{P})-\frac{1}{P}\},\nu}\right)^{\otimes \floor{\frac{n}{P}}} \right\}_{n\in\mathbb{N}} \right)\\&= \frac{1}{P} C\left( \bigotimes_{p=1}^{P} \mathcal{E}_{\max\{0,\eta^{(\lambda,\mu)}(2\pi\frac{p}{P})-\frac{1}{P}\},\nu}\right)\,,
\ee
where (vi) is a consequence of the fact that 
\bb
\bigotimes_{i=1}^{\floor{\frac{n}{P}}P} \mathcal{E}_{\eta_i^{(n,\lambda,\mu)},\nu} 
\ee 
can be simulated via
\bb
\bigotimes_{i=1}^n \mathcal{E}_{\eta_i^{(n,\lambda,\mu)},\nu}
\ee
by using only the first $\floor{\frac{n}{P}}P$ channel uses and discarding the remaining $n-\floor{\frac{n}{P}}P$ ones. Consequently, the capacity $C(\lambda,\mu,\nu)$ can be bounded as
    \bb\label{bound0}
        &\frac{1}{P}C\left(\bigotimes_{p=1}^{P-1}\mathcal{E}_{\max(0,\eta^{(\lambda,\mu)}\left(2\pi\frac{p}{P}\right)-\frac{1}{P}),\nu}\right)\le C(\lambda,\mu,\nu)\\
        &\le\frac{1}{P}C\left( \bigotimes_{p=1}^P\mathcal{E}_{\min(1,\eta^{(\lambda,\mu)}\left(2\pi\frac{p}{P}\right)+\frac{1}{P}),\nu} \right)\,
    \ee
for all $P\in\N$ with $P\ge 2$. In particular
    \bb\label{bound}
        &\limsup_{P\rightarrow\infty}\frac{1}{P}C\left(\bigotimes_{p=1}^{P-1}\mathcal{E}_{\max(0,\eta^{(\lambda,\mu)}\left(2\pi\frac{p}{P}\right)-\frac{1}{P}),\nu}\right)\\
        &\le C(\lambda,\mu,\nu)\le\liminf_{P\rightarrow\infty}\frac{1}{P}C\left( \bigotimes_{p=1}^P\mathcal{E}_{\min(1,\eta^{(\lambda,\mu)}\left(2\pi\frac{p}{P}\right)+\frac{1}{P}),\nu} \right)\,
    \ee
\textbf{Part 2)} Note that
\begin{align}
    C(\lambda,\mu,\nu)&\geqt{(vii)} \limsup_{P\rightarrow\infty}\frac{1}{P}C\left(\bigotimes_{p=1}^{P-1}\mathcal{E}_{\max(0,\eta^{(\lambda,\mu)}\left(2\pi\frac{p}{P}\right)-\frac{1}{P}),\nu}\right)\\
    &\geqt{(viii)} \limsup_{P\rightarrow\infty}\frac{1}{P} \sum_{p=1}^{P-1} C\left(\mathcal{E}_{\max(0,\eta^{(\lambda,\mu)}\left(2\pi\frac{p}{P}\right)-\frac{1}{P}),\nu}\right)\label{ineq_capacity_term_sm0}\\
    &\eqt{(ix)}\int_{0}^{2\pi}\frac{\mathrm{d}x}{2\pi} C\left(\mathcal{E}_{ \eta^{(\lambda,\mu)}(x)   ,\nu}\right)\,,\label{ineq_capacity_term_sm}
\end{align}
where: (vii) is a consequence of \eqref{bound}; (viii) follows from the fact that Alice and Bob can independently employ the optimal communication strategy for each of the $P$ single-mode channels that define the $P$-mode attenuator $\bigotimes_{p=1}^{P-1}\mathcal{E}_{\max(0,\eta^{(\lambda,\mu)}\left(2\pi\frac{p}{P}\right)-\frac{1}{P}),\nu}$; in (ix) we exploited the fact that the function $\lambda\longmapsto C(\mathcal{E}_{\lambda,0}) $ is continuous in $\lambda\in[0,1)$ and that
\bb
    \lambda^{\frac{1}{1-\mu}}\le\eta^{(\lambda,\mu)}(x)\le \lambda^{\frac{1-\sqrt{\mu}}{1+\sqrt{\mu}}}<1\,,
\ee
as established by~\eqref{upp_eta2pi}. Hence, we have proved~\eqref{eq_low_therm}.

\textbf{Part 3)} Now, let us assume that the thermal noise is zero, i.e.~$\nu=0$.  Under this assumption, we can exploit the additivity of the capacity $C$ of the pure-loss channels, i.e.~$C$ is such that for all transmissivities $\{\lambda_p\}_{p=1,2,\ldots,P}$ and all $P\in\N$ it holds that
\bb
    C\left(\bigotimes_{p=1}^P\mathcal{E}_{\lambda_p,\nu}\right)=\sum_{p=1}^PC\left( \mathcal{E}_{\lambda_p,\nu} \right)\,.
\ee
Indeed, the capacities $Q$~\cite{Wolf2007}, $Q_2$~\cite{PLOB}, and $K$~\cite{PLOB} of the pure-loss channel are additive. Hence,~\eqref{bound} implies that:
\bb\label{eq_proof_integral}
        &\limsup_{P\rightarrow\infty}\frac{1}{P}\sum_{p=1}^{P-1}C\left(\mathcal{E}_{\max(0,\eta^{(\lambda,\mu)}\left(2\pi\frac{p}{P}\right)-\frac{1}{P}),0}\right)  \le C(\lambda,\mu,0) \\
        &\le\liminf_{P\rightarrow\infty}\frac{1}{P}\sum_{p=1}^{P}C\left( \mathcal{E}_{\min(1,\eta^{(\lambda,\mu)}\left(2\pi\frac{p}{P}\right)+\frac{1}{P}),0} \right)\,.
\ee
By exploiting that $\eta^{(\lambda,\mu)}(2\pi)<1$ and the fact that the function $\lambda\longmapsto C(\mathcal{E}_{\lambda,0}) $ is continuous in $\lambda\in[0,1)$, we have that the left-hand side and the right-hand side of~\eqref{eq_proof_integral} converge to the same quantity:
\bb
    &\liminf_{P\rightarrow\infty}\frac{1}{P}\sum_{p=1}^{P}C\left( \mathcal{E}_{\min(1,\eta^{(\lambda,\mu)}\left(2\pi\frac{p}{P}\right)+\frac{1}{P}),0} \right)\\
    &\qquad=\int_{0}^{2\pi}\frac{\mathrm{d}x}{2\pi}C(\mathcal{E}_{\eta^{(\lambda,\mu)}(x),0})
\ee
and
\bb
     &\limsup_{P\rightarrow\infty}\frac{1}{P}\sum_{p=1}^{P-1}C\left(\mathcal{E}_{\max(0,\eta^{(\lambda,\mu)}\left(2\pi\frac{p}{P}\right)-\frac{1}{P}),0}\right)\\
     &\qquad=\int_{0}^{2\pi}\frac{\mathrm{d}x}{2\pi}C(\mathcal{E}_{\eta^{(\lambda,\mu)}(x),0})\,.
\ee
Hence, in the absence of thermal noise, the capacity $C(\lambda,\mu,0)$ of our quantum memory channel is given by
\bb
    C(\lambda,\mu,0)=\int_{0}^{2\pi}\frac{\mathrm{d}x}{2\pi}C(\mathcal{E}_{\eta^{(\lambda,\mu)}(x),0})\,.
\ee
Consequently, by leveraging the expression of the capacities $Q$~\cite{Wolf2007}, $Q_2$~\cite{PLOB}, and $K$~\cite{PLOB} of the pure-loss channel reported in~\eqref{capacities_pure_loss}, we obtain~\eqref{thesis_ris_cap}.

\end{proof}

\section{Conclusions}
In this work we investigated quantum communication across optical fibres in the little-studied case where memory effects are present, and hence where commonly employed approximations break down. We showed that memory effects enhance quantum communication, enabling the transmission of qubits, entanglement, and secret keys even in highly noisy regimes. While our results are model dependent, they do offer some concrete hope that memory effects could contribute to achieving efficient long-distance quantum communication with fewer quantum repeaters than previously believed.

\medskip
\section*{Acknowledgment}
FAM and VG acknowledge financial support by MUR (Ministero dell'Istruzione, dell'Universit\`a e della Ricerca) through the following projects: PNRR MUR project PE0000023-NQSTI, PRIN 2017 Taming complexity via Quantum Strategies: a Hybrid Integrated Photonic approach (QUSHIP) Id. 2017SRN-BRK, and project PRO3 Quantum Pathfinder. GDP has been supported by the HPC Italian National Centre for HPC, Big Data and Quantum Computing - Proposal code CN00000013 and by the Italian Extended Partnership PE01 - FAIR Future Artificial Intelligence Research - Proposal code PE00000013 under the MUR National Recovery and Resilience Plan funded by the European Union - NextGenerationEU.
GDP is a member of the ``Gruppo Nazionale per la Fisica Matematica (GNFM)'' of the ``Istituto Nazionale di Alta Matematica ``Francesco Severi'' (INdAM)''. MF is supported by a Juan de la Cierva Formaci\`on fellowship (project FJC2021-047404-I), with funding from MCIN/AEI/10.13039/501100011033 and European Union NextGenerationEU/PRTR, and by Spanish Agencia Estatal de Investigaci\`on, project PID2019-107609GB-I00/AEI/10.13039/501100011033, by the European Union Regional Development Fund within the ERDF Operational Program of Catalunya (project QuantumCat, ref. 001-P-001644), and by European Space Agency, project ESA/ESTEC 2021-01250-ESA. FAM and LL thank the Freie Universit\"{a}t Berlin for hospitality. FAM, LL, and VG acknowledge valuable discussions with Paolo Villoresi, Giuseppe Vallone, and Marco Avesani and their hospitality at the University of Padua. FAM and MF acknowledge valuable discussions with Giovanni Barbarino regarding the Avram--Parter theorem and Szeg\H{o} theorem.

\begin{IEEEbiographynophoto}{Francesco Anna Mele}
was born in Acquaviva delle Fonti, Italy, in 1997.
He received a B.Sc.\ and a M.Sc.\ degree in Physics from the University of Pisa, Italy, in 2021, a Diploma in Physics at the Scuola Normale Superiore, Pisa, Italy, in 2021. Since November 2021, he is a PhD student at Scuola Normale Superiore, Pisa, Italy. His research interests include all aspects of quantum information and computation.
\end{IEEEbiographynophoto}

\begin{IEEEbiographynophoto}{Giacomo De Palma}
was born in Lanciano, Italy, in 1990. He received the B.S. and M.S. degrees in physics from the University of Pisa in 2011 and 2013, respectively, and the Diploma di Licenza and Ph.D. degrees in physics from Scuola Normale Superiore in 2014 and 2016, respectively. From 2016 to 2018, he held a post-doctoral position at the University of Copenhagen. From 2018 to 2019, he was a Marie-Sklodowska Curie Fellow at the University of Copenhagen. From 2019 to March 2021, he was a Post-Doctoral Associate at MIT. From March to December 2021, he was a Tenure-track Assistant Professor in mathematical physics at the Scuola Normale Superiore. Since December 2021, he is Associate Professor in mathematical physics at the University of Bologna. He is the author of 49 scientific articles. His research interests include all aspects of quantum information and computation.
Prof. De Palma is a member of the International Association of Mathematical Physics (IAMP), of the Italian Mathematical Union (UMI) and of the Istituto Nazionale di Alta Matematica "Francesco Severi" (INdAM). He was a recipient of the Best Italian Researcher in Denmark (BIRD) Award in 2018.
\end{IEEEbiographynophoto}

\begin{IEEEbiographynophoto}{Marco Fanizza}
received the M.S. degree in Physics in 2017 from the University of Pisa, Italy and the Diploma di Licenza in Physics from Scuola Normale Superiore, Pisa, Italy, in 2018. In 2021 he earned his Ph.D. in Nanosciences from Scuola Normale Superiore. Since then, he has been a postdoc at Universitat Autonoma de Barcelona, Spain. He has been supported by a Juan de la Cierva fellowship since 2023.
\end{IEEEbiographynophoto}

\begin{IEEEbiographynophoto}{Vittorio Giovannetti}
was born in Castelnuovo Garfagnana, Italy, in 1970. He graduated with a degree in Theoretical Physics from the University of Pisa in 1997 and earned his Ph.D. in Physics from the University of Camerino in 2001. From 2001 to 2004, he served as a postdoctoral associate at the Research Laboratory of Electronics at the Massachusetts Institute of Technology. Between 2004 and 2010, he held a postdoctoral position at the Scuola Normale Superiore (SNS) in Pisa. From 2010 to 2022, he was an Associate Professor of Physics at SNS, and since 2022, he has been a Full Professor of Physics there.
\end{IEEEbiographynophoto}

\begin{IEEEbiographynophoto}{Ludovico Lami}
received a B.Sc.\ and a M.Sc.\ degree in Physics from the University of Pisa, Italy, in 2014, a Diploma in Physics at the Scuola Normale Superiore, Pisa, Italy, in 2015, and a Ph.D.\ degree from the Autonomous University of Barcelona, Spain, in 2017, where he was later conferred the Special Award for Doctoral Studies. He was an Alexander von Humboldt Research Fellow at the University of Ulm, Germany, from 2020 to 2022, and from 2022 he is an Assistant Professor at the University of Amsterdam. His research interests lie in quantum information, entanglement theory, quantum Shannon theory, and in foundational aspects of quantum physics.
\end{IEEEbiographynophoto}

\bibliographystyle{unsrt}
\bibliography{biblio}
\section{Capacities of the Localised Interaction Model}\label{sec_appendix_capacities_LIM}
In this section, we provide an expression for the capacities of the "Localised Interaction Model" (LIM) of references~\cite{Memory1, Memory2, Memory3}.

Let $Q^{\text{LIM}}(\lambda,\mu,\nu)$, $Q_2^{\text{LIM}}(\lambda,\mu,\nu)$, and $K^{\text{LIM}}(\lambda,\mu,\nu)$ be the quantum capacity, two-way quantum capacity, and secret-key capacity of the LIM with transmissivity $\lambda$, memory parameter $\mu$, and thermal noise $\nu$. The effective transmissivity function of the LIM is given by~\cite[Eq.~(5)]{Memory1}:
\bb
    \eta_{\text{LIM}}^{(\lambda,\mu)}(x)= \frac{\mu+\lambda-2\sqrt{\mu\lambda}\cos\left(\frac{x}{2}\right)}{1+\mu\lambda-2\sqrt{\mu\lambda}\cos\left(\frac{x}{2}\right)}\qquad\forall\, x\in[0,2\pi]\,.
\ee
Note that $\eta_{\text{LIM}}^{(\lambda,\mu)}(x)$ is monotonically increasing in $x$ and thus achieves is maximum in $x=2\pi$:
\bb
    \max_{x\in[0,2\pi]}\eta_{\text{LIM}}^{(\lambda,\mu)}(x)=\eta_{\text{LIM}}^{(\lambda,\mu)}(x=2\pi)=\left(\frac{\sqrt{\mu}+\sqrt{\lambda}}{1+\sqrt{\mu\lambda}}\right)^2\,.
\ee
We can apply the same reasoning used in Theorem~\ref{pos_cap_thm_sm} in order to determine the parameter region where the capacities are strictly positive. Specifically, in the absence of thermal noise ($\nu=0$), the quantum capacity $Q^{\text{LIM}}(\lambda,\mu,\nu=0)$ is strictly positive if and only if $\eta_{\text{LIM}}^{(\lambda,\mu)}(x=2\pi)>\frac{1}{2}$. This condition can be expressed as follows:
\bb
    Q^{\text{LIM}}(\lambda,\mu,\nu=0)>0 \Longleftrightarrow \sqrt{\mu}>\frac{1-\sqrt{2\lambda}}{\sqrt{2}-\sqrt{\lambda}}\,.
\ee
Additionally, the two-way quantum capacity $Q_2^{\text{LIM}}(\lambda,\mu,\nu)$ and secret-key capacity $K^{\text{LIM}}(\lambda,\mu,\nu)$ are strictly positive if and only if $\eta_{\text{LIM}}^{(\lambda,\mu)}(x=2\pi)>\frac{\nu}{\nu+1}$. This condition can be expressed as follows:
\bb
    &Q^{\text{LIM}}(\lambda,\mu,\nu),K^{\text{LIM}}(\lambda,\mu,\nu)>0 \\&\qquad\Longleftrightarrow \sqrt{\mu}>\frac{ \sqrt{\nu}-\sqrt{\lambda(\nu+1)}  }{ \sqrt{\nu+1}-\sqrt{\lambda\nu} } \,.
\ee
The quantum capacity $Q^{\text{LIM}}(\lambda,\mu,0)$ of the LIM in the absence of thermal noise ($\nu=0$) is given by~\cite{Memory1}:
\bb
    Q^{\text{LIM}}(\lambda,\mu,0)=\int_{0}^{2\pi}\frac{\mathrm{d}x}{2\pi}\max\left(0,\log_2\left(\frac{\eta_{\text{LIM}}^{(\lambda,\mu)}(x)}{1-\eta_{\text{LIM}}^{(\lambda,\mu)}(x)}\right)\right)\,.
\ee
One can numerically verify that $Q^{\text{LIM}}(\lambda,\mu,\nu=0)$ is an increasing function of $\mu$ for each assigned transmissivity value $\lambda$, implying that memory effects can improve quantum communication.

Furthermore, by exploiting the same reasoning used in Theorem~\ref{Cap_absence_thermal_sm}, we can directly derive an expression for the two-way quantum and the secret-key capacity of the LIM in the absence of thermal noise. Namely, they are given by:
\bb
    K^{\text{LIM}}(\lambda,\mu,0)&=Q^{\text{LIM}}_2(\lambda,\mu,0)\\
    &=\int_{0}^{2\pi}\frac{\mathrm{d}x}{2\pi}\log_2\left(\frac{1}{1-\eta_{\text{LIM}}^{(\lambda,\mu)}(x)}\right)\,.
\ee

\section{Two-way quantum communication and memory effects}
The reader might question the meaningfulness of considering communication tasks assisted by two-way classical communication in a setting where the memory parameter $\mu$, 
and thus the time interval between subsequent signals, is fixed. Indeed, in general, the rounds of classical communication between consecutive channel uses may vary during the communication protocol. 
Additionally, another problem is that, when the sender and receiver are very far apart, even a single round of classical communication requires an excessively long waiting time which prevents the exploitation of memory effects altogether. However, in the case of Choi-simulable channels~\cite{PLOB}, 
it is meaningful to consider such communication tasks. This is because optimal entanglement and secret-key distribution strategies for Choi-simulable channels can be achieved by initially utilising the channel (with a constant and short time interval between subsequent input signals) to share multiple copies of the Choi state, followed by employing optimal distillation protocols (either for entanglement or secret-key) to distil the Choi state. 
Fortunately, our model is associated with the quantum memory channel $\{\Phi_{\lambda,\mu,\nu}^{(n)}\}_{n\in\N}$, which is Choi simulable because it is Gaussian~\cite{PLOB}. 

\section{Transversal attenuation}\label{section_new_model_SM}

 A possible limitation of the DIM introduced before is the fact that the environments are taken to be single-mode. This would not be an issue in absence of memory effects, since one could always select a single-mode component of the environment as the effective state interacting with the input. Instead, memory effects allow for more complex dynamics. One could easily extend our model to capture a memoryless contribution to the noise, by considering the case where between each interaction with $E_1^{(j)}$ and $E_1^{(j+1)}$ there is an additional thermal loss of attenuation $\gamma^{1/M}$ and environment state $\tau_{\nu}$, but without memory effects, as depicted in Fig.~\ref{fig_new_for_Marco}. In this way, we are allowing for the presence of an additional mechanism of interaction with the environment which happens on a faster timescale. By taking $\gamma^{1/M}$ to be uniform along the several uses of the transmission line, the action on the annihilation operators of the input modes turns out to be a simple rescaling by $\sqrt{\gamma^{1/M}}$, which commutes with the original interaction. The net effect is to place identical thermal attenuators at each mode, before or after the memory channel.

\begin{figure*}[t]
	\centering
	\includegraphics[width=1\linewidth]{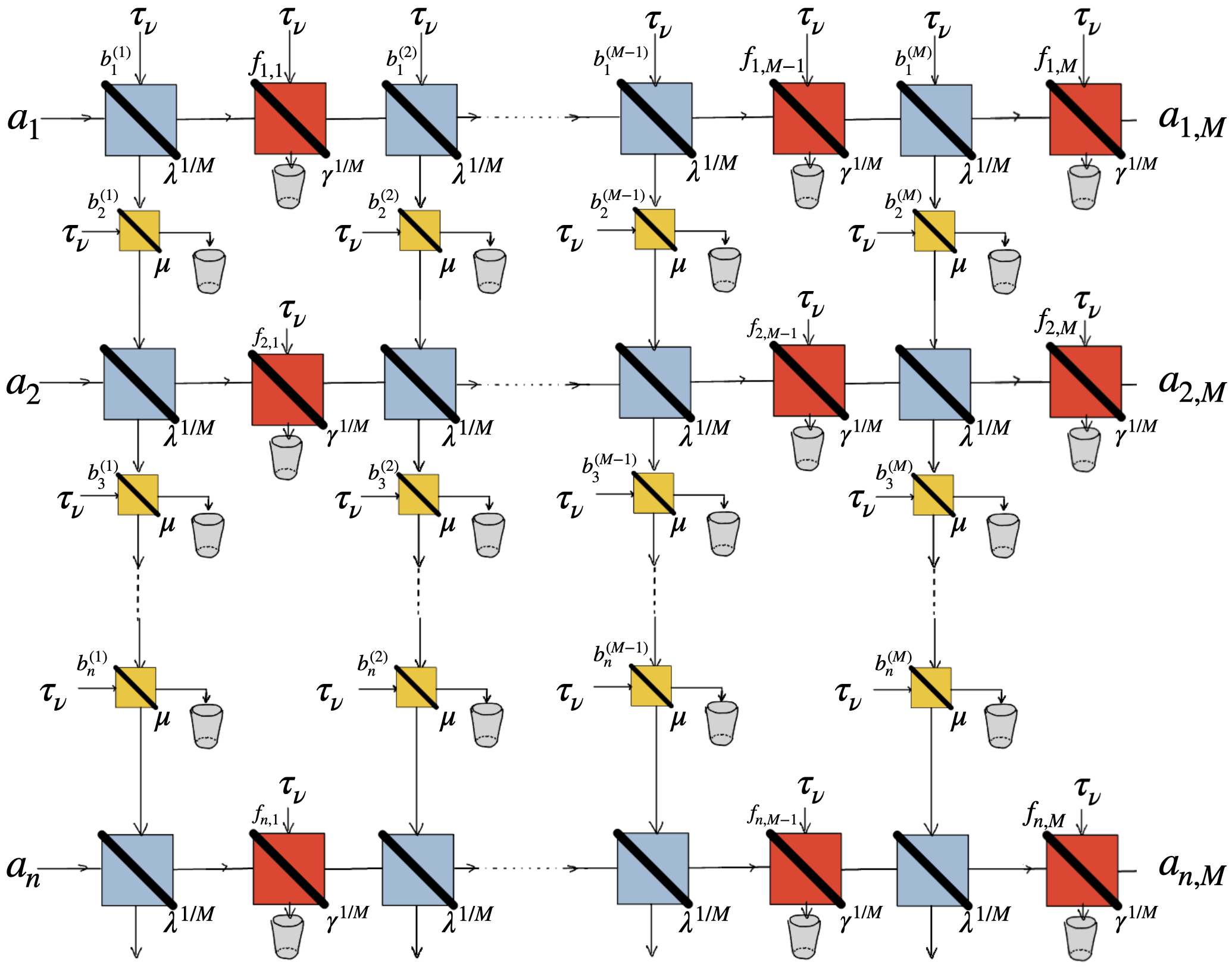}
	\caption{Depiction of the quantum memory channel which models both memory and memoryless noise. In contrast to the model presented in Figure~\ref{fig:griglia_SM}, here we have also the red beam splitters of transmissivity $\gamma^{1/M}$ which models the presence of memoryless contribution to the noise. This addition, while it does not substantially affect the mathematical structure of the model, captures the possibility of having different kinds of noise effects on the fibre at the same time, such as scattering, which could enable memory effects, and absorption, which may be memoryless. }
    \label{fig_new_for_Marco}
\end{figure*}
This simple model captures the possibility of having the communication line interacting with a multi-mode system, where effectively there are two independent noise effects acting on different timescales. For example, one may think of the memory effects to be due to scattering and the memoryless noise to be due to absorption. Since at the end we will see that the memory channel can be written as a tensor product of attenuators, by the composition rules for attenuators, the capacities of this enhanced model are also completely computable. In particular, it follows that the advantage on the quantum capacity might remain, if $\gamma$ is sufficiently large. If $\nu=0$, the advantage due to memory effects on one-way qubit transmission  remains if and only if $\gamma>\frac{1}{2}$, while the two-way quantum capacity can be non-zero as soon as $\gamma\geq \frac{\nu}{\nu+1}$.

In concrete, we modify the model by introducing new modes $\mathbf{f}$ as depicted in Fig.~\ref{fig_new_for_Marco}, i.e.
\bb
\mathbf{a}&\coloneqq(a_1,\,a_2,\ldots,\,a_n)^T\,,\\ \mathbf{a_M}&\coloneqq(a_{1,M},\,a_{2,M},\ldots,\,a_{n,M})^T\,,\\
\mathbf{b}&\coloneqq (b_{1}^{(1)},\,b_{1}^{(2)},\ldots,\, b_{1}^{(M)},\, b_{2}^{(1)},\\
&\qquad b_{2}^{(2)},\ldots,\,b_{2}^{(M)},\ldots,\,b_{n}^{(1)},\,b_{n}^{(2)},\ldots,\,b_{n}^{(M)})^T\,,\\
\mathbf{f}&\coloneqq (f_{1,1},f_{1,2},\ldots,\,f_{1,M-1}, f_{1,M},\ldots,\\
&\qquad f_{n-1,1},\,f_{n,2},\ldots,\,f_{n,M-1}, f_{n,M})^T\,,
\ee
one can modify~\eqref{vector_annihilation} as
\bb\label{vector_annihilation2}
\mathbf{a_M}=A^{(M,n,\lambda,\mu,\gamma)}\mathbf{a}+E^{(M,n,\lambda,\mu,\gamma)}\mathbf{b}+F^{(M,n,\lambda,\mu,\gamma)}\mathbf{f}\,,
\ee
where $A^{(M,n,\lambda,\mu,\gamma)}$ is an $n\times n$ real matrix, $E^{(M,n,\lambda,\mu,\gamma)}$  is an $n\times nM$ real matrix, and $F^{(M,n,\lambda,\mu,\gamma)}$ is an $n\times nM$ real matrix which satisfy
\bb\label{implication_comm2}
    &A^{(M,n,\lambda,\mu,\gamma)}{A^{(M,n,\lambda,\mu,\gamma)}}^T+E^{(M,n,\lambda,\mu,\gamma)}{E^{(M,n,\lambda,\mu,\gamma)}}^T\\
    &\qquad+F^{(M,n,\lambda,\mu,\gamma)}{F^{(M,n,\lambda,\mu,\gamma)}}^T=\mathbb{1}_{n\times n}\,,
\ee
thanks to~\eqref{comm_rel}. By following the same steps of the proof of Theorem~\ref{thm_strong_sm}, we get again
\bb
&\chi_{\Phi_{\lambda,\mu,\nu,\gamma}^{(M,n)}(\rho^{(n)})}(\mathbf{z})\coloneqq \Tr\left[\Phi_{\lambda,\mu,\nu,\gamma}^{(M,n)}(\rho^{(n)})\,e^{\mathbf{a}^\dagger\mathbf{z}-\mathbf{z}^\dagger\mathbf{a} }\right]\\
&=\chi_{\rho^{(n)}}\left({A^{(M,n,\lambda,\mu,\gamma)}}^T\mathbf{z}\right)\\
&\quad\times\chi_{\tau_\nu^{\otimes nM}}\left({E^{(M,n,\lambda,\mu,\gamma)}}^T\mathbf{z}\right)\chi_{\tau_\nu^{\otimes nM}}\left({F^{(M,n,\lambda,\mu,\gamma)}}^T\mathbf{z}\right)\\
&=\chi_{\rho^{(n)}}\left({A^{(M,n,\lambda,\mu,\gamma)}}^T\mathbf{z}\right)\\
&\quad\times e^{-(\nu+\frac{1}{2})\mathbf{z}^\dagger ( \mathbb{1}_{n\times n}- A^{(M,n,\lambda,\mu,\gamma)}{A^{(M,n,\lambda,\mu,\gamma)}}^T ) \mathbf{z}}\qquad\forall\,\mathbf{z}\in\mathbb{C}^{n}\,,
\ee
And the analysis continues as in the proof of Theorem~\ref{thm_strong_sm}. We are left to specify $A^{(M,n,\lambda,\mu,\gamma)}$ in the modified model. By exploiting~\eqref{transf_beam} and the notation in Fig.~\ref{fig_new_for_Marco}, we derive the following passive transformation: 
\bb
a_{i,j}&=\sqrt{\gamma^{1/M}}\sqrt{\lambda^{1/M}}\,a_{i,j-1}+\sqrt{\gamma^{1/M}}\sqrt{1-\lambda^{1/M}}\,m_{i-1,j}\\
&\quad+\sqrt{1-\gamma^{1/M}}\,f_{i,j-1}\,,
\ee
\bb
m_{i,j}&=-\sqrt{1-\mu}\,b_{i+1}^{(j)}+\sqrt{\mu\lambda^{1/M}}\,m_{i-1,j}\\
&\quad-\sqrt{\mu(1-\lambda^{1/M})}\,a_{i,j-1}\,,
\ee
where the transmissivity $\gamma$ is associated to memoryless noise. By defining $\tilde{a}_{i,j}\coloneqq\gamma^{-j/M}a_{i,j}$, and $\tilde{m}_{i,j}\coloneqq\gamma^{-(j-1)/M}m_{i,j}$  we obtain
\bb\label{eq_a_m2}
\tilde{a}_{i,j}&\simeq\sqrt{\lambda^{1/M}}\,\tilde{a}_{i,j-1}+\sqrt{1-\lambda^{1/M}}\,\tilde{m}_{i-1,j}\,,\\
\tilde{m}_{i,j}&\simeq\sqrt{\mu\lambda^{1/M}}\,\tilde{m}_{i-1,j}-\sqrt{\mu(1-\lambda^{1/M})}\,\tilde{a}_{i,j-1}\,,
\ee
which are the recurrence relations for the model at $\gamma=1$.
We immediately get 
\bb
A^{(M,n,\lambda,\mu,\gamma)}=\sqrt{\gamma} A^{(M,n,\lambda,\mu,1)}=\sqrt{\gamma} A^{(M,n,\lambda,\mu)}\,
\ee
and therefore
\bb\label{lim_elements_def2}
    \bar{A}_{i,h}^{(n,\lambda,\mu,\gamma)}\coloneqq \lim\limits_{M\rightarrow\infty}A_{i,h}^{(M,n,\lambda,\mu,\gamma)}=\sqrt{\gamma}  \bar{A}_{i,h}^{(n,\lambda,\mu)}\,.
\ee
This means that this modification of the model is equivalent to placing a thermal attenuator with transmissivity $\gamma$ before or after the memory channel. Mathematically, the square of the absolute values of the singular values of $\bar{A}_{i,h}^{(n,\lambda,\mu,\gamma)}$ are obtained by rescaling those of  $\bar{A}_{i,h}^{(n,\lambda,\mu)}$ by $\gamma$. It follows that as long as $\gamma>1/2$ and for $\nu=0$, there are values of $\mu$ and $\lambda$ such that the quantum capacity is still non-zero. In particular $\gamma\lambda$, which can be thought as the effective attenuation on the fibre, could be less than $1/2$ while the quantum capacities are still non zero. A similar observation holds for the two-way capacities. If there were no memory effects, the two-way quantum capacity would be zero if and only if $\gamma\lambda\leq\frac{\nu}{\nu+1}$, while in our model it can be still non-zero even if the inequality is satisfied.

\end{document}